\documentclass[pdflatex,sn-mathphys-ay]{sn-jnl}

\usepackage{graphicx}%
\usepackage{multirow}%
\usepackage{amsmath,amssymb,amsfonts}%
\usepackage{amsthm}%
\usepackage{mathrsfs}%
\usepackage[title]{appendix}%
\usepackage{xcolor}%
\usepackage{textcomp}%
\usepackage{manyfoot}%
\usepackage{booktabs}%
\usepackage{algorithm}%
\usepackage{algorithmicx}%
\usepackage{algpseudocode}%
\usepackage{listings}%
\usepackage{lineno,hyperref}
\usepackage{caption}
\usepackage{mathtools}
\usepackage{changes}
\usepackage{multirow}
\usepackage{verbatim}
\usepackage{subcaption}
\usepackage{bm}

\theoremstyle{thmstyleone}%
\newtheorem{theorem}{Theorem}
\newtheorem{proposition}[theorem]{Proposition}%

\theoremstyle{thmstyletwo}%
\newtheorem{remark}{Remark}%

\theoremstyle{thmstylethree}%

\begin{document}

\title[Data-driven Implementations of Various Generalizations of Balanced Truncation]{Data-driven Implementations of Various Generalizations of Balanced Truncation}


\author[1]{\fnm{Umair} \sur{Zulfiqar}}\email{umair@yangtzeu.edu.cn}

\author*[2]{\fnm{Qiu-Yan} \sur{Song}}\email{qysong@hbnu.edu.cn}
\equalcont{These authors contributed equally to this work.}

\author[3]{\fnm{Zhi-Hua} \sur{Xiao}}\email{zhxiao@yangtzeu.edu.cn}
\equalcont{These authors contributed equally to this work.}

\author[4]{\fnm{Victor} \sur{Sreeram}}\email{victor.sreeram@uwa.edu.au}
\equalcont{These authors contributed equally to this work.}

\affil[1]{\orgdiv{School of Electronic Information and Electrical Engineering}, \orgname{Yangtze University}, \orgaddress{\city{Jingzhou}, \postcode{434023}, \state{Hubei}, \country{China}}}

\affil*[2]{\orgdiv{Huangshi Key Laboratory of Metaverse and Virtual Simulation, School of Mathematics and Statistics}, \orgname{Hubei Normal University}, \orgaddress{\city{Huangshi}, \postcode{435002}, \state{Hubei}, \country{China}}}

\affil[3]{\orgdiv{School of Mathematics and Systems Science}, \orgname{Wuhan University of Science and Technology}, \orgaddress{\city{Wuhan}, \postcode{430065}, \state{Hubei}, \country{China}}}

\affil[4]{\orgdiv{Department of Electrical, Electronic, and Computer Engineering}, \orgname{The University of Western Australia}, \orgaddress{\city{Perth}, \postcode{6009}, \state{Western Australia}, \country{Australia}}}

\abstract{Quadrature-based approximation of Gramians in standard balanced truncation yields a non-intrusive, data-driven implementation that requires only transfer function samples on the imaginary axis, which can be measured experimentally. This idea has recently been extended to several generalizations of balanced truncation, including positive-real balanced truncation, bounded-real balanced truncation, and balanced stochastic truncation. However, these extensions require samples of some spectral factorizations on the imaginary axis, and no practical method exists to obtain such data experimentally. As a result, these non-intrusive implementations are mainly of theoretical interest at present.

This paper shows that if the Gramians in these generalizations are approximated via rational interpolation rather than numerical integration, the resulting non-intrusive implementations do not require spectral factorization samples. Instead, they rely only on transfer function samples. Based on this idea, non-intrusive implementations are first developed for several variants of balanced truncation, wherein the Gramians are approximated implicitly using low-rank Alternating Direction Implicit (ADI) methods for Lyapunov and Riccati equations. These formulations include Linear Quadratic Gaussian (LQG) balanced truncation, \(\mathcal{H}_\infty\) balanced truncation, positive-real balanced truncation, bounded-real balanced truncation, self-weighted balanced truncation, and balanced stochastic truncation. They require transfer function samples in the right half of the \(s\)-plane, which cannot be measured experimentally.

Next, building on these results, novel data-driven non-intrusive implementations are proposed that require only transfer function samples on the imaginary axis. Hence, unlike the quadrature-based and ADI-based approaches, these non-intrusive formulations can be implemented using practically measurable data. Numerical results are presented for benchmark problems in model order reduction, which show that the proposed non-intrusive implementations achieve accuracy comparable to their intrusive counterparts.}

\keywords{ADI, Balanced truncation, Data-driven, Low-rank, Non-intrusive}



\maketitle

\section{Introduction}\label{sec1}
Model order reduction (MOR) addresses the problem of constructing a low-dimensional dynamical system that approximates an original high-dimensional system. The aim is to create a reduced-order model (ROM) whose input-output response approximates that of the full-order system within an acceptable error tolerance. This provides a systematic approach for generating computationally efficient surrogate systems to accelerate simulation, analysis, and control design for large-scale models. For comprehensive reviews of MOR techniques, see \citep{schilders2008model,quarteroni2014reduced,obinata2012model,benner2021model}.

Balanced truncation (BT) is a classical and effective MOR technique for linear time-invariant (LTI) systems \citep{moore1981principal}. It constructs a ROM by truncating states associated with small Hankel singular values, which correspond to weakly controllable and observable dynamics. The resulting ROM preserves the stability of the full-order model (FOM) and provides an \textit{a priori} error bound.

The primary computational cost of BT stems from solving Lyapunov equations to compute the controllability and observability Gramians. Efficient low-rank solvers for this task are reviewed in \citep{benner2013numerical,simoncini2016computational}. These methods require explicit access to the system’s state-space matrices and are therefore classified as intrusive. In contrast, non-intrusive approaches construct ROMs solely from input–output data, such as samples of the transfer function or impulse response \citep{gustavsen2002rational,mayo2007framework,nakatsukasa2018aaa,gosea2022data,cherifi2022greedy,pradovera2023toward,scarciotti2024interconnection,goyal2024rank,aumann2025practical}. One such method is Quadrature-based BT (QuadBT) \citep{goseaQuad}, which constructs the ROM using numerical quadrature applied to transfer function samples along the imaginary axis ($j\omega$) or to sampled impulse response and their derivatives. QuadBT has been extended to several classes of dynamical systems, such as second-order systems \citep{reiter2025data,wang2025data} and quadratic output systems \citep{padhi2025data}. Other non-intrusive approximate balancing approaches include \citep{liljegren2024data,willcox2002balanced,phillips2004poor,opmeer2011model,gosea2024non,burohman2023data}.

Over the past two decades, the low-rank Alternating Direction Implicit (ADI) method has emerged as a numerically efficient and effective solver for the Lyapunov equations in BT \citep{benner2013efficient}. Recently, a non-intrusive implementation of ADI-based BT was proposed \citep{zulfiqar2025non}. This approach constructs the ROM directly from transfer function samples taken in the right half of the $s$-plane.

BT has been extended to various system classes, including descriptor systems \citep{mehrmann2005balanced}, parametric systems \citep{wittmuess2016parametric}, nonlinear systems \citep{lall2002subspace}, and second-order systems \citep{reis2008balanced}. These extensions preserve various system properties such as passivity \citep{phillips2002guaranteed}, contractivity \citep{opdenacker2002contraction}, and the minimum-phase condition \citep{green1988balanced}. Recently, QuadBT has been extended to include certain generalizations of BT \citep{reiter2023generalizations}, such as balanced stochastic truncation (BST) \citep{green1988balanced}, positive-real BT (PRBT) \citep{green1988balanced}, and bounded-real BT (BRBT) \citep{opdenacker2002contraction,glover1988state}. Although this generalization \citep{reiter2023generalizations} is formulated in a non-intrusive setting, it depends on samples of spectral factorizations of the transfer function. Since practical methods for obtaining such samples in real-world settings are currently unavailable, the method remains largely theoretical.

In this work, we first introduce a low-rank ADI-based non-intrusive framework applicable to several BT generalizations, including Linear Quadratic Gaussian BT (LQGBT) \citep{jonckheere1983new}, \(\mathcal{H}_\infty\)BT \citep{mustafa1991controller}, PRBT \citep{green1988balanced}, BRBT \citep{opdenacker2002contraction}, self-weighted BT (SWBT) \citep{zhou1995frequency}, and BST \citep{green1988balanced}. In contrast to the quadrature-based approach in \citep{reiter2023generalizations}, the ADI-based non-intrusive implementations rely on transfer function samples in the right half of the \(s\)-plane. Although such samples cannot be measured in an experimental setting, the results developed for the ADI-based implementations are later generalized to obtain projection-based non-intrusive implementations that rely on transfer function samples on the imaginary axis, which can be measured experimentally. This leads to data-driven implementations of approximate BT, LQGBT, \(\mathcal{H}_\infty\)BT, PRBT, BRBT, SWBT, and BST. Numerical results show that both non-intrusive formulations perform comparably to their intrusive counterparts.

The paper is structured as follows. Section \ref{sec2} introduces the MOR problem and provides a brief overview of key MOR techniques relevant to this work. Sections \ref{sec3} and \ref{sec4} present the main contributions. In Section \ref{sec3}, we propose non-intrusive low-rank ADI-based implementations of LQGBT, \(\mathcal{H}_\infty\) BT, PRBT, BRBT, SWBT, and BST. Section \ref{sec4} presents projection-based data-driven counterparts of the algorithms in Section \ref{sec3}. The numerical performance of the proposed non-intrusive algorithms is tested in Section \ref{sec5}. Finally, Section \ref{sec6} offers concluding remarks.
\section{Preliminaries}\label{sec2}
Consider a stable LTI system \(G(s)\) of order $n$, given in minimal state-space form by
\[
G(s)=C(sI-A)^{-1}B+D,
\]
where \(A\in\mathbb{R}^{n\times n}\), \(B\in\mathbb{R}^{n\times m}\), \(C\in\mathbb{R}^{p\times n}\), and \(D\in\mathbb{R}^{p\times m}\).

Let \(\hat G(s)\) denote a stable ROM of order \(r\ll n\) that approximates \(G(s)\), with minimal realization
\[
\hat G(s)=\hat C(sI-\hat A)^{-1}\hat B+D,
\]
where \(\hat A\in\mathbb{C}^{r\times r}\), \(\hat B\in\mathbb{C}^{r\times m}\), and \(\hat C\in\mathbb{C}^{p\times r}\).

The ROM \(\hat G(s)\) is obtained via a Petrov–Galerkin projection of the FOM \(G(s)\) as:
\[
\hat A=\hat W^{*}A\hat V,\qquad
\hat B=\hat W^{*}B,\qquad
\hat C=C\hat V,
\]
with full-column-rank matrices \(\hat V,\hat W\in\mathbb{C}^{n\times r}\) satisfying \(\hat W^{*}\hat V=I\). The resulting ROM \(\hat G(s)\) is invariant under the transformations \(\hat V\rightarrow \hat V T_v\) and \(\hat W\rightarrow \hat W T_w\), where \(T_v,T_w\in\mathbb{C}^{r\times r}\) are nonsingular. This degree of freedom can be exploited to apply suitable similarity transformations that convert the generally complex matrices \( \hat{V}, \hat{W} \) and the reduced-state-space triplet \( (\hat{A}, \hat{B}, \hat{C}) \) into equivalent real-valued representations \citep{gallivan2004sylvester}.
\subsection{Krylov subspace-based Rational Interpolation \citep{beattie2017model}}
Let \((\sigma_1, \dots, \sigma_v)\) denote a set of right interpolation points and \((\mu_1, \dots, \mu_w)\) a set of left interpolation points. In the rational interpolation framework, the projection matrices \(\hat{V} \in \mathbb{C}^{n \times vm}\) and \(\hat{W} \in \mathbb{C}^{n \times wp}\) are defined as
\begin{align}
\hat{V} &= \begin{bmatrix} (\sigma_1 I - A)^{-1} B & \cdots & (\sigma_v I - A)^{-1} B\end{bmatrix},\label{Kry_V}\\
\hat{W} &= \begin{bmatrix} (\mu_1^* I - A^T)^{-1} C^T & \cdots & (\mu_w^*I - A^T)^{-1} C^T\end{bmatrix}.\label{Kry_W}
\end{align}
To satisfy the bi-orthogonality condition \(\hat W^{*}\hat V=I\), the left basis is modified as \(\hat W\leftarrow \hat W T_w\), with \(T_w=(\hat V^{*}\hat W)^{-1}\). The resulting ROM interpolates the FOM at the prescribed interpolation points:
\begin{align}
G(\sigma_i)= \hat{G}(\sigma_i), \quad G(\mu_i) = \hat{G}(\mu_i).\label{int_cond_1}
\end{align}
Moreover, at any interpolation point that appears in both sets—i.e., when \(\sigma_i = \mu_j\)—the ROM also matches the first derivative of the FOM at that point, satisfying the Hermite interpolation condition:
\begin{align}
G^\prime(\sigma_i) = \hat{G}^\prime(\sigma_i).\label{int_cond_2}
\end{align}
\subsection{The Loewner framework \citep{mayo2007framework}}
The feedthrough matrix \( D \) can be determined by sampling the transfer function \( G(s) \) at a sufficiently large frequency, as it holds that \( D = \lim_{s \to \infty} G(s) \). Define the strictly proper part of the FOM as  
\begin{align}
H(s)=C(sI-A)^{-1}B.\nonumber
\end{align}
Given samples of \( G(s) \) at selected points \( s_i \), the corresponding samples of \( H(s_i) \) are obtained via  
\begin{align}
H(s_i)=G(s_i)-G(\infty).\nonumber
\end{align}
In the Loewner framework, the following matrices are assembled using samples of \( H(s) \) at the right interpolation points \( \{\sigma_j\}_{j=1}^v \) and left interpolation points \( \{\mu_i\}_{i=1}^w \): 
\begin{align}
\hat{W}^*\hat{V}&=\begin{bmatrix}-\frac{H(\sigma_1) - H(\mu_1)}{\sigma_1 - \mu_1}&\cdots&-\frac{H(\sigma_v) - H(\mu_1)}{\sigma_v - \mu_1}\\\vdots&\ddots&\vdots\\-\frac{H(\sigma_1) - H(\mu_w)}{\sigma_1 - \mu_w}&\cdots&-\frac{H(\sigma_v) - H(\mu_w)}{\sigma_v - \mu_w}\end{bmatrix},\nonumber\\
\hat{W}^*A\hat{V}&=\begin{bmatrix}-\frac{\sigma_1 H(\sigma_1) - \mu_1 H(\mu_1)}{\sigma_1 - \mu_1}&\cdots&-\frac{\sigma_v H(\sigma_v) - \mu_1 H(\mu_1)}{\sigma_v - \mu_1}\\\vdots&\ddots&\vdots\\-\frac{\sigma_1 H(\sigma_1) - \mu_w H(\mu_w)}{\sigma_1 - \mu_w}&\cdots&-\frac{\sigma_v H(\sigma_v) - \mu_w H(\mu_w)}{\sigma_v - \mu_w}\end{bmatrix},\nonumber\\
\hat{W}^*B&=\begin{bmatrix}H(\mu_1)\\\vdots\\H(\mu_w)\end{bmatrix},\nonumber\\
 C\hat{V}&=\begin{bmatrix}H(\sigma_1)&\cdots H(\sigma_v)\end{bmatrix},\label{LF}
\end{align}where \( \hat{V} \) and \( \hat{W} \) are the projection matrices defined in (\ref{Kry_V}) and (\ref{Kry_W}), respectively.

When a right and a left interpolation point are nearly identical (i.e., \( \sigma_j \approx \mu_i \)), the divided differences approach derivatives:  
\begin{align}
 \frac{H(\sigma_j) - H(\mu_i)}{\sigma_j - \mu_i}&\approx H^\prime(\sigma_j),\nonumber\\
  \frac{\sigma_j H(\sigma_j) - \mu_i H(\mu_i)}{\sigma_j - \mu_i} &\approx H(\sigma_j) + \sigma_j H^\prime(\sigma_j).\nonumber
\end{align}Consequently, when the interpolation sets share common points, constructing \( \hat{W}^*\hat{V} \) and \( \hat{W}^*A\hat{V} \) also requires samples of the derivative \( H'(\sigma_j) \). The matrices \( \hat{W}^*\hat{V} \) and \( \hat{W}^*A\hat{V} \) have a special structure and are referred to as the Loewner and shifted Loewner matrices, respectively.

The matrix \(\hat{W}^*\hat{V}\) can be singular, in which case the standard reduced-order matrices  
\(\hat{A} = (\hat{W}^*\hat{V})^{-1}\hat{W}^*A\hat{V}\) and \(\hat{B} = (\hat{W}^*\hat{V})^{-1}\hat{W}^*B\) are not well defined. In such cases, the ROM is constructed in descriptor form as  
\begin{align}
\hat{G}(s)=C\hat{V}(s\hat{W}^*\hat{V}-\hat{W}^*A\hat{V})^{-1}\hat{W}^*B+G(\infty).\nonumber
\end{align}
\subsection{Pseudo-optimal Rational Krylov Algorithm (PORK) \citep{wolf2014h}}
Define the matrices $S_v$, $S_w$, $L_v$, and $L_w$ as:
\begin{align}
S_v &= \text{diag}(\sigma_1, \dots, \sigma_v)\otimes I_m,& S_w &= \text{diag}(\mu_1, \dots, \mu_w)\otimes I_p,\nonumber\\
L_v &= \begin{bmatrix} 1, \dots, 1 \end{bmatrix}\otimes I_m,& L_w^T &= \begin{bmatrix} 1, \dots, 1 \end{bmatrix}\otimes I_p.\label{SbLbScLc}
\end{align}
Then the projection matrices \( \hat{V} \) and \( \hat{W} \), defined in (\ref{Kry_V}) and (\ref{Kry_W}), satisfy the following Sylvester equations:
\begin{align}
A\hat{V}-\hat{V}S_v+BL_v&=0,\label{sylv_V}\\
A^T\hat{W}-\hat{W}S_w^*+C^TL_w^T&=0.\label{sylv_W}
\end{align}Pre-multiplying (\ref{sylv_V}) by \( \hat{W}^* \) and enforcing the bi-orthogonality condition \( \hat{W}^*\hat{V} = I \) yields \( \hat{A} = S_v - \hat{B} L_v \). This expresses the reduced-order state matrix \(\hat{A}\) in terms of the free parameter \(\hat{B} = \zeta \), which preserves the interpolation conditions enforced by \( \hat{V} \). Therefore, varying \(\zeta\) is equivalent to modifying \( \hat{W} \).

Assume that all interpolation points \( \sigma_1, \dots, \sigma_v \) are located in the right half of the \(s\)-plane and that the pair \( (S_v, L_v) \) satisfies the Lyapunov equation: 
\begin{align}
-S_v^* Q_v - Q_v S_v + L_v^T L_v = 0.\label{pork_qs}
\end{align}
If the pair \( (S_v, L_v) \) is observable, setting the free parameter \( \hat{B} = Q_v^{-1} L_v^T \) yields $\hat{A} = S_v - \hat{B} L_v = -Q_v^{-1} S_v^* Q_v$. The resulting ROM
\begin{align}
\hat{A}_r = -Q_v^{-1} S_v^* Q_v,\quad \hat{B} = Q_v^{-1} L_v^T, \quad \hat{C} = C\hat{V},\nonumber
\end{align}
satisfies the $\mathcal{H}_2$-optimality condition $\frac{\partial}{\partial\hat{C}}\big(||G(s)-\hat{G}(s)||_{\mathcal{H}_2}^2\big)=0$. Throughout this paper, this method is referred to as Input PORK (I-PORK).

Similarly, pre-multiplying (\ref{sylv_W}) by \(\hat{V}^*\) reveals that \( \hat{A} \) can be written as \( \hat{A} = S_w - L_w \hat{C} \). This expresses \( \hat{A} \) in terms of the free parameter \( \hat{C}=\zeta \), which is equivalent to varying \( \hat{V} \); the interpolation conditions induced by \( \hat{W} \) are not affected.

Now, assume all interpolation points \( \mu_1, \dots, \mu_w \) lie in the right half of the $s$-plane and the pair \((S_w, L_w)\) satisfies the Lyapunov equation:
\begin{align}  
-S_w P_w - P_w S_w^* + L_w L_w^T = 0. \label{pork_ps}  
\end{align}  
If \((S_w, L_w)\) is controllable, choosing the free parameter \( \hat{C} = L_w^T P_w^{-1} \) yields \( \hat{A} = S_w - L_w\hat{C} = -P_w S_w^* P_w^{-1} \). The resulting ROM
\begin{align}  
\hat{A} = -P_w S_w^* P_w^{-1}, \quad \hat{B} = \hat{W}^* B, \quad \hat{C} = L_w^T P_w^{-1},  
\end{align}  
satisfies the \(\mathcal{H}_2\)-optimality condition $\frac{\partial}{\partial\hat{B}}\bigl(\|G(s)-\hat{G}(s)\|_{\mathcal{H}_2}^2\bigr)=0$. This approach is referred to as Output PORK (O-PORK) in the remainder of the paper.
\subsection{Balanced Truncation (BT) and Its extensions}
The controllability Gramian \( P \) and observability Gramian \( Q \) of the state-space realization $(A,B,C,D)$ can be expressed in the frequency domain via the integrals
\begin{align}
P &= \frac{1}{2\pi} \int_{-\infty}^{\infty} (j\omega I - A)^{-1} BB^T (-j\omega I - A^T)^{-1} \, d\omega, \label{int1}\\
Q &= \frac{1}{2\pi} \int_{-\infty}^{\infty} (-j\omega I - A^T)^{-1} C^T C (j\omega I - A)^{-1} \, d\omega. \label{int2}
\end{align}
These Gramians are solutions to the Lyapunov equations
\begin{align}
AP + PA^T + BB^T = 0,\label{lyap_P}\\
A^TQ + QA + C^T C = 0.\label{lyap_Q}
\end{align}
If the Cholesky factorizations \( P = L_p L_p^T \) and \( Q = L_q L_q^T \) are computed, the balanced square-root algorithm (BSA) \citep{tombs1987truncated} provides a numerically robust way to perform BT \citep{moore1981principal}. The method starts with the singular value decomposition (SVD) of \( L_q^T L_p \):
\[
L_q^TL_p = \begin{bmatrix} U_1 & U_2 \end{bmatrix} \begin{bmatrix} \Sigma_{r} & 0 \\ 0 & \Sigma_{n-r} \end{bmatrix} \begin{bmatrix} V_1^T \\ V_2^T \end{bmatrix}.
\]
From this decomposition, the projection matrices are obtained as 
\[
\hat{W} = L_q U_1 \Sigma_r^{-\frac{1}{2}} \quad \text{and} \quad \hat{V} = L_p V_1 \Sigma_r^{-\frac{1}{2}}.  
\]  
The resulting ROM preserves the asymptotic stability of the original model \(G(s)\) as well as its \(r\) largest Hankel singular values \(\sqrt{\lambda_i(PQ)}\). To preserve additional properties—such as stable minimum phase, positive-realness, or bounded-realness—generalized versions of BT have been developed, which primarily differ in their definitions of the Gramians.  

In LQGBT \citep{jonckheere1983new}, the Gramian-like matrices are computed as stabilizing solutions to the following filter and controller Riccati equations:
\begin{align}
AP_{\text{LQG}}+P_{\text{LQG}}A^T+BB^T-P_{\text{LQG}}C^TCP_{\text{LQG}}=0,\label{Ricc_LQG_P}\\
A^TQ_{\text{LQG}}+Q_{\text{LQG}}A+C^TC-Q_{\text{LQG}}BB^TQ_{\text{LQG}}=0.\label{Ricc_LQG_Q}
\end{align}By replacing \( P \) and \( Q \) with \( P_{\text{LQG}} \) and \( Q_{\text{LQG}} \) in BT, respectively, a ROM can be obtained that is suitable for designing a reduced-order LQG controller with good closed-loop performance when used with the original plant.

In \(\mathcal{H}_\infty\)BT \citep{mustafa1991controller}, the Gramian-like matrices are computed as stabilizing solutions to the following filter and controller Riccati equations with \(\gamma > 0\):
\begin{align}
AP_{\mathcal{H}_\infty}+P_{\mathcal{H}_\infty}A^T+BB^T-(1-\gamma^{-2})P_{\mathcal{H}_\infty}C^TCP_{\mathcal{H}_\infty}=0,\label{Ricc_H_inf_P}\\
A^TQ_{\mathcal{H}_\infty}+Q_{\mathcal{H}_\infty}A+C^TC-(1-\gamma^{-2})Q_{\mathcal{H}_\infty}BB^TQ_{\mathcal{H}_\infty}=0.\label{Ricc_H_inf_Q}
\end{align}By replacing \( P \) and \( Q \) with \( P_{\mathcal{H}_\infty} \) and \( Q_{\mathcal{H}_\infty} \) in BT, respectively, a ROM can be obtained that is suitable for designing a reduced-order \(\mathcal{H}_\infty\) controller with good closed-loop performance when used with the original plant.

In PRBT \citep{green1988balanced,phillips2003guaranteed}, the Gramians are computed as stabilizing solutions to the following Riccati equations:
\begin{align}
AP_{\text{PR}}+P_{\text{PR}}A^T+(B-P_{\text{PR}}C^T)(D+D^T)^{-1}(B-P_{\text{PR}}C^T)^T=0,\label{P_PR}\\
A^TQ_{\text{PR}}+Q_{\text{PR}}A+(C-B^TQ_{\text{PR}})^T(D+D^T)^{-1}(C-B^TQ_{\text{PR}})=0.\label{Q_PR}
\end{align}By replacing \( P \) and \( Q \) with \( P_{\text{PR}} \) and \( Q_{\text{PR}} \) in BT, respectively, a ROM that preserves the positive-realness of the original model can be obtained.

In BRBT \citep{pernebo2003model,phillips2003guaranteed}, the Gramians are computed as stabilizing solutions to the following Riccati equations:
\begin{align}
AP_{\text{BR}}+P_{\text{BR}}A^T+BB^T+(P_{\text{BR}}C^T&+BD^T)(I-DD^T)^{-1}\nonumber\\
&\times(P_{\text{BR}}C^T+BD^T)^T=0,\label{P_BR}\\
A^TQ_{\text{BR}}+Q_{\text{BR}}A+C^TC+(B^TQ_{\text{BR}}&+D^TC)^T(I-D^TD)^{-1}\nonumber\\
&\times(B^TQ_{\text{BR}}+D^TC)=0.\label{Q_BR}
\end{align}By replacing \( P \) and \( Q \) with \( P_{\text{BR}} \) and \( Q_{\text{BR}} \) in BT, respectively, a ROM that preserves the bounded-realness of the original model can be obtained.

In SWBT \citep{zhou1995frequency}, the controllability Gramian \( P \) remains the same as in BT, while the weighted observability Gramian \( Q_{\text{SW}} \) solves the following Lyapunov equation:
\begin{align}
(A-BD^{-1}C)^TQ_{\text{SW}}+Q_{\text{SW}}(A-BD^{-1}C)+C^T(DD^T)^{-1}C=0.\label{Q_SW}
\end{align}By replacing \( Q \) with \( Q_{\text{SW}} \) in BT, a ROM that preserves the minimum-phase property of the original model can be obtained.

In BST \citep{green1988balanced,desai1984transformation}, the controllability Gramian \( P \) remains the same as in BT, while the weighted observability Gramian \( Q_{\text{S}} \) is computed as stabilizing solution to the following Riccati equation:
\begin{align}
A^TQ_{\text{S}}+Q_{\text{S}}A+\big(C-(CP+DB^T)Q_{\text{S}}\big)^T&(DD^T)^{-1}\nonumber\\
&\times\big(C-(CP+DB^T)Q_{\text{S}}\big)=0.\label{Q_S}
\end{align}By replacing \( Q \) with \( Q_{\text{S}} \) in BT, a ROM that preserves the minimum-phase property of the original model can be obtained. Unlike SWBT, BST can handle non-minimum-phase models as well. Both SWBT and BST minimize the relative error
\[ G^{-1}(s)\Big(G(s) - \hat{G}(s)\Big)
\] and tend to ensure uniform accuracy across the entire frequency spectrum.
\subsection{Quadrature-based BT (QuadBT) \citep{goseaQuad}}
The frequency-domain integrals in (\ref{int1}) and (\ref{int2}) can be approximated by numerical quadrature as: 
\begin{align}  
P &\approx \tilde{P} = \sum_{i=1}^{v} w_{p,i}^2 (j\sigma_i I - A)^{-1} B B^T (-j\sigma_i I - A^T)^{-1}, \nonumber \\  
Q &\approx \tilde{Q} = \sum_{i=1}^{w} w_{q,i}^2 (-j\mu_i I - A^T)^{-1} C^T C (j\mu_i I - A)^{-1}, \nonumber  
\end{align}  
where \(\sigma_i\) and \(\mu_i\) denote the quadrature nodes, and \(w_{p,i}^2\) and \(w_{q,i}^2\) are the associated weights. These approximations admit Cholesky-like factorizations \(\tilde{P} = \tilde{L}_p \tilde{L}_p^*\) and \(\tilde{Q} = \tilde{L}_q \tilde{L}_q^*\), which can be written as  
\[
\tilde{L}_p = \hat{V} \hat{L}_p, \qquad \tilde{L}_q = \hat{W} \hat{L}_q,
\]with \(\hat{V}\) and \(\hat{W}\) defined as in (\ref{Kry_V}) and (\ref{Kry_W}), respectively, and where
\begin{align}   
\hat{L}_p &= \text{diag}(w_{p,1}, \dots, w_{p,v}) \otimes I_m, \nonumber \\  
\hat{L}_q &= \text{diag}(w_{q,1}, \dots, w_{q,w}) \otimes I_p. \nonumber  
\end{align}  
In BSA, the exact Cholesky factors \(L_p\) and \(L_q\) are replaced by \(\tilde{L}_p\) and \(\tilde{L}_q\), resulting in the SVD  
\begin{align} 
\hat{L}_q^* (\hat{W}^*\hat{V}) \hat{L}_p = \begin{bmatrix} \tilde{U}_1 & \tilde{U}_2 \end{bmatrix} \begin{bmatrix} \tilde{\Sigma}_r & 0 \\ 0 & \tilde{\Sigma}_{n-r} \end{bmatrix} \begin{bmatrix} \tilde{V}_1^* \\ \tilde{V}_2^* \end{bmatrix}. \label{bsa_svd} 
\end{align} From this decomposition, the projection matrices are computed as:
\begin{align}
W_r = \hat{L}_q \tilde{U}_1 \tilde{\Sigma}_r^{-1/2} \quad \text{and} \quad V_r = \hat{L}_p \tilde{V}_1 \tilde{\Sigma}_r^{-1/2}.  \label{bsa_proj}
\end{align}
Finally, the ROM in QuadBT is obtained as:
\begin{align}  
\hat{A} &= W_r^* (\hat{W}^* A \hat{V}) V_r, & \hat{B} &= W_r^* (\hat{W}^* B), & \hat{C} &= (C \hat{V}) V_r. \label{bsa_rom}
\end{align}
Note that \(\hat{L}_p\) and \(\hat{L}_q\) depend solely on the quadrature weights. Furthermore, the quantities \(\hat{W}^*\hat{V}\), \(\hat{W}^* A \hat{V}\), \(\hat{W}^* B\), \(C \hat{V}\), and \(D\) can be constructed in a non-intrusive manner using frequency-domain samples of the transfer function \(G(s)\) within the Loewner framework \eqref{LF}. Thus, the ROM in \eqref{bsa_rom} can be constructed non-intrusively, without access to the state-space realization \((A, B, C, D)\).

In \citep{reiter2023generalizations}, the following observations were made about PRBT, BRBT, and BST:
\begin{enumerate}
\item In PRBT, the Gramians $P_{\text{PR}}$ and $Q_{\text{PR}}$ are respectively the controllability Gramian and observability Gramian associated with state-space realizations of the spectral factorizations of $G(s)+G^*(s)$.
\item In BRBT, the Gramians $P_{\text{BR}}$ and $Q_{\text{BR}}$ are respectively the controllability Gramian and observability Gramian associated with state-space realizations of the spectral factorizations of $I_m-G^*(s)G(s)$ and $I_p-G(s)G^*(s)$.
\item In BST, the observability Gramian $Q_\text{S}$ is the frequency-weighted observability Gramian with weight associated with the spectral factorization of $G(s)G^*(s)$.
\item The QuadBT framework could be extended to PRBT, BRBT, and BST if samples were available for the spectral factorizations of $G(s)+G^*(s)$, $I_m-G^*(s)G(s)$, $I_p-G(s)G^*(s)$, and $G(s)G^*(s)$.
\item Practical methods to acquire these samples are not yet available and are left for future work.
\end{enumerate}
We do not discuss these generalizations in detail here because we focus on methods requiring only samples of $G(s)$, for which practical measurement methods exist, unlike for the spectral factorizations of $G(s)+G^*(s)$, $I_m-G^*(s)G(s)$, $I_p-G(s)G^*(s)$, and $G(s)G^*(s)$.
\subsection{Non-intrusive ADI-based Approximation of BT \citep{zulfiqar2025non}}
In \citep{wolf2016adi}, it is demonstrated that the PORK algorithm and the low-rank Cholesky factor ADI (LRCF-ADI) method \citep{benner2013efficient} yield identical low-rank approximations of the solutions to Lyapunov equations, provided the ADI shift parameters are chosen as the mirror images of the interpolation points used in PORK. Specifically, when the ADI shifts are set to \((-\sigma_1, \dots, -\sigma_v)\), the resulting approximation of the controllability Gramian \(P\) coincides with the I-PORK output and can be written as
\[
P \approx \tilde{P} = \hat{V} Q_v^{-1} \hat{V}^*.
\]
Likewise, for the observability Gramian \(Q\), using ADI shifts \((-\mu_1, \dots, -\mu_w)\) leads to the O-PORK approximation
\[
Q \approx \tilde{Q} = \hat{W} P_w^{-1} \hat{W}^*.
\]
By computing Cholesky-like factorizations \(Q_v^{-1} = \hat{L}_p \hat{L}_p^*\) and \(P_w^{-1} = \hat{L}_q \hat{L}_q^*\), and defining the low-rank factors \(\tilde{L}_p = \hat{V} \hat{L}_p\) and \(\tilde{L}_q = \hat{W} \hat{L}_q\), the approximations admit the factorized representations \(\tilde{P} = \tilde{L}_p \tilde{L}_p^*\) and \(\tilde{Q} = \tilde{L}_q \tilde{L}_q^*\). Note that, unlike QuadBT, \(\hat{L}_p\) and \(\hat{L}_q\) here are not diagonal matrices. As observed in \citep{zulfiqar2025non}, substituting the exact Gramian factors \(L_p, L_q\) in BSA with these approximations \(\tilde{L}_p, \tilde{L}_q\) yields a non-intrusive ADI-based approximation of BT (NI-ADI-BT). This is because \(\hat{L}_p\) and \(\hat{L}_q\) depend solely on the interpolation points, and all required projected system quantities—namely \(\hat{W}^* \hat{V}\), \(\hat{W}^* A \hat{V}\), \(\hat{W}^* B\), \(C \hat{V}\), and \(D\)—can be assembled directly from transfer function evaluations using the Loewner framework \eqref{LF}, without access to the original state-space matrices $(A,B,C,D)$. Consequently, NI-ADI-BT can be implemented using transfer function samples in the right half of the $s$-plane, in contrast to QuadBT, which relies on samples along the $j\omega$ axis.
\section{Non-intrusive ADI-based Approximations of Various Generalizations of BT}\label{sec3}
Recall that the interpolant $\hat{G}(s)=C\hat{V}(sI-S_v+\zeta L_v)^{-1}\zeta+D$ interpolates $G(s)$ at $(\sigma_1,\cdots,\sigma_v)$, where $\hat{V}$ is defined in (\ref{Kry_V}) and $\zeta$ is a free parameter. Similarly, the interpolant $\hat{G}(s) = \zeta(sI - S_w + L_w\zeta)^{-1} \hat{W}^*B + D$ interpolates $G(s)$ at $(\mu_1, \ldots, \mu_w)$, where $\hat{W}$ defined in (\ref{Kry_W}). This section derives specific choices of free parameter $\zeta$ so that rational interpolation produces approximations of various Lyapunov and Riccati equations defining Gramians in various generalizations of BT such that these approximations are identical to ones produced by ADI methods for Lyapunov equations \citep{benner2013efficient} and Riccati equations \citep{benner2018radi}. Furthermore, when these approximations are inserted into BSA, it is noted that the resulting ROM in various generalizations of BT can be obtained non-intrusively from the samples of $G(s)$ in the right-half of the $s$-plane. The parameter $\zeta$ can be computed by solving specific Lyapunov and Sylvester equations according to the specific generalization of BT.

The following assumptions are made on the $D$ matrix for PRBT, BRBT, SWBT, and BST. For PRBT, $D + D^T>0$. For BRBT, $I - DD^T>0$ and $I - D^T D>0$. For SWBT and BST, $DD^T>0$.
\subsection{Non-intrusive ADI-based Approximation of LQGBT}\label{sec3_1}
As demonstrated in \citep{wolf2016adi}, the LRCF-ADI method for computing the low-rank solution of $P$ with shifts $(-\sigma_1,\dots,-\sigma_v)$ corresponds to a Petrov-Galerkin approach. This method implicitly interpolates at $(\sigma_1,\dots,\sigma_v)$ using the projection matrix $\hat{V}$, while the projection matrix $\hat{W}$ is implicitly chosen to satisfy $\hat{W}^*\hat{V}=I$, ensuring that the poles of $\hat{A}$ are placed at $(-\sigma_1^*,\dots,-\sigma_v^*)$. Similarly, \citep{bertram2024family} shows that the low-rank solution of $P_{\text{LQG}}$ obtained via ADI method for Riccati equations (RADI) with shifts $(-\sigma_1,\dots,-\sigma_v)$ can be interpreted as a Petrov-Galerkin method. Here, interpolation is implicitly performed at $(\sigma_1,\dots,\sigma_v)$ using $\hat{V}$, and $\hat{W}$ is selected such that $\hat{W}^*\hat{V}=I$, placing the poles of $\hat{A}-\hat{P}_{\text{LQG}}\hat{C}^*\hat{C}$ at $(-\sigma_1^*,\dots,-\sigma_v^*)$. The matrix $\hat{P}_{\text{LQG}}$ solves the projected Riccati equation:  
\begin{align}  
\hat{A}\hat{P}_{\text{LQG}}+\hat{P}_{\text{LQG}}\hat{A}^*+\hat{B}\hat{B}^*-\hat{P}_{\text{LQG}}\hat{C}^*\hat{C}\hat{P}_{\text{LQG}}=0.\label{proj_ricc_lqg}
\end{align}  
In the following theorem, we specify a choice of $\zeta$ for the interpolant $\hat{G}(s)=C\hat{V}(sI-S_v+\zeta L_v)^{-1}\zeta+D$ such that it places the poles of $\hat{A}-\hat{P}_{\text{LQG}}\hat{C}^*\hat{C}$ at $(-\sigma_1^*,\dots,-\sigma_v^*)$, thereby yielding the same approximation as RADI.  
\begin{theorem}\label{Theorem1}
Let \(\hat{V}\) be as defined in (\ref{Kry_V}), with all interpolation points \((\sigma_1, \dots, \sigma_v)\) located in the right half of the \( s \)-plane. Assume further that the pair \((S_v, L_v)\) is observable and that \(Q_v > 0\) uniquely solves the Lyapunov equation:  
\begin{align}
-S_v^* Q_v - Q_v S_v + L_v^T L_v + \hat{C}^* \hat{C} = 0.\label{Qv_LQG}
\end{align} 
Then, the ROM \(\hat{H}(s)\), defined by  
\[
\hat{A} = S_v - \hat{B} L_v, \quad \hat{B} = Q_v^{-1} L_v^T, \quad \hat{C} = C \hat{V},
\]  
which interpolates \(H(s)\) at \((\sigma_1, \dots, \sigma_v)\), satisfies the following properties:
\begin{enumerate}
\item The matrix \(\hat{A}\) equals \(Q_v^{-1}(-S_v^* + \hat{C}^* \hat{C} Q_v^{-1}) Q_v\).  
\item The solution \(\hat{P}_{\text{LQG}}\) to the projected Riccati equation (\ref{proj_ricc_lqg}) is \(Q_v^{-1}\).  
\item The matrix \(\hat{A} - \hat{P}_{\text{LQG}} \hat{C}^* \hat{C}\) is Hurwitz, with eigenvalues at \((-\sigma_1^*, \dots, -\sigma_v^*)\).
\end{enumerate}
\end{theorem}
\begin{proof}
\begin{enumerate}
  \item\label{item1} Pre-multiplying (\ref{Qv_LQG}) by $Q_v^{-1}$ yields:
\begin{align}
-Q_v^{-1}S_v^*Q_v-S_v&+Q_v^{-1}L_v^TL_v+Q_v^{-1}\hat{C}^*\hat{C}=0\nonumber\\
S_v-\hat{B}L_v&=-Q_v^{-1}S_v^*Q_v+Q_v^{-1}\hat{C}^*\hat{C}\nonumber\\
\hat{A}&=Q_v^{-1}(-S_v^*+\hat{C}^*\hat{C}Q_v^{-1})Q_v.\nonumber
\end{align}
\item\label{item2} Consider the expression:
    \begin{align}
&\hat{A}Q_v^{-1}+Q_v^{-1}\hat{A}^*+\hat{B}\hat{B}^*-Q_v^{-1}\hat{C}^*\hat{C}Q_v^{-1}\nonumber\\
&=-Q_v^{-1}S_v^*+Q_v^{-1}\hat{C}^*\hat{C}Q_v^{-1}-S_vQ_v^{-1}+Q_v^{-1}\hat{C}^*\hat{C}Q_v^{-1}\nonumber\\
&\hspace*{4.2cm}+Q_v^{-1}L_v^TL_vQ_v^{-1}-Q_v^{-1}\hat{C}^*\hat{C}Q_v^{-1}\nonumber\\
&=-Q_v^{-1}S_v^*-S_vQ_v^{-1}+Q_v^{-1}L_v^TL_vQ_v^{-1}+Q_v^{-1}\hat{C}^*\hat{C}Q_v^{-1}\nonumber\\
&=Q_v^{-1}\big(-S_v^*Q_v-Q_vS_v+L_v^TL_v+\hat{C}^*\hat{C}\big)Q_v^{-1}\nonumber\\
&=0.\nonumber
    \end{align}This shows that $Q_v^{-1}$ satisfies the projected Riccati equation (\ref{proj_ricc_lqg}).
    \item Given $\hat{A}= -\hat{P}_{\text{LQG}}S_v^*\hat{P}_{\text{LQG}}^{-1}+\hat{P}_{\text{LQG}} \hat{C}^* \hat{C}$, the eigenvalues of $\hat{A}$ are equal to those of $-S_v^*+\hat{P}_{\text{LQG}}\hat{C}^*\hat{C}$. As $-S_v^*$ is Hurwitz, the matrix $\hat{A}-\hat{P}_{\text{LQG}}\hat{C}^*\hat{C}$ is Hurwitz with eigenvalues at $(-\sigma_1^*,\cdots,-\sigma_v^*)$.
\end{enumerate}
\end{proof}It is clear from Theorem \ref{Theorem1} that the RADI-based approximation of \(P_{\text{LQG}}\) is given by \(P_{\text{LQG}} \approx \hat{V} Q_v^{-1} \hat{V}^*\). 

Note that $Q_{\text{LQG}}$ is the dual of $P_{\text{LQG}}$. For completeness, we present the dual of Theorem \ref{Theorem1} below, which produces the same approximation of $Q_{\text{LQG}}$ as RADI when the shifts are chosen as $(-\mu_1,\cdots,-\mu_w)$.
\begin{theorem}\label{Theorem2}
Let \(\hat{W}\) be as defined in (\ref{Kry_W}), with all interpolation points \((\mu_1, \dots, \mu_w)\) located in the right half of the \( s \)-plane. Assume further that the pair \((S_w, L_w)\) is controllable and that \(P_w > 0\) uniquely solves the Lyapunov equation:  
\begin{align}
-S_wP_w-P_wS_w^*+L_wL_w^T+\hat{B}\hat{B}^*=0.\label{Pw_LQG}
\end{align} 
Then, the ROM \(\hat{H}(s)\), defined by  
\[
\hat{A}=S_w-L_w\hat{C}, \quad \hat{B}=\hat{W}^*B,\quad \hat{C}=L_w^TP_w^{-1},
\]  
which interpolates \(H(s)\) at \((\mu_1, \dots, \mu_w)\), satisfies the following properties:
\begin{enumerate}
\item The matrix \(\hat{A}\) equals \(P_w(-S_w^*+P_w^{-1}\hat{B}\hat{B}^*)P_w^{-1}\).  
\item The solution \(\hat{Q}_{\text{LQG}}\) to the following projected Riccati equation \begin{align}
    \hat{A}^*\hat{Q}_{\text{LQG}}+\hat{Q}_{\text{LQG}}\hat{A}+\hat{C}^*\hat{C}-\hat{Q}_{\text{LQG}}\hat{B}\hat{B}^*\hat{Q}_{\text{LQG}}=0\label{proj_ricc_lqg_Q}
    \end{align}
    is \(P_w^{-1}\).
\item The matrix \(\hat{A}-\hat{B}\hat{B}^*\hat{Q}_{\text{LQG}}\) is Hurwitz, with eigenvalues at \((-\mu_1^*, \dots, -\mu_w^*)\).
\end{enumerate}
\end{theorem}
\begin{proof}
The proof is similar to that of Theorem \ref{Theorem1} and hence omitted for brevity.
\end{proof}
Again, it is clear from Theorem \ref{Theorem2} that the RADI-based approximation of \(Q_{\text{LQG}}\) is given by \(Q_{\text{LQG}} \approx \hat{W} P_w^{-1} \hat{W}^*\). 

Similar to NI-ADI-BT \citep{zulfiqar2025non}, the ADI-based non-intrusive approximation of LQGBT (NI-ADI-LQGBT) can be directly obtained using the results of Theorems \ref{Theorem1} and \ref{Theorem2}. Specifically, by computing Cholesky-like factorizations \(Q_v^{-1} = \hat{L}_p \hat{L}_p^*\) and \(P_w^{-1} = \hat{L}_q \hat{L}_q^*\), and defining the low-rank factors \(\tilde{L}_p = \hat{V} \hat{L}_p\) and \(\tilde{L}_q = \hat{W} \hat{L}_q\), the approximations admit the factorized representations \(P_{\mathrm{LQG}}\approx \tilde{L}_p \tilde{L}_p^*\) and \(Q_{\mathrm{LQG}}\approx\tilde{L}_q \tilde{L}_q^*\). Substituting the exact Gramian factors \(L_p, L_q\) in BSA with these approximations \(\tilde{L}_p, \tilde{L}_q\) yields NI-ADI-LQGBT, which can be implemented using transfer function samples in the right half of the $s$-plane.
\begin{remark}
The ADI-based non-intrusive approximation of $\mathcal{H}_\infty$BT (NI-ADI-$\mathcal{H}_\infty$BT) is similar to NI-ADI-LQGBT with the only difference being in the computation of $Q_v$ and $P_w$. In NI-ADI-$\mathcal{H}_\infty$BT, $Q_v$ and $P_w$ are computed by solving the following Lyapunov equations:
\begin{align}
-S_v^* Q_v - Q_v S_v + L_v^T L_v + (1-\gamma^{-2})\hat{C}^* \hat{C} &= 0,\\
-S_wP_w-P_wS_w^*+L_wL_w^T+(1-\gamma^{-2})\hat{B}\hat{B}^*&=0.
\end{align}
\end{remark}
\subsection{Non-intrusive ADI-based Approximation of PRBT}\label{sec3_2}
Let us define
\begin{align}G_{\text{PR}}(s)=C_{\text{PR}}(sI-A_{\text{PR}})^{-1}B_{\text{PR}},\nonumber
\end{align} where
\begin{align}
A_{\text{PR}}&=A-B_{\text{PR}}C_{\text{PR}},& B_{\text{PR}}&=BR_{\text{PR}}^{-\frac{1}{2}},\nonumber\\
C_{\text{PR}}&=R_{\text{PR}}^{-\frac{1}{2}}C,& R_{\text{PR}}&=D+D^T.\nonumber
\end{align}
Further, define the ROM
\begin{align}
\hat{G}_{\text{PR}}(s)=\hat{C}_{\text{PR}}(sI-\hat{A}_{\text{PR}})^{-1}\hat{B}_{\text{PR}}\nonumber
\end{align}obtained as follows:
\begin{align}
\hat{A}_{\text{PR}}&=W_{\text{PR}}^*(A_{\text{PR}})V_{\text{PR}}=W_{\text{PR}}^*(A-BR_{\text{PR}}^{-1}C)V_{\text{PR}}=\hat{A}-\hat{B}R_{\text{PR}}^{-1}\hat{C},\nonumber\\
\hat{B}_{\text{PR}}&=W_{\text{PR}}^*B_{\text{PR}}=W_{\text{PR}}^*BR_{\text{PR}}^{-\frac{1}{2}}=\hat{B} R_{\text{PR}}^{-\frac{1}{2}},\nonumber\\
\hat{C}_{\text{PR}}&=C_{\text{PR}}V_{\text{PR}}=R_{\text{PR}}^{-\frac{1}{2}}CV_{\text{PR}}=R_{\text{PR}}^{-\frac{1}{2}}\hat{C},\nonumber
\end{align}where $W_{\text{PR}}^*V_{\text{PR}}=I$.

The Riccati equations (\ref{P_PR}) and (\ref{Q_PR}) can be rewritten as follows:
\begin{align}
A_{\text{PR}}P_{\text{PR}}+P_{\text{PR}}A_{\text{PR}}^T+B_{\text{PR}}B_{\text{PR}}^T+P_{\text{PR}}C_{\text{PR}}^TC_{\text{PR}}P_{\text{PR}}=0,\\
A_{\text{PR}}^TQ_{\text{PR}}+Q_{\text{PR}}A_{\text{PR}}+C_{\text{PR}}^TC_{\text{PR}}+Q_{\text{PR}}B_{\text{PR}}B_{\text{PR}}^TQ_{\text{PR}}=0.
\end{align}
From Theorem \ref{Theorem1}, the RADI-based low-rank approximation of $P_{\text{PR}}$ can be obtained as follows. The projection matrix
\begin{align}
V_{\text{PR}}=\begin{bmatrix}(\sigma_1I-A_{\text{PR}})^{-1}B_{\text{PR}}&\cdots&(\sigma_vI-A_{\text{PR}})^{-1}B_{\text{PR}}\end{bmatrix}
\end{align} solves the following Sylvester equation:
\begin{align}
A_{\text{PR}}V_{\text{PR}}-V_{\text{PR}}S_v+B_{\text{PR}}L_v=0.\label{PR_Sylv}
\end{align}
Then $\hat{C}_{\text{PR}}=C_{\text{PR}}V_{\text{PR}}$ can be computed as follows:
\begin{align}
\hat{C}_{\text{PR}}=\begin{bmatrix}G_{\text{PR}}(\sigma_1)&\cdots&G_{\text{PR}}(\sigma_v)\end{bmatrix}.
\end{align}Thereafter, $Q_v$ is computed by solving the Lyapunov equation:
\begin{align}
-S_v^*Q_v-Q_vS_v+L_v^TL_v-\hat{C}_{\text{PR}}^*\hat{C}_{\text{PR}}=0.\label{th3_sylv}
\end{align}
The RADI-based approximation of $P_{\text{PR}}$ is given by $V_{\text{PR}} Q_v^{-1} V_{\text{PR}}^*$.

According to Theorem \ref{Theorem1}, the state-space realization of the ROM $\hat{G}_{\text{PR}}(s)=\hat{C}_{\text{PR}}(sI-\hat{A}_{\text{PR}})^{-1}\hat{B}_{\text{PR}}$ (where $\hat{A}_{\text{PR}}=S_v-Q_v^{-1}L_v^TL_v$, $\hat{B}_{\text{PR}}=Q_v^{-1}L_v^T$, and $\hat{C}_{\text{PR}}=C_{\text{PR}}V_{\text{PR}}$) solves the following projected Riccati equations with $\hat{P}_{\text{PR}}=Q_v^{-1}>0$:
\begin{align}
\hat{A}_{\text{PR}}\hat{P}_{\text{PR}}+\hat{P}_{\text{PR}}\hat{A}_{\text{PR}}^*+\hat{B}_{\text{PR}}\hat{B}_{\text{PR}}^*+\hat{P}_{\text{PR}}\hat{C}_{\text{PR}}^*\hat{C}_{\text{PR}}\hat{P}_{\text{PR}}=0.\label{proj_ricc_Ppr}
\end{align}Since
\begin{align}
\hat{A}&=S_v-Q_v^{-1}L_v^TL_v+Q_v^{-1}L_v^T\hat{C}_{\text{PR}},\nonumber\\
\hat{B}&=Q_v^{-1}L_v^TR_{\text{PR}}^{\frac{1}{2}},\quad \hat{C}=R_{\text{PR}}^{\frac{1}{2}}\hat{C}_{\text{PR}},\label{ss_PR}
\end{align} the Riccati equation (\ref{proj_ricc_Ppr}) can be rewritten as follows:
\begin{align}
\hat{A}\hat{P}_{\text{PR}}+\hat{P}_{\text{PR}}\hat{A}^*+(\hat{B}-\hat{P}_{\text{PR}}\hat{C}^*)R_{\text{PR}}^{-1}(\hat{B}-\hat{P}_{\text{PR}}\hat{C}^*)^*=0.
\end{align}
\begin{theorem}\label{Theorem3}
Let $G(s)$ be a positive-real transfer function. Define $\tilde{L}_v=R_{\text{PR}}^{-\frac{1}{2}}(L_v-\hat{C}_{\text{PR}})=\begin{bmatrix}l_{v,1}&\cdots&l_{v,v}\end{bmatrix}$ and $T_{v}=\text{blkdiag}(l_{v,1},\cdots,l_{v,v})$. Assume that:
\begin{enumerate}
\item The interpolation points \((\sigma_1, \dots, \sigma_v)\) are located in the right half of the \( s \)-plane. 
  \item The pair $(S_v,\tilde{L}_v)$ is observable.
  \item The matrix $T_{v}$ is invertible.
  \item The realization $(\hat{A},\hat{B},\hat{C})$ defined in (\ref{ss_PR}) is minimal.
\end{enumerate}Then the following statements hold:
\begin{enumerate}
  \item The projection matrix $V_\text{PR}$ satisfies $V_\text{PR}=\hat{V}T_{v}$, where $\hat{V}$ is as in (\ref{Kry_V}).
  \item The ROM $\hat{G}(s)=\hat{C}(sI-\hat{A})^{-1}\hat{B}+D$ interpolates $G(s)$ at $(\sigma_1,\cdots,\sigma_v)$.
  \item The ROM $\hat{G}(s)$ is positive-real. 
\end{enumerate}
\end{theorem}
\begin{proof}
By rearranging variables, the Sylvester equations (\ref{PR_Sylv}) can be rewritten as follows:
\begin{align}
AV_{\text{PR}}-V_{\text{PR}}S_v+B\tilde{L}_v=0.
\end{align} It can then readily be noted that \( V_{\text{PR}} = \hat{V}T_{v} \), where \(\hat{V}\) is as in (\ref{Kry_V}). Thus, \( V_{\text{PR}} \) and \(\hat{V}\) enforce the same interpolation conditions, i.e., \( G(\sigma_i) = \hat{G}(\sigma_i) \), since the columns of \( V_{\text{PR}} \) and \(\hat{V}\) span the same subspace and produce the same ROM with different state-space realizations \citep{gallivan2004sylvester}.  

By defining \( K_{\text{PR}} = (\hat{B} - \hat{P}_{\text{PR}}\hat{C}^*)R_{\text{PR}}^{-\frac{1}{2}} \), the following holds:  
\begin{align}
\hat{A}\hat{P}_{\text{PR}}+\hat{P}_{\text{PR}}\hat{A}^*&=-K_{\text{PR}}K_{\text{PR}}^*\nonumber\\
\hat{P}_{\text{PR}}\hat{C}^*-\hat{B}&=-K_{\text{PR}}R_{\text{PR}}^{\frac{1}{2}}\nonumber\\
R_{\text{PR}}^{\frac{1}{2}}(R_{\text{PR}}^{\frac{1}{2}})^T&=D+D^T.
\end{align} Since \(\hat{P}_{\text{PR}} > 0\), \(\hat{G}(s)\) satisfies the positive-real lemma, and thus \(\hat{G}(s)\) is a positive-real transfer function \citep{phillips2003guaranteed}.
\end{proof}
Now, assume that the matrices of the ROM $\hat{G}_{\text{PR}}(s)=\hat{C}_{\text{PR}}(sI-\hat{A}_{\text{PR}})^{-1}\hat{B}_{\text{PR}}$ are obtained as follows:
\begin{align}
\hat{A}_{\text{PR}}&=\hat{W}^*(A_{\text{PR}})\hat{V}=\hat{W}^*(A-BR_{\text{PR}}^{-1}C)\hat{V}=S_v-\zeta L_v-\zeta R_{\text{PR}}^{-1}C\hat{V},\nonumber\\
\hat{B}_{\text{PR}}&=\hat{W}^*B_{\text{PR}}=\hat{W}^*BR_{\text{PR}}^{-\frac{1}{2}}=\zeta R_{\text{PR}}^{-\frac{1}{2}},\nonumber\\
\hat{C}_{\text{PR}}&=C_{\text{PR}}\hat{V}=R_{\text{PR}}^{-\frac{1}{2}}C\hat{V},\label{th3_rom}
\end{align}where $\hat{V}$ is as in (\ref{Kry_V}) and $\hat{W}$ is an unknown projection matrix satisfying $\hat{W}^* \hat{V} = I$. According to Theorem \ref{Theorem3}, the ROM $\hat{G}_{\text{PR}}(s)$ in (\ref{th3_rom}) interpolates $G_{\text{PR}}(s)$ at the interpolation points $(\sigma_1,\cdots,\sigma_v)$ since $V_{\text{PR}}=\hat{V}T_v$. An immediate consequence is that a ROM $\hat{G}_{\text{PR}}(s)$ interpolating $G_{\text{PR}}(s)$ can be constructed directly from the ROM $\hat{G}(s)=C\hat{V}(sI-S_v+\zeta L_v)^{-1}\zeta+D$ using the formula (\ref{th3_rom}). We now use this observation to derive an expression for $T_v$.

Note that
\begin{align}
\hat{C}_{\text{PR}}&=C_{\text{PR}}V_{\text{PR}}=\begin{bmatrix}G_{\text{PR}}(\sigma_1)&\cdots&G_{\text{PR}}(\sigma_v)\end{bmatrix}=\begin{bmatrix}\hat{G}_{\text{PR}}(\sigma_1)&\cdots&\hat{G}_{\text{PR}}(\sigma_v)\end{bmatrix}\nonumber\\
&=C_{\text{PR}}\hat{V}\begin{bsmallmatrix}(\sigma_1I-S_v+\zeta L_v+\zeta R_{\text{PR}}^{-1}C\hat{V})^{-1}\zeta R_{\text{PR}}^{-\frac{1}{2}}&\cdots&(\sigma_vI-S_v+\zeta L_v+\zeta R_{\text{PR}}^{-1}C\hat{V})^{-1}\zeta R_{\text{PR}}^{-\frac{1}{2}}\end{bsmallmatrix}\nonumber\\
&=C_{\text{PR}}\hat{V}T_v.\nonumber
\end{align}
It follows directly that the matrix $T_v$ can be computed from the interpolant $\hat{G}(s)=C\hat{V}(sI-S_v+\zeta L_v)^{-1}\zeta+D$ of $G(s)$. Such an interpolant can be generated by I-PORK, with guaranteed stability. Subsequently,
\begin{align}T_v=\begin{bsmallmatrix}(\sigma_1I-S_v+\zeta L_v+\zeta R_{\text{PR}}^{-1}C\hat{V})^{-1}\zeta R_{\text{PR}}^{-\frac{1}{2}}&\cdots&(\sigma_vI-S_v+\zeta L_v+\zeta R_{\text{PR}}^{-1}C\hat{V})^{-1}\zeta R_{\text{PR}}^{-\frac{1}{2}}\end{bsmallmatrix}\nonumber
\end{align}can be obtained by solving the Sylvester equation:
\begin{align}
(S_v-\zeta L_v-\zeta R_{\text{PR}}^{-1}C\hat{V})T_{v}-T_{v}S_v+\zeta R_{\text{PR}}^{-\frac{1}{2}}L_v=0.\label{Tv_pr}
\end{align}
An immediate consequence is that positive-realness can be enforced non-intrusively on the ROM produced by I-PORK, using only samples of $G(s)$ at $\sigma_i$ and $G(\infty)$. 

Next, the RADI-based approximation of $Q_{\text{PR}}$ can be obtained using Theorem \ref{Theorem2} as follows. The projection matrix
\begin{align}
W_{\text{PR}}^*=\begin{bmatrix}C_{\text{PR}}(\mu_1I-A_{\text{PR}})^{-1}\\\vdots\\C_{\text{PR}}(\mu_wI-A_{\text{PR}})^{-1}\end{bmatrix}
\end{align} solves the Sylvester equation:
\begin{align}
A_{\text{PR}}^TW_{\text{PR}}-W_{\text{PR}}S_w^*+C_{\text{PR}}^TL_w^T=0.\label{PR_Sylv_W}
\end{align}
Then, $\hat{B}_{\text{PR}}=W_{\text{PR}}^*B_{\text{PR}}$ can be computed as:
\begin{align}
\hat{B}_{\text{PR}}=\begin{bmatrix}G_{\text{PR}}(\mu_1)\\\vdots\\G_{\text{PR}}(\mu_w)\end{bmatrix}.
\end{align}Subsequently, $P_w$ is obtained by solving the Lyapunov equation:
\begin{align}
-S_wP_w-P_wS_w^*+L_wL_w^T-\hat{B}_{\text{PR}}\hat{B}_{\text{PR}}^*=0.\label{th4_sylv}
\end{align}
The RADI-based approximation of $Q_{\text{PR}}$ is then given by $W_{\text{PR}}P_w^{-1}W_{\text{PR}}^*$.

According to Theorem \ref{Theorem2}, the state-space realization of the ROM $\hat{G}_{\text{PR}}(s)=\hat{C}_{\text{PR}}(sI-\hat{A}_{\text{PR}})^{-1}\hat{B}_{\text{PR}}$ (where $\hat{A}_{\text{PR}}=S_w-L_wL_w^TP_w^{-1}$, $\hat{B}_{\text{PR}}=W_{\text{PR}}^*B_{\text{PR}}$, and $\hat{C}_{\text{PR}}=L_w^TP_w^{-1}$) satisfies the projected Riccati equations with $\hat{Q}_{\text{PR}}=P_w^{-1}>0$
\begin{align}
\hat{A}_{\text{PR}}^*\hat{Q}_{\text{PR}}+\hat{Q}_{\text{PR}}\hat{A}_{\text{PR}}+\hat{C}_{\text{PR}}^*\hat{C}_{\text{PR}}+\hat{Q}_{\text{PR}}\hat{B}_{\text{PR}}\hat{B}_{\text{PR}}^*\hat{Q}_{\text{PR}}=0.\label{proj_ricc_Qpr}
\end{align}Since
\begin{align}
\hat{A}&=S_w-L_wL_w^TP_w^{-1}+\hat{B}_{\text{PR}}L_w^TP_w^{-1},\nonumber\\
\hat{B}&=\hat{B}_{\text{PR}}R_{\text{PR}}^{\frac{1}{2}},\quad \hat{C}=R_{\text{PR}}^{\frac{1}{2}}L_w^TP_w^{-1},\label{ss_PR_2}
\end{align}the Riccati equation (\ref{proj_ricc_Qpr}) can be rewritten as
\begin{align}
\hat{A}^*\hat{Q}_{\text{PR}}+\hat{Q}_{\text{PR}}\hat{A}+(\hat{C}-\hat{B}^*\hat{Q}_{\text{PR}})^*R_{\text{PR}}^{-1}(\hat{C}-\hat{B}^*\hat{Q}_{\text{PR}})=0.
\end{align}
\begin{theorem}\label{Theorem4}
Let $G(s)$ be a positive-real transfer function. Define $\tilde{L}_w^*=R_{\text{PR}}^{-\frac{1}{2}}(L_w^T-\hat{B}_{\text{PR}}^*)=\begin{bmatrix}l_{w,1}^*&\cdots&l_{w,w}^*\end{bmatrix}$ and $T_{w}=\text{blkdiag}(l_{w,1}^*,\cdots,l_{w,w}^*)$. Assume that:
\begin{enumerate}
\item The interpolation points $(\mu_1,\cdots,\mu_w)$ are located in the right half of the $s$-plane.
  \item The pair $(S_w,\tilde{L}_w)$ is controllable.
  \item The matrix $T_{w}$ is invertible.
  \item The realization $(\hat{A},\hat{B},\hat{C})$ defined in (\ref{ss_PR_2}) is minimal.
\end{enumerate}Then the following statements hold:
\begin{enumerate}
  \item The projection matrix $W_\text{PR}$ satisfies $W_\text{PR}=\hat{W}T_{w}$, where $\hat{W}$ is as in (\ref{Kry_W}).
  \item The ROM $\hat{G}(s)=\hat{C}(sI-\hat{A})^{-1}\hat{B}+D$ interpolates $G(s)$ at $(\mu_1,\cdots,\mu_w)$.
  \item The ROM $\hat{G}(s)$ is positive-real. 
\end{enumerate}
\end{theorem}
\begin{proof}
The proof is similar to the proof of Theorem \ref{Theorem3} and hence omitted for brevity.
\end{proof}
We now derive an expression for \( T_{w} \) when \(\hat{W}\) is used to enforce interpolation conditions instead of \( W_{\text{PR}} \). Since \( W_{\text{PR}} \) and \(\hat{W}\) enforce identical interpolation conditions, the following holds:
\begin{align}
\hat{B}_{\text{PR}}&=W_{\text{PR}}^*B_{\text{PR}}=\begin{bmatrix}G_{\text{PR}}(\mu_1)\\\vdots\\G_{\text{PR}}(\mu_w)\end{bmatrix}=\begin{bmatrix}\hat{G}_{\text{PR}}(\mu_1)\\\vdots\\\hat{G}_{\text{PR}}(\mu_w)\end{bmatrix}\nonumber\\
&=\begin{bmatrix}R_{\text{PR}}^{-\frac{1}{2}}\zeta(\mu_1I-S_w+L_w\zeta+\hat{W}^*BR_{\text{PR}}^{-1}\zeta)^{-1}\\\vdots\\R_{\text{PR}}^{-\frac{1}{2}}\zeta(\mu_wI-S_w+L_w\zeta+\hat{W}^*BR_{\text{PR}}^{-1}\zeta)^{-1}\end{bmatrix}\hat{W}^*B_{\text{PR}},\nonumber\\
&=T_{w}^*\hat{W}^*BR_{\text{PR}}^{-\frac{1}{2}},\nonumber
\end{align}
where \(\hat{G}(s) = \zeta (s I - S_w + L_w \zeta)^{-1} \hat{W}^* B + D\) is an interpolant of \( G(s) \). Such an interpolant can be generated by O-PORK, which guarantees stability, and \( T_{w} \) can be constructed from this interpolant by solving the Sylvester equation:
\begin{align}
(S_w-L_w\zeta-\hat{W}^*BR_{\text{PR}}^{-1}\zeta)^*T_{w}-T_{w}S_w^*+\zeta^*R_{\text{PR}}^{-\frac{1}{2}}L_w^T=0.\label{Tw_pr}
\end{align}
An immediate consequence is that positive-realness can be enforced non-intrusively on the ROM produced by O-PORK, using only samples of \( G(s) \) at \( \mu_i \) and \( G(\infty) \).

Similar to NI-ADI-BT \citep{zulfiqar2025non}, the ADI-based non-intrusive approximation of PRBT (NI-ADI-PRBT) can be directly obtained using the results of Theorems \ref{Theorem3} and \ref{Theorem4}. Specifically, by computing Cholesky-like factorizations \(Q_v^{-1} = \hat{L}_p \hat{L}_p^*\) and \(P_w^{-1} = \hat{L}_q \hat{L}_q^*\), and defining the low-rank factors \(\tilde{L}_p = \hat{V} T_v\hat{L}_p\) and \(\tilde{L}_q = \hat{W} T_w\hat{L}_q\), the approximations admit the factorized representations \(P_{\mathrm{PR}}\approx \tilde{L}_p \tilde{L}_p^*\) and \(Q_{\mathrm{PR}}\approx\tilde{L}_q \tilde{L}_q^*\). Substituting the exact Gramian factors \(L_p, L_q\) in BSA with these approximations \(\tilde{L}_p, \tilde{L}_q\) yields NI-ADI-PRBT, which can be implemented using transfer function samples in the right half of the $s$-plane.
\subsection{Non-intrusive ADI-based Approximations of BRBT}\label{sec3_3}
By noting that $D^T(I - DD^T)^{-1} = (I - D^TD)^{-1}D^T$, define the following matrices: $R_p=I-DD^T$, $R_q=I-D^TD$, $R_b=I+D^TR_p^{-1}D$, and $R_c=I+DR_q^{-1}D^T$.
Further, define the transfer functions
\begin{align}
G_{\text{BR}}^{\text{P}}(s)&=C_{\text{BR}}^{\text{P}}(sI-A_{\text{BR}})^{-1}B_{\text{BR}}^{\text{P}},\nonumber\\
G_{\text{BR}}^{\text{Q}}(s)&=C_{\text{BR}}^{\text{Q}}(sI-A_{\text{BR}})^{-1}B_{\text{BR}}^{\text{Q}},\nonumber
\end{align}where
\begin{align}
A_{\text{BR}}&=A+BD^TR_p^{-1}C=A+BR_q^{-1}D^TC,\nonumber\\
B_{\text{BR}}^{\text{P}}&=BR_b^{\frac{1}{2}},\quad C_{\text{BR}}^{\text{P}}=R_p^{-\frac{1}{2}}C,\nonumber\\
B_{\text{BR}}^{\text{Q}}&=BR_q^{-\frac{1}{2}},\quad C_{\text{BR}}^{\text{Q}}=R_c^{\frac{1}{2}}C.\nonumber
\end{align}
Next, define the ROMs  
\begin{align}
\hat{G}_{\text{BR}}^{\text{P}}(s)&=\hat{C}_{\text{BR}}^{\text{P}}(sI-\hat{A}_{\text{BR}})^{-1}\hat{B}_{\text{BR}}^{\text{P}},\nonumber\\
\hat{G}_{\text{BR}}^{\text{Q}}(s)&=\hat{C}_{\text{BR}}^{\text{Q}}(sI-\hat{A}_{\text{BR}})^{-1}\hat{B}_{\text{BR}}^{\text{Q}},\nonumber
\end{align}obtained as
\begin{align}
\hat{A}_{\text{BR}}&=W_{\text{BR}}^*A_{\text{BR}}V_{\text{BR}}=\hat{A}+\hat{B}D^TR_p^{-1}\hat{C}=\hat{A}+\hat{B}R_q^{-1}D^T\hat{C},\nonumber\\
\hat{B}_{\text{BR}}^{\text{P}}&=W_{\text{BR}}^*B_{\text{BR}}^{\text{P}}=\hat{B}R_b^{\frac{1}{2}},\quad \hat{C}_{\text{BR}}^{\text{P}}=C_{\text{BR}}^{\text{P}}V_{\text{BR}}=R_p^{-\frac{1}{2}}\hat{C},\nonumber\\
\hat{B}_{\text{BR}}^{\text{Q}}&=W_{\text{BR}}^*B_{\text{BR}}^{\text{Q}}=\hat{B}R_q^{-\frac{1}{2}},\quad \hat{C}_{\text{BR}}^{\text{Q}}=C_{\text{BR}}^{\text{Q}}V_{\text{BR}}=R_c^{\frac{1}{2}}\hat{C},\nonumber
\end{align}where $W_{\text{BR}}^*V_{\text{BR}}=I$.

The Riccati equations (\ref{P_BR}) and (\ref{Q_BR}) can be rewritten as follows:  
\begin{align}
A_{\text{BR}}P_{\text{BR}}+P_{\text{BR}}A_{\text{BR}}^T+B_{\text{BR}}^{\text{P}}(B_{\text{BR}}^{\text{P}})^*+P_{\text{BR}}(C_{\text{BR}}^{\text{P}})^*C_{\text{BR}}^{\text{P}}P_{\text{BR}}=0,\\
A_{\text{BR}}^TQ_{\text{BR}}+Q_{\text{BR}}A_{\text{BR}}+(C_{\text{BR}}^{\text{Q}})^*C_{\text{BR}}^{\text{Q}}+Q_{\text{BR}}B_{\text{BR}}^{\text{Q}}(B_{\text{BR}}^{\text{Q}})^*Q_{\text{BR}}=0.
\end{align}
By applying Theorem \ref{Theorem1}, the RADI-based approximation of $P_{\text{BR}}$ can be obtained as follows. The projection matrix  
\begin{align}
V_{\text{BR}}=\begin{bmatrix}(\sigma_1I-A_{\text{BR}})^{-1}B_{\text{BR}}^{\text{P}}&\cdots&(\sigma_vI-A_{\text{BR}})^{-1}B_{\text{BR}}^{\text{P}}\end{bmatrix}
\end{align} solves the following Sylvester equation:  
\begin{align}
A_{\text{BR}}V_{\text{BR}}-V_{\text{BR}}S_v+B_{\text{BR}}^{\text{P}}L_v=0.
\end{align}
Then, $\hat{C}_{\text{BR}}^{\text{P}}=C_{\text{BR}}^{\text{P}}V_{\text{BR}}$ can be computed as follows:
\begin{align}
\hat{C}_{\text{BR}}^{\text{P}}=\begin{bmatrix}G_{\text{BR}}^{\text{P}}(\sigma_1)&\cdots&G_{\text{BR}}^{\text{P}}(\sigma_v)\end{bmatrix}.
\end{align}Thereafter, $Q_v$ is computed by solving the Lyapunov equation:
\begin{align}
-S_v^*Q_v-Q_vS_v+L_v^TL_v-(\hat{C}_{\text{BR}}^{\text{P}})^*\hat{C}_{\text{BR}}^{\text{P}}=0.\label{th5_sylv}
\end{align}
The RADI-based approximation of $P_{\text{BR}}$ is given by $V_{\text{BR}} Q_v^{-1} V_{\text{BR}}^*$.

According to Theorem \ref{Theorem1}, the ROM $(\hat{A}_{\text{BR}}=S_v-\hat{B}_{\text{BR}}^{\text{P}}L_v,\hat{B}_{\text{BR}}^{\text{P}}=Q_v^{-1}L_v^T,\hat{C}_{\text{BR}}^{\text{P}}=C_{\text{BR}}^{\text{P}}V_{\text{BR}})$ solves the following projected Riccati equation with $\hat{P}_{\text{BR}} = Q_v^{-1} > 0$:  
\begin{align}
\hat{A}_{\text{BR}}\hat{P}_{\text{BR}}+\hat{P}_{\text{BR}}\hat{A}_{\text{BR}}^*+\hat{B}_{\text{BR}}^{\text{P}}(\hat{B}_{\text{BR}}^{\text{P}})^*+\hat{P}_{\text{BR}}(\hat{C}_{\text{BR}}^{\text{P}})^*\hat{C}_{\text{BR}}^{\text{P}}\hat{P}_{\text{BR}}=0.\label{proj_ricc_Pbr}
\end{align}
Since
\begin{align}
\hat{A}&=S_v-Q_v^{-1}L_v^TL_v-\hat{B}D^TR_p^{-1}\hat{C},\nonumber\\
\hat{B}&=Q_v^{-1}L_v^TR_b^{-\frac{1}{2}},\quad\hat{C}=R_p^{\frac{1}{2}}C_{\text{BR}}^{\text{P}}V_{\text{BR}},\label{ss_BR}
\end{align}
the Riccati equation (\ref{proj_ricc_Pbr}) can be rewritten as  
\begin{align}
\hat{A}\hat{P}_{\text{BR}}+\hat{P}_{\text{BR}}\hat{A}^*+\hat{B}\hat{B}^*+(\hat{P}_{\text{BR}}\hat{C}^*+\hat{B}D^T)R_p^{-1}(\hat{P}_{\text{BR}}\hat{C}^*+\hat{B}D^T)^*=0.\nonumber
\end{align}
\begin{theorem}\label{Theorem5}
Let $G(s)$ be a bounded-real transfer function. Define $\tilde{L}_v=D^TR_p^{-\frac{1}{2}}\hat{C}_{\text{BR}}^{\text{P}}+R_b^{\frac{1}{2}}L_v=\begin{bmatrix}l_{v,1}&\cdots&l_{v,v}\end{bmatrix}$ and $T_{v}=\text{blkdiag}(l_{v,1},\cdots,l_{v,v})$. Assume that:
\begin{enumerate}
\item The interpolation points \((\sigma_1, \dots, \sigma_v)\) are located in the right half of the \( s \)-plane. 
  \item The pair $(S_v,\tilde{L}_v)$ is observable.
  \item The matrix $T_{v}$ is invertible.
  \item The realization $(\hat{A},\hat{B},\hat{C})$ defined in (\ref{ss_BR}) is minimal.
\end{enumerate}Then the following statements hold:
\begin{enumerate}
  \item\label{BRitem1} The projection matrix $V_\text{BR}$ is equal to $\hat{V}T_{v}$ where $\hat{V}$ is as in (\ref{Kry_V}).
  \item\label{BRitem2} The ROM $\hat{G}(s)=\hat{C}(sI-\hat{A})^{-1}\hat{B}+D$ interpolates $G(s)$ at $(\sigma_1,\cdots,\sigma_v)$.
  \item The ROM $\hat{G}(s)$ is bounded-real. 
\end{enumerate}
\end{theorem}
\begin{proof}
The proofs of \ref{BRitem1} and \ref{BRitem2} are similar to Theorem \ref{Theorem3}; hence, they are omitted for brevity.

By defining $K_{\text{BR}}=(\hat{P}_{\text{BR}}\hat{C}^*+\hat{B}D^T)R_p^{-\frac{1}{2}}$, the following relations hold:
\begin{align}
\hat{A}\hat{P}_{\text{BR}}+\hat{P}_{\text{BR}}\hat{A}^*=-\hat{B}\hat{B}^*-K_{\text{BR}}K_{\text{BR}}^*,\nonumber\\
\hat{P}_{\text{BR}}\hat{C}^*+\hat{B}D^T=-K_{\text{BR}}R_p^{\frac{1}{2}},\nonumber\\
R_p^{\frac{1}{2}}(R_p^{\frac{1}{2}})^*=I-DD^T.\nonumber
\end{align}
Since $\hat{P}_{\text{BR}} > 0$, $\hat{G}(s)$ satisfies the bounded-real lemma \citep{phillips2003guaranteed}.
\end{proof}
We now derive an expression for \( T_{v} \) when \( \hat{V} \) is used to enforce interpolation conditions instead of \( V_{\text{BR}} \). Since \( V_{\text{BR}} \) and \( \hat{V} \) enforce the same interpolation conditions, the following holds:
\begin{align}
\hat{C}_{\text{BR}}^{\text{P}}&=C_{\text{BR}}^{\text{P}}V_{\text{BR}}=\begin{bmatrix}G_{\text{BR}}^{\text{P}}(\sigma_1)&\cdots&G_{\text{BR}}^{\text{P}}(\sigma_v)\end{bmatrix}=\begin{bmatrix}\hat{G}_{\text{BR}}^{\text{P}}(\sigma_1)&\cdots&\hat{G}_{\text{BR}}^{\text{P}}(\sigma_v)\end{bmatrix}\nonumber\\
&=C_{\text{BR}}^{\text{P}}\hat{V}\begin{bsmallmatrix}(\sigma_1 I -S_v+\zeta L_v-\zeta D^TR_p^{-1}C\hat{V})^{-1}\zeta R_b^{\frac{1}{2}}&\cdots&(\sigma_v I -S_v+\zeta L_v-\zeta D^TR_p^{-1}C\hat{V})^{-1}\zeta R_b^{\frac{1}{2}}\end{bsmallmatrix}\nonumber\\
&=C_{\text{BR}}^{\text{P}}\hat{V}T_{v}=R_p^{-\frac{1}{2}}C\hat{V}T_{v},\nonumber
\end{align}where \( \hat{G}(s) = C \hat{V} (sI - S_v + \xi L_v)^{-1} \xi + D \) is an interpolant of \( G(s) \). Such an interpolant can be produced by I-PORK, which is guaranteed to be stable, and \( T_{v} \) can be constructed from this interpolant. Specifically, \( T_v \) is the solution to the following Sylvester equation:
\begin{align}
(S_v-\zeta L_v+\zeta D^TR_p^{-1}C\hat{V})T_{v}-T_{v}S_v+\zeta R_b^{\frac{1}{2}}L_v=0.\label{Tv_br}
\end{align}
An immediate consequence of this observation is that bounded-realness can be enforced non-intrusively on the ROM produced by I-PORK using only samples of \( G(s) \) at \( \sigma_i \) and the sample \( G(\infty) \).

As per Theorem \ref{Theorem2}, the RADI-based approximation of \( Q_{\text{BR}} \) can be obtained as follows. The projection matrix
\begin{align}
W_{\text{BR}}^*=\begin{bmatrix}C_{\text{BR}}^{\text{Q}}(\mu_1I-A_{\text{BR}})^{-1}\\\vdots\\C_{\text{BR}}^{\text{Q}}(\mu_wI-A_{\text{BR}})^{-1}\end{bmatrix}
\end{align} solves the following Sylvester equation:
\begin{align}
A_{\text{BR}}^TW_{\text{BR}}-W_{\text{BR}}S_w^*+(C_{\text{BR}}^{\text{Q}})^*L_w^T=0.
\end{align}
Then, \( \hat{B}_{\text{BR}}^{\text{Q}} = W_{\text{BR}}^* B_{\text{BR}}^{\text{Q}} \) can be computed as:
\begin{align}
\hat{B}_{\text{BR}}^{\text{Q}}=\begin{bmatrix}G_{\text{BR}}^{\text{Q}}(\mu_1)\\\vdots\\G_{\text{BR}}^{\text{Q}}(\mu_w)\end{bmatrix}.
\end{align}Thereafter, \( P_w \) is computed by solving the Lyapunov equation:
\begin{align}
-S_wP_w-P_wS_w^*+L_wL_w^T-\hat{B}_{\text{BR}}^{\text{Q}}(\hat{B}_{\text{BR}}^{\text{Q}})^*=0.\label{th6_sylv}
\end{align}
The RADI-based approximation of $Q_{\text{BR}}$ is given by $W_{\text{BR}} P_w^{-1} W_{\text{BR}}^*$.

According to Theorem \ref{Theorem2}, the ROM \( (\hat{A}_{\text{BR}} = S_w - L_w\hat{C}_{\text{BR}}^{\text{Q}}, \hat{B}_{\text{BR}}^{\text{Q}} = W_{\text{BR}}^*B_{\text{BR}}^{\text{Q}}, \hat{C}_{\text{BR}}^{\text{Q}} = L_w^TP_{w}^{-1}) \) solves the following projected Riccati equation with \( \hat{Q}_{\text{BR}} = P_w^{-1} > 0 \):
\begin{align}
\hat{A}_{\text{BR}}^*\hat{Q}_{\text{BR}}+\hat{Q}_{\text{BR}}\hat{A}_{\text{BR}}+(\hat{C}_{\text{BR}}^{\text{Q}})^*\hat{C}_{\text{BR}}^{\text{Q}}+\hat{Q}_{\text{BR}}\hat{B}_{\text{BR}}^{\text{Q}}(\hat{B}_{\text{BR}}^{\text{Q}})^*\hat{Q}_{\text{BR}}=0.\label{proj_ricc_Qbr}
\end{align}
Since
\begin{align}
\hat{A}&=S_w-L_wL_w^TP_w^{-1}-\hat{B}R_q^{-1}D^T\hat{C},\nonumber\\
\hat{B}&=\hat{W}_{\text{BR}}^*B_{\text{BR}}^{\text{Q}}R_q^{\frac{1}{2}},\quad\hat{C}=R_c^{-\frac{1}{2}}L_w^TP_w^{-1},\label{ss_BR_2}
\end{align}the Riccati equation (\ref{proj_ricc_Qbr}) can be rewritten as:
\begin{align}
\hat{A}^*\hat{Q}_{\text{BR}}+\hat{Q}_{\text{BR}}\hat{A}+\hat{C}^*\hat{C}+(\hat{B}^*\hat{Q}_{\text{BR}}+D^T\hat{C})^*R_q^{-1}(\hat{B}^*\hat{Q}_{\text{BR}}+D^T\hat{C})=0.
\end{align}
\begin{theorem}\label{Theorem6}
Let $G(s)$ be a bounded-real transfer function. Define $\tilde{L}_w^*=D^TR_p^{-\frac{1}{2}}\hat{B}_{\text{BR}}^{\text{Q}}+R_c^{\frac{1}{2}}L_w^T=\begin{bmatrix}l_{w,1}^*&\cdots&l_{w,w}^*\end{bmatrix}$ and $T_{w}=\text{blkdiag}(l_{w,1}^*,\cdots,l_{w,w}^*)$. Assume that:
\begin{enumerate}
\item The interpolation points $(\mu_1,\cdots,\mu_w)$ are located in the right half of the $s$-plane.
  \item The pair $(S_w,\tilde{L}_w)$ is controllable.
  \item The matrix $T_{w}$ is invertible.
  \item The realization $(\hat{A},\hat{B},\hat{C})$ defined in (\ref{ss_BR_2}) is minimal.
\end{enumerate}Then the following statements hold:
\begin{enumerate}
  \item The projection matrix $W_\text{BR}$ is equal to $\hat{W}T_{w}$ where $\hat{W}$ is as in (\ref{Kry_W}).
  \item The ROM $\hat{G}(s)=\hat{C}(sI-\hat{A})^{-1}\hat{B}+D$ interpolates $G(s)$ at $(\mu_1,\cdots,\mu_w)$.
  \item The ROM $\hat{G}(s)$ is bounded-real. 
\end{enumerate}
\end{theorem}
\begin{proof}
The proof is similar to that of Theorem \ref{Theorem5} and hence omitted for brevity.
\end{proof}
We now derive an expression for \( T_{w} \) when \(\hat{W}\) is used to enforce interpolation conditions instead of \( W_{\text{BR}} \). Since \( W_{\text{BR}} \) and \(\hat{W}\) enforce identical interpolation conditions, the following holds:
\begin{align}
\hat{B}_{\text{BR}}^\text{Q}&=W_{\text{BR}}^*B_{\text{BR}}^\text{Q}=\begin{bmatrix}G_{\text{BR}}^\text{Q}(\mu_1)\\\vdots\\G_{\text{BR}}^\text{Q}(\mu_w)\end{bmatrix}=\begin{bmatrix}\hat{G}_{\text{BR}}^\text{Q}(\mu_1)\\\vdots\\\hat{G}_{\text{BR}}^\text{Q}(\mu_w)\end{bmatrix}\nonumber\\
&=\begin{bmatrix}R_c^{\frac{1}{2}}\zeta(\mu_1I-S_w+L_w\zeta-\hat{W}^*BR_q^{-1}D^T\zeta)^{-1}\\\vdots\\R_c^{\frac{1}{2}}\zeta(\mu_wI-S_w+L_w\zeta-\hat{W}^*BR_q^{-1}D^T\zeta)^{-1}\end{bmatrix}\hat{W}^*B_{\text{BR}}^\text{Q},\nonumber\\
&=T_{w}^*\hat{W}^*BR_{q}^{-\frac{1}{2}},\nonumber
\end{align}
where \(\hat{G}(s) = \zeta (sI - S_w + L_w \zeta)^{-1} \hat{W}^* B + D\) is an interpolant of \( G(s) \). This interpolant can be produced by O-PORK, which is guaranteed to be stable, and \( T_{w} \) can be constructed from it. Specifically, \( T_w \) is obtained by solving the Sylvester equation:
\begin{align}
(S_w-L_w\zeta+\hat{W}^*BR_q^{-1}D^T\zeta)^*T_{w}-T_{w}S_w^*+\zeta^*R_c^{\frac{1}{2}}L_w^T=0.\label{Tw_br}
\end{align}
An immediate consequence is that bounded-realness can be enforced non-intrusively on the ROM produced by O-PORK, using only samples of \( G(s) \) at \( \mu_i \) and \( G(\infty) \).

Similar to NI-ADI-BT \citep{zulfiqar2025non}, the ADI-based non-intrusive approximation of BRBT (NI-ADI-BRBT) can be directly obtained using the results of Theorems \ref{Theorem5} and \ref{Theorem6}. Specifically, by computing Cholesky-like factorizations \(Q_v^{-1} = \hat{L}_p \hat{L}_p^*\) and \(P_w^{-1} = \hat{L}_q \hat{L}_q^*\), and defining the low-rank factors \(\tilde{L}_p = \hat{V} T_v\hat{L}_p\) and \(\tilde{L}_q = \hat{W} T_w\hat{L}_q\), the approximations admit the factorized representations \(P_{\mathrm{BR}}\approx \tilde{L}_p \tilde{L}_p^*\) and \(Q_{\mathrm{BR}}\approx\tilde{L}_q \tilde{L}_q^*\). Substituting the exact Gramian factors \(L_p, L_q\) in BSA with these approximations \(\tilde{L}_p, \tilde{L}_q\) yields NI-ADI-BRBT, which can be implemented using transfer function samples in the right half of the $s$-plane.
\subsection{Non-intrusive ADI-based Approximation of SWBT}\label{sec3_4}
Let us define the transfer function
\begin{align}
G_{\text{SW}}(s)=C_{\text{SW}}(sI-A_{\text{SW}})^{-1}B,\nonumber
\end{align} where $A_{\text{SW}}=A-BC_{\text{SW}}$ and $C_{\text{SW}}=D^{-1}C$. Next, define the ROM $\hat{G}_{\text{SW}}(s)$
\begin{align}
\hat{G}_{\text{SW}}(s)=\hat{C}_{\text{SW}}(sI-\hat{A}_{\text{SW}})^{-1}\hat{B},\nonumber
\end{align} obtained as
\begin{align}
\hat{A}_{\text{SW}}&=W_{\text{SW}}^*A_{\text{SW}}V_{\text{SW}}=\hat{A}-\hat{B}D^{-1}\hat{C},\nonumber\\
\hat{B}&=W_{\text{SW}}^*B,\quad \hat{C}_{\text{SW}}=C_{\text{SW}}V_{\text{SW}}=D^{-1}\hat{C},\nonumber
\end{align}where $W_{\text{SW}}^*V_{\text{SW}}=I$.

The Lyapunov equation (\ref{Q_SW}) can be rewritten as:
\begin{align}
A_{\text{SW}}^TQ_{\text{SW}}+Q_{\text{SW}}A_{\text{SW}}+C_{\text{SW}}^TC_{\text{SW}}=0.
\end{align}
The O-PORK-based approximation of $Q_{\text{SW}}$ (equivalent to the ADI method) is computed as follows. First, solve the Sylvester equation for the projection matrix $W_{\text{SW}}$:
\begin{align}
A_{\text{SW}}^TW_{\text{SW}}-W_{\text{SW}}S_w^*+C_{\text{SW}}^TL_w^T=0.
\end{align}
Then, compute $P_w$ by solving the Lyapunov equation (\ref{pork_ps}). The ADI-based approximation of $Q_{\text{SW}}$ is obtained as: $Q_{\text{S}}\approx W_{\text{SW}}P_w^{-1}W_{\text{SW}}^*$. The following ROM $\hat{G}(s)$ can be recovered from the ROM $\hat{G}_{\text{SW}}(s)=L_w^TP_{w}^{-1}(sI-S_w+L_wL_w^TP_w^{-1})^{-1}W_{\text{SW}}^*B$ produced by O-PORK:
\begin{align}
\hat{A}&=S_w+L_wL_w^TP_w^{-1}+W_{\text{S}}^*BL_w^TP_w^{-1},\nonumber\\
\hat{B}&=W_{\text{SW}}^*B,\quad \hat{C}=DL_w^TP_w^{-1}.\label{ss_SW}
\end{align}
\begin{theorem}\label{Theorem7}
Let $G(s)$ be a square stable minimum phase transfer function. Define $\tilde{L}_w^*=D^{-T}(L_w^T-\hat{B}^*)=\begin{bmatrix}l_{w,1}^*&\cdots&l_{w,w}^*\end{bmatrix}$ and $T_{w}=\text{blkdiag}(l_{w,1}^*,\cdots,l_{w,w}^*)$. Assume that:
\begin{enumerate}
\item The interpolation points $(\mu_1,\cdots,\mu_w)$ are located in the right half of the $s$-plane.
  \item The pair $(S_w,\tilde{L}_w)$ is controllable.
  \item The matrix $T_{w}$ is invertible.
  \item The realization $(\hat{A},\hat{B},\hat{C})$ defined in (\ref{ss_SW}) is minimal.
\end{enumerate}Then the following statements hold:
\begin{enumerate}
  \item\label{SWitem1} The projection matrix $W_\text{SW}$ is equal to $\hat{W}T_{w}$ where $\hat{W}$ is as in (\ref{Kry_W}).
  \item\label{SWitem2} The ROM $\hat{G}(s)=\hat{C}(sI-\hat{A})^{-1}\hat{B}+D$ interpolates $G(s)$ at $(\mu_1,\cdots,\mu_w)$.
  \item The ROM $\hat{G}(s)$ is minimum phase. 
\end{enumerate}
\end{theorem}
\begin{proof}
The proof of \ref{SWitem1} and \ref{SWitem2} is similar to that of Theorem \ref{Theorem4} and hence omitted for brevity.

The zeros of $\hat{G}(s)$, i.e., the eigenvalues of $\hat{A}_{\text{SW}}=\hat{A}-\hat{B}D^{-1}\hat{C}$, are $(-\mu_1^*,\cdots,-\mu_w^*)$ in O-PORK. Since the interpolation points in O-PORK are in the right half of the $s$-plane, $\hat{A}_{\text{SW}}$ is Hurwitz and $\hat{G}(s)$ is minimum phase.  
\end{proof}
We now derive an expression for $T_{w}$ when $\hat{W}$ is used for enforcing interpolation conditions instead of $W_{\text{SW}}$. Since $W_{\text{SW}}$ and $\hat{W}$ enforce the same interpolation conditions, the following holds:  
\begin{align}
\hat{B}&=W_{\text{SW}}^*B=\begin{bmatrix}G_{\text{SW}}(\mu_1)\\\vdots\\G_{\text{SW}}(\mu_w)\end{bmatrix}=\begin{bmatrix}\hat{G}_{\text{SW}}(\mu_1)\\\vdots\\\hat{G}_{\text{SW}}(\mu_w)\end{bmatrix}\nonumber\\
&=\begin{bmatrix}D^{-1}\zeta(\mu_1I-S_w+L_w\zeta+\hat{W}^*BD^{-1}\zeta)^{-1}\\\vdots\\D^{-1}\zeta(\mu_wI-S_w+L_w\zeta+\hat{W}^*BD^{-1}\zeta)^{-1}\end{bmatrix}\hat{W}^*B,\nonumber\\
&=T_{w}^*\hat{W}^*B,\nonumber
\end{align}
where $\hat{G}(s)=\zeta(sI-S_w+L_w\zeta)^{-1}\hat{W}^*B+D$ is the interpolant of $G(s)$. Such an interpolant can be produced by O-PORK, which is guaranteed to be stable, and $T_{w}$ can be constructed from this interpolant. $T_w$ can be computed by solving the following Sylvester equation:  
\begin{align}
\big(S_w-L_w\zeta-\hat{W}^*BD^{-1}\zeta\big)^*T_{w}-T_{w}S_w^*+\zeta^*D^{-T}L_w^T=0.\label{Tw_sw}
\end{align}
An immediate consequence of this observation is that the minimum-phase property can be enforced on the ROM nonintrusively produced by O-PORK using only samples of $G(s)$ at $\mu_i$ and the sample $G(\infty)$.

Similar to NI-ADI-BT \citep{zulfiqar2025non}, the ADI-based non-intrusive approximation of SWBT (NI-ADI-SWBT) can be directly obtained using the results of Theorem \ref{Theorem7}. Specifically, by computing Cholesky-like factorizations \(Q_v^{-1} = \hat{L}_p \hat{L}_p^*\) and \(P_w^{-1} = \hat{L}_q \hat{L}_q^*\), and defining the low-rank factors \(\tilde{L}_p = \hat{V}\hat{L}_p\) and \(\tilde{L}_q = \hat{W} T_w\hat{L}_q\), the approximations admit the factorized representations \(P\approx\tilde{L}_p \tilde{L}_p^*\) and \(Q_{\mathrm{SW}}\approx\tilde{L}_q \tilde{L}_q^*\). Substituting the exact Gramian factors \(L_p, L_q\) in BSA with these approximations \(\tilde{L}_p, \tilde{L}_q\) yields NI-ADI-SWBT, which can be implemented using transfer function samples in the right half of the $s$-plane.
\subsection{Non-intrusive ADI-based Approximation of BST}\label{sec3_5}
Let us define the transfer function
\begin{align}
G_{\text{S}}(s)=C_{\text{S}}(sI-A_{\text{S}})^{-1}B_{\text{S}},\nonumber
\end{align}where 
\begin{align}
A_{\text{S}}&=A-B_{\text{S}}C_{\text{S}},& B_{\text{S}}&=(PC^T+BD^T)R_\text{S}^{-\frac{1}{2}},\nonumber\\
C_{\text{S}}&=R_\text{S}^{-\frac{1}{2}}C,& R_\text{S}&=DD^T.\nonumber
\end{align}
The Riccati equation (\ref{Q_S}) can be rewritten as follows:
\begin{align}
A_{\text{S}}^TQ_{\text{S}}+Q_{\text{S}}A_{\text{S}}+C_{\text{S}}^TC_{\text{S}}+Q_{\text{S}}B_{\text{S}}B_{\text{S}}^TQ_{\text{S}}=0.
\end{align}
Unlike the Riccati equations considered up until now, the RADI-based approximation of this Riccati equation is a bit more involved. We first need to replace $B_{\text{S}}$ with its approximation. Let us replace $P$ in $B_{\text{S}}$ with $\hat{V}Q_v^{-1}\hat{V}^*$, which is an ADI-based approximation of $P$ obtained via I-PORK. Then, $B_{\text{S}}\approx\tilde{B}_{\text{S}}=(\hat{V}Q_v^{-1}\hat{V}^*C^T+BD^T)R_\text{S}^{-\frac{1}{2}}$ and $A_{\text{S}}\approx\tilde{A}_{\text{S}}=A-\tilde{B}_{\text{S}}C_{\text{S}}$. Thereafter, we proceed with the RADI-based approximation of the following Riccati equation:
\begin{align}
\tilde{A}_{\text{S}}^*\tilde{Q}_{\text{S}}+\tilde{Q}_{\text{S}}\tilde{A}_{\text{S}}+C_{\text{S}}^*C_{\text{S}}+\tilde{Q}_{\text{S}}\tilde{B}_{\text{S}}\tilde{B}_{\text{S}}^*\tilde{Q}_{\text{S}}=0.\label{QtS}
\end{align}The RADI-based approximation of $\tilde{Q}_{\text{S}}$ can be obtained using Theorem \ref{Theorem2} as follows. The projection matrix $W_\text{S}$ is computed by solving the following Sylvester equation:  
\begin{align}
\tilde{A}_{\text{S}}^*W_\text{S}-W_\text{S}S_w^*+C_{\text{S}}^TL_w^T=0.
\end{align}
The matrix $P_w$ is computed by solving the following Sylvester equation:  
\begin{align}
-S_wP_w-P_wS_w^*+L_wL_w^T-\hat{B}_\text{S}\hat{B}_\text{S}^*=0,\label{bst_sylv}
\end{align}where $\hat{B}_\text{S}=W_\text{S}^*\tilde{B}_\text{S}$. The RADI-based approximation of $\tilde{Q}_{\text{S}}$ is given by $W_\text{S} P_w^{-1} W_\text{S}^*$.

Once again, it can be shown that $W_\text{S}$ and $\hat{W}$ in (\ref{Kry_W}) enforce the same interpolation conditions and $W_\text{S} = \hat{W} T_{w}$, where $T_{w}$ is invertible. The proof is similar and hence omitted here for brevity.

Let us define the ROM
\begin{align}
\hat{G}_{\text{S}}(s)=\hat{C}_{\text{S}}(sI-\hat{A}_{\text{S}})^{-1}\hat{B}_{\text{S}},\nonumber
\end{align} obtained as
\begin{align}
\hat{A}_{\text{S}}=W_{\text{S}}^*\tilde{A}_{\text{S}}V_{\text{S}},\quad \hat{B}_{\text{S}}=W_{\text{S}}^*\tilde{B}_{\text{S}},\quad \hat{C}_{\text{S}}=C_{\text{S}}V_{\text{S}},\nonumber
\end{align}where $W_{\text{S}}^*V_{\text{S}}=I$.

We now derive an expression for $T_{w}$ when $\hat{W}$ is used for enforcing interpolation conditions instead of $W_{\text{S}}$. Since $W_{\text{S}}$ and $\hat{W}$ enforce the same interpolation conditions, the following holds: 
\begin{align}
\hat{B}_{\text{S}}&=W_{\text{S}}^*\tilde{B}_{\text{S}}=\begin{bmatrix}G_{\text{S}}(\mu_1)\\\vdots\\G_{\text{S}}(\mu_w)\end{bmatrix}=\begin{bmatrix}\hat{G}_{\text{S}}(\mu_1)\\\vdots\\\hat{G}_{\text{S}}(\mu_w)\end{bmatrix}\nonumber\\
&=\begin{bmatrix}R_{\text{S}}^{-\frac{1}{2}}\zeta\big(\mu_1I-S_w+L_w\zeta+(\hat{W}^*\hat{V}Q_v^{-1}\hat{V}^*C^T+\hat{W}^*BD^T)R_{\text{S}}^{-1}\zeta\big)^{-1}\\\vdots\\R_{\text{S}}^{-\frac{1}{2}}\zeta\big(\mu_wI-S_w+L_w\zeta+(\hat{W}^*\hat{V}Q_v^{-1}\hat{V}^*C^T+\hat{W}^*BD^T)R_{\text{S}}^{-1}\zeta\big)^{-1}\end{bmatrix}\hat{W}^*\tilde{B}_{\text{S}},\nonumber\\
&=T_{w}^*\hat{W}^*\tilde{B}_{\text{S}},\nonumber\\
&=T_{w}^*\big((\hat{W}^*\hat{V})Q_v^{-1}(C\hat{V})^*+(\hat{W}^*B)D^T\big)R_{\text{S}}^{-\frac{1}{2}},\nonumber
\end{align}
where $\hat{G}(s) = \zeta (sI - S_w + L_w \zeta)^{-1} \hat{W}^* B + D$ is the interpolant of $G(s)$. Such an interpolant can be produced by O-PORK, which is guaranteed to be stable. $T_{w}$ can be constructed from this interpolant by solving the following Sylvester equation:  
\begin{align}
\big(S_w-L_w\zeta-[(\hat{W}^*\hat{V})Q_v^{-1}(C\hat{V})^*+(\hat{W}^*B)D^T]&R_{\text{S}}^{-1}\zeta\big)^*T_{w}\nonumber\\
&-T_{w}S_w^*+\zeta^*R_{\text{S}}^{-\frac{1}{2}}L_w^T=0.\label{Tw_bst}
\end{align}

Similar to NI-ADI-BT \citep{zulfiqar2025non}, the ADI-based non-intrusive approximation of BST (NI-ADI-BST) can be obtained readily. Specifically, by computing Cholesky-like factorizations \(Q_v^{-1} = \hat{L}_p \hat{L}_p^*\) and \(P_w^{-1} = \hat{L}_q \hat{L}_q^*\), and defining the low-rank factors \(\tilde{L}_p = \hat{V} \hat{L}_p\) and \(\tilde{L}_q = \hat{W} T_w\hat{L}_q\), the approximations admit the factorized representations \(P\approx\tilde{L}_p \tilde{L}_p^*\) and \(Q_{\mathrm{S}}\approx\tilde{L}_q \tilde{L}_q^*\). Substituting the exact Gramian factors \(L_p, L_q\) in BSA with these approximations \(\tilde{L}_p, \tilde{L}_q\) yields NI-ADI-BST, which can be implemented using transfer function samples in the right half of the $s$-plane.
\subsection{Numerical Issues and Their Solutions}\label{sec3_6}
The ADI-based non-intrusive approximations presented in this section involve the inverses \( Q_v^{-1} \) and \( P_w^{-1} \), which can cause numerical problems as the number of samples used for implementation increases. Specifically, as the number of samples increases, the matrices \( Q_v \) and \( P_w \) start to lose numerical rank. Note that these matrices do not appear explicitly in the original low-rank ADI algorithms. This section addresses this issue by providing an appropriate choice of interpolation points that guarantee the invertibility of these matrices.

Proposition \ref{prop_A_diag} establishes that when $\sigma_i=\epsilon+j\omega_i$, the free parameter $\zeta$ in the interpolant $\hat{G}(s)=C\hat{V}(sI-S_v+\zeta L_v)^{-1}\zeta+D$ can be set to enforce a modal structure on $\hat{A}=S_v-\zeta L_v$ such that it has poles at $-\epsilon+j\omega_i$. The motivation for enforcing this structure is to benefit from the interesting properties of modal state-space realizations with lightly damped poles, wherein the modes are effectively decoupled from each other leading to block-diagonally dominant solutions to several matrix equations appearing in MOR and control design; see \citep{jonckheere1984principal,gawronski2004dynamics,gawronski2006balanced} for details.
\begin{proposition}\label{prop_A_diag}
Let $\epsilon > 0$ and let distinct frequencies $\omega_1, \dots, \omega_v \in \mathbb{R}$ be given such that $\omega_i \neq 0$ for all $i=1,\dots,v$. Define the minimum frequency magnitude as $\omega_{\min} := \min_{i} |\omega_i| > 0$. Consider the matrices
\begin{align*}
    S_v &= \operatorname{diag}(\epsilon+j\omega_1, \dots, \epsilon+j\omega_v) \otimes I_m,& L_v &= \mathbf{1}_v^T \otimes I_m, \\
    \zeta &= (2\epsilon \mathbf{1}_v) \otimes I_m,& \hat{A}_d &= \operatorname{diag}(- \epsilon+j\omega_1, \dots, - \epsilon+j\omega_v) \otimes I_m,
\end{align*}
where $\mathbf{1}_v \in \mathbb{R}^v$ is the column vector of ones. Let us decompose $\hat{A}=S_v - \zeta L_v$ into its block-diagonal part $\hat{A}_d$ and off-diagonal remainder $\hat{E}$ such that $\hat{A} = \hat{A}_d + \hat{E}$. Then, the following properties hold:
\begin{enumerate}
    \item For any $v \ge 1$, the spectral norm of the off-diagonal approximation error is exactly
    \begin{equation}
        \|\hat{A} - \hat{A}_d\|_2 = 2\epsilon(v-1).
    \end{equation}
    \item For any $v \ge 2$, $\hat{A}$ is strictly block diagonally dominant with respect to the spectral norm provided that
    \begin{equation}
        \epsilon < \frac{\omega_{\min}}{\sqrt{4(v-1)^2 - 1}}.
    \end{equation}
    \item For any desired error tolerance $\delta > 0$ and $v \ge 2$, choosing
    \begin{equation}
        \label{eq:epsilon_bound}
        \epsilon \le \min \left( \frac{\delta}{2(v-1)}, \; \frac{\omega_{\min}}{\sqrt{4(v-1)^2 - 1}} \right)
    \end{equation}
    guarantees that $\|\hat{A} - \hat{A}_d\|_2 \le \delta$ and that $\hat{A}$ is strictly block diagonally dominant.
\end{enumerate}
\end{proposition}
\begin{proof}
The proof is given in Appendix A.
\end{proof}
The next proposition establishes the block diagonal dominance of $Q_v$ in NI-ADI-LQGBT, which solves the Lyapunov equation \eqref{Qv_LQG}.
\begin{proposition}
\label{prop:Qv_dominance}
Let Theorem \ref{Theorem1} and Proposition \ref{prop_A_diag} hold. Define the matrix $M := L_v^* L_v + \hat{C}^* \hat{C}$ and denote its $(i,j)$-th block entry by $M_{ij} \in \mathbb{C}^{m \times m}$. Then:
\begin{enumerate}
    \item The block entries $[Q_v]_{ij}$ of $Q_v$ are given by
    \begin{equation}
        \label{eq:Q_block_sol}
        [Q_v]_{ij} = \frac{M_{ij}}{2\epsilon + j(\omega_j - \omega_i)}.
    \end{equation}
    \item Since $L_v^* L_v$ has identity blocks on the diagonal, $\|M_{ii}\|_2 \ge 1$ for all $i$. There exists an $\epsilon^* > 0$ such that for all $0 < \epsilon < \epsilon^*$, $Q_v$ is strictly block diagonally dominant with respect to the spectral norm. Specifically, dominance holds if
    \begin{equation}
        \label{eq:epsilon_dominance_Q}
        \epsilon < \min_{i} \frac{\|M_{ii}\|_2}{2 \sum_{k \neq i} \frac{\|M_{ik}\|_2}{|\omega_k - \omega_i|}}.
    \end{equation}
    \item As $\epsilon \to 0$, the off-diagonal blocks converge to finite values determined by the frequency separation:
    \begin{equation}
        \label{eq:Q_offdiag_limit}
        \lim_{\epsilon \to 0} [Q_v]_{ij} = \frac{M_{ij}}{j(\omega_j - \omega_i)}, \quad \text{for } i \neq j.
    \end{equation}
    Conversely, the diagonal blocks diverge as $\| [Q_v]_{ii} \|_2 \sim O(\epsilon^{-1})$.
\end{enumerate}
\end{proposition}
\begin{proof}
The proof is given in Appendix B.
\end{proof}
The next proposition establishes the block-diagonal dominance of the matrix $T_v$ appearing in NI-ADI-PRBT and NI-ADI-BRBT.
\begin{proposition}
\label{thm:Tv_structure}
Let Proposition \ref{prop_A_diag} hold. Let $\mathcal{P} \in \mathbb{C}^{m \times vm}$ and $\mathcal{Q} \in \mathbb{C}^{m \times m}$ be matrices independent of $\epsilon$. Decompose $\mathcal{P}$ into block columns $\mathcal{P} = [\mathcal{P}_1, \dots, \mathcal{P}_v]$ with $\mathcal{P}_k \in \mathbb{C}^{m \times m}$. Assume that $(I_m + \mathcal{P}_i)$ is invertible for all $i=1,\dots,v$ and that $\mathcal{Q} \neq 0$. Let $T_v \in \mathbb{C}^{vm \times vm}$ be the unique solution to the Sylvester equation:
\begin{equation}
    \label{eq:Tv_sylvester}
    (S_v - \zeta L_v - \zeta \mathcal{P}) T_v - T_v S_v + \zeta \mathcal{Q} L_v = 0.
\end{equation}
Denote the $(i,j)$-th block of $T_v$ by $[T_v]_{ij} \in \mathbb{C}^{m \times m}$. Then:
\begin{enumerate}
    \item As $\epsilon \to 0$, the diagonal blocks converge to a finite limit, while the off-diagonal blocks vanish linearly with $\epsilon$:
    \begin{align}
        \lim_{\epsilon \to 0} [T_v]_{ii} &= (I_m + \mathcal{P}_i)^{-1} \mathcal{Q}, \label{eq:Tv_diag_limit} \\
        \| [T_v]_{ij} \|_2 &= O(\epsilon), \quad \text{for } i \neq j. \label{eq:Tv_offdiag_limit}
    \end{align}
    \item There exists an $\epsilon^* > 0$ such that for all $0 < \epsilon < \epsilon^*$, $T_v$ is strictly block diagonally dominant with respect to the spectral norm. A sufficient condition for dominance is
    \begin{equation}
        \label{eq:epsilon_dominance_T}
        \epsilon < \min_{i} \frac{\| (I_m + \mathcal{P}_i)^{-1} \mathcal{Q} \|_2}{2 \|\mathcal{Q}\|_2 \sum_{k \neq i} \frac{1}{|\omega_k - \omega_i|}}.
    \end{equation}
    \item The error between $T_v$ and its block-diagonal part $\operatorname{blkdiag}([T_v]_{11}, \dots, [T_v]_{vv})$ satisfies
    \begin{equation}
        \| T_v - \operatorname{blkdiag}(T_v) \|_2 \le K \epsilon,
    \end{equation}
    for some constant $K > 0$ independent of $\epsilon$.
\end{enumerate}
\end{proposition}
\begin{proof}
The proof is given in Appendix C.
\end{proof}
Let $\sigma_i=\epsilon+j\omega_i$ and $\mu_i=\epsilon+j\nu_i$, where $\epsilon$ is determined by Propositions \ref{prop_A_diag}, \ref{prop:Qv_dominance}, and \ref{thm:Tv_structure} to guarantee block diagonal dominance. This selection of interpolation points effectively resolves the invertibility challenges in the proposed ADI-based non-intrusive implementations, ensuring that the matrices $Q_v$, $P_w$, $T_v$, and $T_w$ remain invertible.

If the number of samples of the transfer function to be processed is significantly high such that the Lyapunov and Sylvester equations involved in the proposed non-intrusive implementations become computationally expensive to solve, then one can exploit the block diagonal dominance properties and use the block diagonal approximations instead. As shown in \citep{jonckheere1984principal,gawronski2004dynamics,gawronski2006balanced}, the decoupling of modes in lightly damped modal state-space realizations makes the block diagonal approximations quite effective. By exploiting the block diagonal dominance properties, the non-intrusive ADI-based implementations presented in this section can be implemented as follows.

For NI-ADI-LQGBT:
\begin{align}
Q_v^{-1}&\approx\text{blkdiag}\big(2\epsilon[I_m+H^*(\epsilon+j\omega_1)H(\epsilon+j\omega_1)]^{-1},\cdots,\nonumber\\
&\hspace*{4cm}2\epsilon[I_m+H^*(\epsilon+j\omega_v)H(\epsilon+j\omega_v)]^{-1}\big),\\
P_w^{-1}&\approx\text{blkdiag}\big(2\epsilon[I_p+H(\epsilon+j\nu_1)H^*(\epsilon+j\nu_1)]^{-1},\cdots,\nonumber\\
&\hspace*{4cm}2\epsilon[I_p+H(\epsilon+j\nu_w)H^*(\epsilon+j\nu_w)]^{-1}\big).
\end{align}

For NI-ADI-$\mathcal{H}_\infty$BT:
\begin{align}
Q_v^{-1}&\approx\text{blkdiag}\big(2\epsilon[I_m+(1-\gamma^{-2})H^*(\epsilon+j\omega_1)H(\epsilon+j\omega_1)]^{-1},\cdots,\nonumber\\
&\hspace*{3cm}2\epsilon[I_m+(1-\gamma^{-2})H^*(\epsilon+j\omega_v)H(\epsilon+j\omega_v)]^{-1}\big),\\
P_w^{-1}&\approx\text{blkdiag}\big(2\epsilon[I_p+(1-\gamma^{-2})H(\epsilon+j\nu_1)H^*(\epsilon+j\nu_1)]^{-1},\cdots,\nonumber\\
&\hspace*{3cm}2\epsilon[I_p+(1-\gamma^{-2})H(\epsilon+j\nu_w)H^*(\epsilon+j\nu_w)]^{-1}\big).
\end{align}

For NI-ADI-PRBT:
\begin{align}
Q_v^{-1}&\approx\text{blkdiag}\big(2\epsilon[I_m-t_{v,1}^*H^*(\epsilon+j\omega_1)R_{\text{PR}}^{-1}H(\epsilon+j\omega_1)t_{v,1}]^{-1},\cdots,\nonumber\\
&\hspace*{2cm}2\epsilon[I_m-t_{v,v}^*H^*(\epsilon+j\omega_v)R_{\text{PR}}^{-1}H(\epsilon+j\omega_v)t_{v,v}]^{-1}\big),\\
P_w^{-1}&\approx\text{blkdiag}\big(2\epsilon[I_p-t_{w,1}^*H(\epsilon+j\nu_1)R_{\text{PR}}^{-1}H^*(\epsilon+j\nu_1)t_{w,1}]^{-1},\cdots,\nonumber\\
&\hspace*{2cm}2\epsilon[I_p-t_{w,w}^*H(\epsilon+j\nu_w)R_{\text{PR}}^{-1}H^*(\epsilon+j\nu_w)t_{w,w}]^{-1}\big),\\
T_v&\approx\text{blkdiag}\big(t_{v,1},\cdots,t_{v,v}\big),\\
T_w&\approx\text{blkdiag}\big(t_{w,1},\cdots,t_{w,w}\big),
\end{align}where
\begin{align}
t_{v,i}&=[I_m+R_{\text{PR}}^{-1}H(\epsilon+j\omega_i)]^{-1}R_{\text{PR}}^{-\frac{1}{2}},\\
t_{w,i}&=[I_p+R_{\text{PR}}^{-1}H^*(\epsilon+j\nu_i)]^{-1}R_{\text{PR}}^{-\frac{1}{2}}.
\end{align}

For NI-ADI-BRBT:
\begin{align}
Q_v^{-1}&\approx\text{blkdiag}\big(2\epsilon[I_m-t_{v,1}^*H^*(\epsilon+j\omega_1)R_p^{-1}H(\epsilon+j\omega_1)t_{v,1}]^{-1},\cdots,\nonumber\\
&\hspace*{2cm}2\epsilon[I_m-t_{v,v}^*H^*(\epsilon+j\omega_v)R_p^{-1}H(\epsilon+j\omega_v)t_{v,v}]^{-1}\big),\\
P_w^{-1}&\approx\text{blkdiag}\big(2\epsilon[I_p-t_{w,1}^*H(\epsilon+j\nu_1)R_q^{-1}H^*(\epsilon+j\nu_1)t_{w,1}]^{-1},\cdots,\nonumber\\
&\hspace*{2cm}2\epsilon[I_p-t_{w,w}^*H(\epsilon+j\nu_w)R_q^{-1}H^*(\epsilon+j\nu_w)t_{w,w}]^{-1}\big),\\
T_v&\approx\text{blkdiag}\big(t_{v,1},\cdots,t_{v,v}\big),\\
T_w&\approx\text{blkdiag}\big(t_{w,1},\cdots,t_{w,w}\big),
\end{align}where
\begin{align}
t_{v,i}&=[I_m-D^TR_p^{-1}H(\epsilon+j\omega_i)]^{-1}R_b^{\frac{1}{2}},\\
t_{w,i}&=[I_p-DR_q^{-1}H^*(\epsilon+j\nu_i)]^{-1}R_c^{\frac{1}{2}}.
\end{align}

For NI-ADI-SWBT:
\begin{align}
Q_v^{-1}&\approx\text{blkdiag}\big(2\epsilon I_m,\cdots,2\epsilon I_m\big),\\
P_w^{-1}&\approx\text{blkdiag}\big(2\epsilon I_p,\cdots,2\epsilon I_p\big),\\
T_w&\approx\text{blkdiag}\big([G(\epsilon+j\nu_1)]^{-*},\cdots,[G(\epsilon+j\nu_w)]^{-*}\big).
\end{align}

For NI-ADI-BST:
\begin{align}
Q_v^{-1}&\approx\text{blkdiag}\big(2\epsilon I_m,\cdots,2\epsilon I_m\big),\\
P_w^{-1}&\approx\text{blkdiag}\big(2\epsilon[I_p-t_{w,1}^*\mathscr{H}(\epsilon+j\nu_1)R_{\text{S}}^{-1}\mathscr{H}^*(\epsilon+j\nu_1)t_{w,1}]^{-1},\cdots,\nonumber\\
&\hspace*{2cm}2\epsilon[I_p-t_{w,w}^*\mathscr{H}(\epsilon+j\nu_w)R_{\text{S}}^{-1}\mathscr{H}^*(\epsilon+j\nu_w)t_{w,w}]^{-1}\big),\\
T_w&\approx\text{blkdiag}\big(t_{w,1},\cdots,t_{w,w}\big),
\end{align}
where
\begin{align}
\mathscr{H}(\epsilon+j\nu_i) &= \begin{cases}
\frac{2\epsilon}{j\nu_i-j\omega_i}\big[H(\epsilon+j\omega_i)-H(\epsilon+j\nu_i)\big]H^*(\epsilon+j\omega_i)\nonumber\\
\hspace*{2cm}+H(\epsilon+j\nu_i)D^T & \text{if } \omega_i\neq\nu_i \\
2\epsilon H^{\prime}(\epsilon+j\omega_i)H^*(\epsilon+j\omega_i)+H(\epsilon+j\nu_i)D^T  & \text{if } \omega_i=\nu_i
\end{cases},\\
t_{w,i}&=\Big[R_\text{S}+\mathscr{H}^*(\epsilon+j\nu_i)\Big]^{-1}R_{\text{S}}^\frac{1}{2}.
\end{align}
Note that if these block diagonally dominant approximations are used, the theoretical guarantees on positive-realness, bounded-realness, and the minimum-phase property proved in Theorems \ref{Theorem3}–\ref{Theorem7} are lost, even if, in practice, the interpolant retains these properties.
\begin{remark}
Accuracy in data-driven interpolation depends on factors like the choice of interpolation points, placement of the poles of the interpolant, and noisy data \citep{aumann2025practical,antoulas2017tutorial,karachalios2021loewner}. The approximate non-intrusive BT algorithms proposed here and in the next section inherit all these practical challenges from the Loewner framework. Although the final ROM in a non-intrusive BT variant does not have interpolatory properties in general, it is essentially a compressed version of the Loewner quadruplet $(\hat{W}^*\hat{V},\hat{W}^*A\hat{V},\hat{W}^*B,C\hat{V})$. Thus, it can face all the problems and practical limitations that rational interpolation in the Loewner framework faces.
\end{remark}
\section{Data-driven Projection-based Approximations of Various Generalizations of BT}\label{sec4}
ADI-based approximations of various generalizations of BT are non-intrusive but not data-driven, as samples of $G(s)$ in the right-half $s$-plane cannot be measured directly in an experimental setting. In ADI methods, the shifts must have a non-zero real part; this effectively makes it theoretically impossible to implement ADI-based approximations of various generalizations of BT using measurable data (i.e., samples of $G(s)$ on the imaginary axis). This theoretical limitation is reflected in the computation of the free parameter $\zeta$ within the interpolants $\hat{G}(s)=C\hat{V}(sI-S_v+\zeta L_v)^{-1}\zeta+D$ and $\hat{G}(s) = \zeta(sI - S_w + L_w\zeta)^{-1} \hat{W}^*B + D$, because computing $\zeta$ requires the solution of Lyapunov equations, which only have unique solutions when the interpolation points are in the right-half of the $s$-plane.

Let \(\mathcal{P}\) and \(\mathcal{Q}\) denote the generalized controllability and observability Gramians that arise in various generalizations of BT. The key observation in \citep{goseaQuad} is that if
\[
\mathcal{P}\approx \hat{V}\hat{L}_p\hat{L}_p^*\hat{V}^*, \qquad
\mathcal{Q}\approx \hat{W}\hat{L}_q\hat{L}_q^*\hat{W}^*,
\]
where \(\hat{V}\) and \(\hat{W}\) are defined in \eqref{Kry_V} and \eqref{Kry_W} with all interpolation points \(\sigma_i\) and \(\mu_i\) on the imaginary axis, and the factors \(\hat{L}_p\) and \(\hat{L}_q\) can be computed non-intrusively, then a data-driven implementation of any such generalization of BT follows directly. We note that if \(\mathcal{P}\) and \(\mathcal{Q}\) are approximated via Krylov subspace–based Petrov–Galerkin projection, approximations of the form \(\mathcal{P}\approx \hat{V}\hat{L}_p\hat{L}_p^*\hat{V}^*\) and \(\mathcal{Q}\approx \hat{W}\hat{L}_q\hat{L}_q^*\hat{W}^*\) arise naturally. The spectral factorizations appearing in \citep{reiter2023generalizations} result from rewriting the Riccati equations in these BT variants as integral expressions so that they can be approximated by numerical integration. This manipulation leads to the need for samples of certain spectral factorizations, which cannot be measured in an experimental setting. Since ADI methods are special cases of Krylov subspace–based methods, as noted in the previous section, the corresponding non-intrusive implementations require only samples of \(G(s)\). The requirement to sample in the right half of the \(s\)-plane stems from the theoretical conditions on the ADI shifts; it is not inherent to the Krylov subspace–based projection framework itself. We now present a proof-of-concept example to illustrate that the Gramians \(\mathcal{P}\) and \(\mathcal{Q}\) in the BT generalizations considered here can be approximated using interpolation points \(\sigma_i\) and \(\mu_i\) on the imaginary axis, while the factors \(\hat{L}_p\) and \(\hat{L}_q\) are still computed non-intrusively.\\
 
\noindent\textbf{Illustrative Example}\\
Consider an \(8^{\text{th}}\)-order LTI model with the following state-space realization:
\begin{align}
A &=\begin{bsmallmatrix}
-22.1414&  -14.1915&  -35.8543&   -7.8301&   54.2479&   -1.6149&    4.5713&  -47.9895\\
  -1.4098&    6.0485&   -2.0663&   36.2832&   88.6974&   15.7929&   74.5229&  -30.3651\\
  -20.9974&   -3.8320&   -7.5951&  -40.6679&  -71.4159&  -45.2401&  -39.4774&   -4.9186\\
 -107.6001&  -67.2020& -108.3961&  -33.5432&   43.9644&  -77.3467&   21.1433& -109.2904\\
  143.7150&   66.7452&  146.1617&   81.1110&  -20.7467&  135.0017&   12.6135&  158.4851\\
 -101.2836&  -48.3008&  -93.7796&  -47.9424&  -20.5617&  -89.5448&  -28.7377&  -78.5117\\
  129.4943&   22.2846&   56.8624&  112.4968&   85.9715&  148.5772&   61.3542&  106.4938\\
   65.1568&   63.9156&  113.2290&  -42.8674& -170.4269&  -12.2243&  -97.6582&   98.1684
\end{bsmallmatrix},\nonumber\\
B &=\begin{bmatrix}
0.6007 &1.6263&-0.4206&2.5576&-2.1955&1.2682&-0.4758&-3.1936\end{bmatrix}^T,\nonumber\\
C &=\begin{bmatrix}1.9237&1.2498&1.3247&1.5407&1.1059&1.3546&0.9754&1.2928\end{bmatrix},\nonumber\\
D &=0.2378.\nonumber
\end{align}
This model is stable, minimum phase, square, and passive, with \(\|G(s)\|_{\mathcal{H}_\infty}=0.4999\). Hence all generalizations of BT considered in this paper are applicable to this example.

For \(\mathcal{H}_\infty\)BT, set \(\gamma=2\). Let the right interpolation points be \(\sigma_i=(j9.99,-j9.99,j19.99,-j19.99,j29.99,-j29.99)\). Next, construct the interpolant $\hat{G}(s)=C\hat{V}(sI-S_v+\zeta L_v)^{-1}\zeta+D$ from the samples \(\big(G(j9.99),G(j19.99),G(j29.99),G(\infty)\big)\) by choosing the free parameter \(\zeta\) as
\[
\zeta=\begin{bmatrix}1.0075+j0.0417\\1.0075-j0.0417\\1.0080-j0.0792\\1.0080+j0.0792\\0.9845-j0.2113\\0.9845+j0.2113\end{bmatrix}.
\]
Then solve the projected matrix equations by replacing $(A,B,C)$ with $(S_v-\zeta L_v,\zeta,C\hat{V})$ in \eqref{lyap_P}, \eqref{Ricc_LQG_P}, \eqref{Ricc_H_inf_P}, \eqref{P_PR}, and \eqref{P_BR} to compute $\hat{P}$, $\hat{P}_{\mathrm{LQG}}$, $\hat{P}_{\mathcal{H}_\infty}$, $\hat{P}_{\mathrm{PR}}$, and $\hat{P}_{\mathrm{BR}}$, respectively.

Dually, construct the interpolant $\hat{G}(s) = \zeta(sI - S_w + L_w\zeta)^{-1} \hat{W}^*B + D$
from the samples \(\big(G(j10),G(j20),G(j30),G(\infty)\big)\) by choosing
\[
\zeta=\begin{bmatrix}1.0075+j0.0417\\1.0075-j0.0417\\1.0080-j0.0792\\1.0080+j0.0792\\0.9845-j0.2113\\0.9845+j0.2113\end{bmatrix}^T.
\]
Next, solve the projected matrix equations by replacing $(A,B,C)$ with $(S_w-L_w\zeta,\hat{W}^*B,\zeta)$ in \eqref{lyap_Q}, \eqref{Ricc_LQG_Q}, \eqref{Ricc_H_inf_Q}, \eqref{Q_PR}, \eqref{Q_BR}, and \eqref{Q_SW} to compute $\hat{Q}$, $\hat{Q}_{\mathrm{LQG}}$, $\hat{Q}_{\mathcal{H}_\infty}$, $\hat{Q}_{\mathrm{PR}}$, $\hat{Q}_{\mathrm{BR}}$, and $\hat{Q}_{\mathrm{SW}}$, respectively. 
Furthermore, compute $\hat{Q}_{\mathrm{S}}$ by replacing $(\tilde{A}_\mathrm{S},\tilde{B}_\mathrm{S},\tilde{C}_{\mathrm{S}})$ with $\big(S_w-L_w\zeta-\hat{W}^*BR_{\mathrm{S}}^{-1}\zeta,(\hat{W}^*\hat{V}\hat{P}\hat{V}^*C^T+\hat{W}^*BD^T)R_{\mathrm{S}}^{-\frac{1}{2}},R_{\mathrm{S}}^{-\frac{1}{2}}\zeta\big)$ in \eqref{QtS}.

These non-intrusively computed Gramians of the projected state-space realizations are factorized as \(\hat{L}_p \hat{L}_p^*\) and \(\hat{L}_q \hat{L}_q^*\). Thereafter, the respective Gramians \(\mathcal{P}\) and \(\mathcal{Q}\) are replaced with their Petrov–Galerkin projection–based approximations \(\mathcal{P}\approx\hat{V}\hat{L}_p \hat{L}_p^*\hat{V}^*\) and \(\mathcal{Q}\approx\hat{W}\hat{L}_q \hat{L}_q^*\hat{W}^*\), respectively, in BSA, leading to a non-intrusive, data-driven implementations of BT, LQGBT, \(\mathcal{H}_\infty\)BT, PRBT, BRBT, SWBT, and BST constructing the ROMs from the samples \(\big(G(j9.99),G(j10),G(j19.99),G(j20),G(j29.99),G(j30),G(\infty)\big)\). The relative error $\frac{\|G(s)-\hat{G}(s)\|_{\mathcal{H}_\infty}}{\|G(s)\|_{\mathcal{H}_\infty}}$ for order-3 ROMs produced by the intrusive and non-intrusive versions of these BT variants is given in Table \ref{tab1}. The results show that the intrusive and non-intrusive implementations perform indistinguishably.
\begin{table}[!h]
\centering
\caption{Relative Error Comparison for the ROMs}\label{tab1}
\begin{tabular}{|c|c|}
\hline
&\\
Method & $\frac{\|G(s)-\hat{G}(s)\|_{\mathcal{H}_\infty}}{\|G(s)\|_{\mathcal{H}_\infty}}$ \\
&\\ \hline
Intrusive BT     &  $0.4039$\\ 
Non-intrusive BT     &  $0.4039$\\ 
Intrusive LQGBT     &  $0.4037$\\ 
Non-intrusive LQGBT     &  $0.4037$\\ 
Intrusive $\mathcal{H}_\infty$BT     & $0.4038$\\ 
Non-intrusive $\mathcal{H}_\infty$BT     & $0.4037$\\ 
Intrusive PRBT     & $0.4014$\\ 
Non-intrusive PRBT     & $0.4013$\\ 
Intrusive BRBT    & $0.4045$\\ 
Non-intrusive BRBT    & $0.4045$\\ 
Intrusive SWBT     & $0.4014$\\ 
Non-intrusive SWBT     & $0.4014$\\ 
Intrusive BST&0.4014\\
Non-intrusive BST&0.4014\\\hline
\end{tabular}
\end{table}\\

The elegance of ADI methods lies in the fact that they provide a free parameter \(\zeta\) that guarantees the projected Lyapunov equation admits a unique positive-definite solution and the projected Riccati equation admits a unique stabilizing positive-definite solution. This is achieved through their pole placement properties. In particular, LRCF-ADI places the poles of \(\hat{A}\) at the conjugates of ADI shifts \(-\sigma_i\), while RADI places the poles of \(\hat{A}-\hat{P}_{\mathrm{LQG}}\hat{C}^*\hat{C}\) at the conjugates of ADI shifts \(-\sigma_i\). Since the pole locations and the ADI shifts are tied together, the shifts \(-\sigma_i\) must be Hurwitz in LRCF-ADI and RADI. In a general setting, however, the interpolation points \(\sigma_i\) and the poles of \(\hat{A}\) and \(\hat{A}-\hat{P}_{\mathrm{LQG}}\hat{C}^*\hat{C}\) need not be linked in Krylov subspace–based projection methods. Therefore, it is possible to choose the free parameter \(\zeta\) so that the projected Lyapunov equation admits a unique positive-definite solution and the projected Riccati equation admits a unique stabilizing positive-definite solution even when the interpolation points \(\sigma_i\) lie anywhere in the \(s\)-plane, including the \(j\omega\)-axis. In this section, we give specific choices of \(\zeta\) that ensure the projected Lyapunov equation admits a unique positive-definite solution and the projected Riccati equation admits a unique stabilizing positive-definite solution when the interpolation points lie on the imaginary axis instead of in the right half of the \(s\)-plane.
\subsection{Data-driven Projection-based Approximation of BT}\label{sec4_0}
Assume that all interpolation points $\sigma_i$ are located on the imaginary axis, i.e., $\sigma_i=j\omega_i$, so that the interpolant $\hat{G}(s)=\hat{C}(sI-S_v+\zeta L_v)^{-1}\zeta+D$ can be computed from measurable data $G(j\omega_i)$. Recall that the projected matrices satisfy the following:
\[
\hat{A}=S_v-\hat{B} L_v,\quad \hat{B}=\zeta, \quad \hat{C}=\begin{bmatrix}H(j\omega_1)&\cdots&H(j\omega_v)\end{bmatrix},\quad D=G(\infty),
\]
where $S_v = \mathrm{diag}(j\omega_1, \dots, j\omega_v) \otimes I_m$ and $L_v=\begin{bmatrix}1&\cdots&1\end{bmatrix}\otimes I_m$.

To ensure that the projected controllability Gramian $\hat{P}$ associated with the pair $(\hat{A},\hat{B})$ has a unique solution, the projected matrix $\hat{A}$ must have all the eigenvalues in the left half of the $s$-plane. Therefore, the free parameter $\zeta$ should place the poles of interpolant $\hat{G}(s)$ in left half of the $s$-plane. Proposition \ref{prop_stable} gives an analytical method to compute $\zeta$ that place the poles of the interpolant $\hat{G}(s)$ at the desired location in the $s$-plane.
\begin{proposition}\label{prop_stable}
Let \((\sigma_1, \dots, \sigma_{v})\) be \(v\) distinct interpolation points, and let \((\lambda_1, \dots, \lambda_{v})\) be \(v\) distinct desired poles such that the two sets have no elements in common. Define
\[
S_p = \mathrm{diag}(-\overline{\lambda_1}, \dots, -\overline{\lambda_{v}}) \otimes I_m.
\]
Further, assume that the pairs \((-S_v, L_v)\) and \((-S_p, L_v)\) are observable. Let \(X_p\) be the unique solution of the Sylvester equation
\begin{align}
-S_p^* X_p - X_p S_v + L_v^T L_v = 0. \label{Xp}
\end{align}
Then the interpolant
\begin{align}
\hat{A} = S_v - \hat{B} L_v, \quad
\hat{B} = X_p^{-1} L_v^T, \quad
\hat{C} = \begin{bmatrix}H(\sigma_1)&\cdots&H(\sigma_v)\end{bmatrix},\quad D=G(\infty) \label{ROM1}
\end{align}
satisfies the interpolation condition $G(\sigma_i)=\hat{G}(\sigma_i)$ and has poles at \((\lambda_1, \dots, \lambda_{v})\) with multiplicity \(m\).
\end{proposition}
\begin{proof}
The assumption that the sets \((\sigma_1, \dots, \sigma_{v})\) and \((\lambda_1, \dots, \lambda_{v})\) are disjoint ensures the uniqueness of the Sylvester equation \eqref{Xp}. Moreover, the observability of the pairs \((-S_v, L_v)\) and \((-S_p, L_v)\) guarantees that \(X_p\) is invertible. Premultiplying \eqref{Xp} by \(X_p^{-1}\) gives
\begin{align}
-X_p^{-1} S_p^* X_p - S_v + X_p^{-1} L_v^T L_v = 0, \nonumber\\
S_v - \hat{B} L_v = - X_p^{-1} S_p^* X_p, \nonumber\\
\hat{A} = - X_p^{-1} S_p^* X_p. \nonumber
\end{align}
Hence, the eigenvalues of \(\hat{A}\) coincide with those of \(-S_p^*\), and the ROM \eqref{ROM1} has the desired poles \((\lambda_1, \dots, \lambda_{v})\) with multiplicity \(m\).
\end{proof}
The next proposition focuses on ensuring the invertibility of the matrix $X_p$ when the interpolation points $\sigma_i$ are located on the imaginary axis.
\begin{proposition}\label{prop_BT_diag}
Assume that Proposition \ref{prop_stable} holds. Further, for \(\epsilon > 0\), let \(\sigma_i = j\omega_i\) and \(\lambda_i = -\epsilon + j\omega_i\), satisfying
\[\Delta_{\min} := \min_{\substack{i,j=1 \\ i \neq j}}^{v} |\omega_i - \omega_j| > 0.\]
Then \(X_p\) admits the explicit Kronecker form \(X_p = Y \otimes I_m\), where \(Y \in \mathbb{C}^{v \times v}\) has entries
\begin{equation}
Y_{ij} = \frac{1}{\epsilon + j(\omega_j - \omega_i)}, \qquad i,j = 1,\dots,v.
\end{equation}
The following properties hold:
\begin{enumerate}
  \item $X_p$ is asymptotically diagonal as $\epsilon \to 0$ with $X_p = \frac{1}{\epsilon} I_{vm} + \mathcal{O}(1)$.
  \item By defining the diagonal approximation $\widetilde{X}_p := \frac{1}{\epsilon} I_{vm}$, the relative error in Frobenius norm satisfies
\begin{equation}
\label{eq:relative_error}
\frac{\| X_p - \widetilde{X}_p \|_F}{\| \widetilde{X}_p \|_F} \leq \frac{\epsilon \sqrt{v - 1}}{\Delta_{\min}}.
\end{equation}
In particular, for any tolerance $\tau \in (0,1)$, if
\begin{equation}
\epsilon < \tau \, \frac{\Delta_{\min}}{\sqrt{v - 1}}, \label{small_epsilon}
\end{equation}
then $\| X_p - \widetilde{X}_p \|_F / \| \widetilde{X}_p \|_F < \tau$.
  \item  $X_p$ is strictly diagonally dominant whenever
\begin{equation}
\epsilon < \frac{\Delta_{\min}}{v - 1}.
\end{equation}
\end{enumerate}
\end{proposition}
\begin{proof}
The proof is given in Appendix D.
\end{proof}
By selecting a small positive scalar \(\epsilon\) as in Proposition \ref{prop_BT_diag}, we can effectively impose a modal structure with lightly damped stable poles on the ROM \eqref{ROM1}; that is, as \(\epsilon \to 0\),
\[
\hat{A}\approx \mathrm{diag}(-\epsilon+j\omega_1,\cdots,-\epsilon+j\omega_{v})\otimes I_m, \quad \hat{B}\approx \begin{bmatrix}\epsilon&\cdots&\epsilon\end{bmatrix}^T\otimes I_m.
\]
Thus Proposition \ref{prop_BT_diag} allows us to ensure that the controllability Gramian $\hat{P}$ of the pair $(\hat{A},\hat{B})$ has a unique solution. As discussed in the last section, the Gramian $\hat{P}$ is diagonally dominant since $\hat{A}$ has modal structure with lightly damped stable poles as proven in \citep{jonckheere1984principal,gawronski2004dynamics,gawronski2006balanced}. For completeness, Proposition \ref{prop_gramian} establishes the diagonal dominance and positive definiteness of the controllability Gramian $\hat{P}$ of the pair $(\hat{A},\hat{B})$.
\begin{proposition}\label{prop_gramian}
Assume that Proposition \ref{prop_BT_diag} holds. Let $\hat{A} \in \mathbb{C}^{vm \times vm}$ and $\hat{B} \in \mathbb{C}^{vm \times m}$ be defined as:
\[
\hat{A} = S_v - \hat{B} L_v, \quad \hat{B} = X_p^{-1} L_v^T,
\]
where $S_v = \mathrm{diag}(j\omega_1, \dots, j\omega_v) \otimes I_m$, $L_v = \mathbf{1}_v^T \otimes I_m$, and $X_p$ is the solution from Proposition \ref{prop_BT_diag}. Let $\hat{P} \in \mathbb{C}^{vm \times vm}$ denote the controllability Gramian of the pair $(\hat{A}, \hat{B})$, satisfying the Lyapunov equation:
\begin{equation}
\label{eq:lyapunov_P}
\hat{A} \hat{P} + \hat{P} \hat{A}^* + \hat{B} \hat{B}^* = 0.
\end{equation}
For any tolerance $\delta \in (0, 1)$, if $\epsilon$ satisfies
\begin{equation}
\label{eq:epsilon_bound_gramian}
0 < \epsilon \le \min \left( \frac{\Delta_{\min}}{4 v}, \quad \frac{\delta \cdot \Delta_{\min}}{8 v^2} \right),
\end{equation}
then the following properties hold:
\begin{enumerate}
    \item $\hat{A}$ is Hurwitz, and $\hat{P}$ is Hermitian positive definite.
    \item $\hat{P}$ is strictly block diagonally dominant. Specifically, let $\hat{P}_d = \mathrm{blkdiag}(\hat{P}_{11}, \dots, \hat{P}_{vv})$ denote the block diagonal part of $\hat{P}$. Then for every block row $i$:
    \[
    \sum_{j \neq i} \| \hat{P}_{ij} \| \le \delta \cdot \sigma_{\min}(\hat{P}_{ii}).
    \]
    \item The diagonal blocks satisfy the asymptotic limit:
    \[
    \lim_{\epsilon \to 0} \frac{1}{\epsilon} \hat{P}_{ii} = \frac{1}{2} I_m, \qquad i = 1, \dots, v.
    \]
    Consequently, $\hat{P} = \frac{\epsilon}{2} I_{vm} + \mathcal{O}(\epsilon^3)$.
    \item The block diagonal approximation error in the induced infinity norm satisfies:
    \[
    \| \hat{P} - \hat{P}_d \|_\infty \le \delta \cdot \min_i \sigma_{\min}(\hat{P}_{ii}).
    \]
\end{enumerate}
\end{proposition}
\begin{proof}
The proof is given in Appendix E.
\end{proof}
By choosing $\epsilon$ based on Proposition \ref{prop_gramian} and assigning the desired poles of $\hat{A}$ to $-\epsilon+j\omega_i$, the controllability Gramian $\hat{P}$ for the projected pair $(\hat{A},\hat{B})$ admits a positive-definite solution, enabling the factorization $\hat{P}=\hat{L}_p\hat{L}_p^*$. Consequently, the controllability Gramian $P$ admits the projection-based approximation $P\approx\tilde{P}=\hat{V}\hat{L}_p\hat{L}_p^*\hat{V}^*$. Similarly, applying the dual of Proposition \ref{prop_gramian} yields the observability Gramian approximation $Q\approx\tilde{Q}=\hat{W}\hat{L}_q\hat{L}_q^*\hat{W}^*$. Replacing the exact Gramian factors $L_p$ and $L_q$ in BSA with the approximations $\tilde{L}_p=\hat{V}\hat{L}_p$ and $\tilde{L}_q=\hat{W}\hat{L}_q$ produces a data-driven projection-based approximation of BT (DD-P-BT). This approach allows DD-P-BT implementation using only measurable data $G(j\omega_i)$, $G(j\nu_i)$, and $G(\infty)$. Furthermore, if the amount of data to be processed makes it computationally infeasible to compute $\hat{P}$ and $\hat{Q}$ directly, one can replace them with their diagonally-dominant approximations $\hat{P}\approx\frac{\epsilon}{2} I$ and $\hat{Q}\approx\frac{\epsilon}{2} I$. As shown in Proposition \ref{prop_gramian}, such diagonally-dominant approximations are admissible.
\subsection{Data-driven Projection-based Approximation of LQGBT}\label{sec4_1}
In this subsection, similar to NI-ADI-LQG, $P_{\mathrm{LQG}}$ and $Q_{\mathrm{LQG}}$ are approximated via Petrov–Galerkin projection, i.e.,
$P_{\mathrm{LQG}}\approx \hat{V}\hat{P}_{\mathrm{LQG}}\hat{V}^*$ and $Q_{\mathrm{LQG}}\approx \hat{W}\hat{Q}_{\mathrm{LQG}}\hat{W}^*$. To this end, we need to choose the free parameter $\zeta$ such that the projected Riccati equations \eqref{proj_ricc_lqg} and \eqref{proj_ricc_lqg_Q} admit positive-definite stabilizing solutions. Theorems \ref{Theorem1} and \ref{Theorem2} provide specific choices of $\zeta$ that achieve this when the interpolation points lie in the right half of the $s$-plane. The focus here is to develop a method to compute $\zeta$ that ensures this property when the interpolation points lie on the imaginary axis.
\begin{theorem}\label{Theorem10}
Consider the algebraic Riccati equation (ARE)
\begin{align}
\hat{A}(\epsilon)\hat{P}_{\mathrm{LQG}}+\hat{P}_{\mathrm{LQG}}\hat{A}(\epsilon)^*+\hat{B}(\epsilon)\hat{B}(\epsilon)^*-\hat{P}_{\mathrm{LQG}}\hat{C}^*\hat{C}\hat{P}_{\mathrm{LQG}}=0, \label{eq:are_main}
\end{align}
where $\epsilon > 0$, and the system matrices are defined as:
\begin{align*}
\hat{A}(\epsilon) &= \mathrm{diag}(-\epsilon+j\omega_1,\cdots,-\epsilon+j\omega_v)\otimes I_m, \\
\hat{B}(\epsilon) &= \begin{bmatrix}\epsilon I_m\\\vdots\\\epsilon I_m\end{bmatrix}, \\
\hat{C} &= \begin{bmatrix}H(j\omega_1)&\cdots&H(j\omega_v)\end{bmatrix}.
\end{align*}
Assume the non-zero frequencies $\omega_1, \dots, \omega_v \in \mathbb{R}$ are distinct, and let $\Delta_{\mathrm{min}} := \min_{\substack{i,j=1 \\ i \neq j}}^{v} |\omega_i - \omega_j| > 0$. Assume further that each $H(j\omega_i)$ has full column rank. Then, there exist constants $\epsilon_0 > 0$ and $K > 0$ such that for all $0 < \epsilon \le \epsilon_0$, the unique stabilizing solution $\hat{P}_{\mathrm{LQG}}(\epsilon)$ exists and admits the decomposition:
\begin{align}
\hat{P}_{\mathrm{LQG}}(\epsilon) = \tilde{P}_{\mathrm{LQG}}(\epsilon) + E(\epsilon),
\end{align}
where $\tilde{P}_{\mathrm{LQG}}(\epsilon) = \mathrm{blkdiag}(p_1(\epsilon), \dots, p_v(\epsilon))$ with
\begin{align}
p_i(\epsilon) = \epsilon\big(H^*(j\omega_i)H(j\omega_i)\big)^{-1}\Big(\big(I_m+H^*(j\omega_i)H(j\omega_i)\big)^{\frac{1}{2}}-I_m\Big),
\end{align}
and the error term satisfies $\|E(\epsilon)\| \le K \epsilon^2$.
\end{theorem}
\begin{proof}
The proof is given in Appendix F.
\end{proof}
Theorem \ref{Theorem10} indicates that the projected Riccati equation \eqref{proj_ricc_lqg} has a stabilizing solution when the projected matrix $\hat{A}$ is Hurwitz and has a modal form with lightly damped poles, which is consistent with the observation made in \citep{gawronski2004dynamics,gawronski2006balanced}. Next, it is shown how to enforce this structure on $\hat{A}$ using Proposition \ref{prop_stable}.
\begin{proposition}\label{Prop12}
Let $\epsilon > 0$ and let $\omega_1, \dots, \omega_v \in \mathbb{R}$ be distinct non-zero frequencies. Define the minimum frequency magnitude $\omega_{\min} := \min_{i} |\omega_i|$ and the minimum frequency separation $\Delta_{\mathrm{min}} := \min_{\substack{i,j=1 \\ i \neq j}}^{v} |\omega_i - \omega_j| > 0$. Consider the matrix $\hat{A} \in \mathbb{C}^{vm \times vm}$ defined by
\[
\hat{A} = S_v - \hat{B}L_v,
\]
where $S_v = \mathrm{diag}(j\omega_1, \dots, j\omega_v) \otimes I_m$, $L_v = (\mathbf{1}_v^T \otimes I_m)$, and $\hat{B} = (\epsilon \mathbf{1}_v \otimes I_m)$, with $\mathbf{1}_v \in \mathbb{R}^v$ denoting the vector of all ones. Assume $v \ge 2$. For any admissible relative error tolerance $\delta \in (0, 1)$, if $\epsilon$ satisfies
\begin{equation}
\label{eq:epsilon_bound}
0 < \epsilon \le \min \left( \frac{\delta \cdot \omega_{\min}}{v - 1}, \quad \frac{\Delta_{\mathrm{min}}}{2(v - 1)} \right),
\end{equation}
then the following properties hold:
\begin{enumerate}
    \item $\hat{A}$ is strictly diagonally dominant by rows.
    \item Let $\hat{A}_d = \mathrm{diag}(\hat{A})$. The approximation error in the induced infinity norm satisfies
    \[
    \| \hat{A} - \hat{A}_d \|_\infty \le \delta \cdot \omega_{\min}.
    \]
    \item The eigenvalues of $\hat{A}$ are distinct and lie within $v$ disjoint Gershgorin disks in the complex plane, each centered at $-\epsilon+j\omega_i $ with radius $(v-1)\epsilon$.
\end{enumerate}
\end{proposition}
\begin{proof}
The proof is given in Appendix G.
\end{proof}
Thus, by selecting \(\epsilon\) according to Theorem \ref{Theorem10} and Proposition \ref{Prop12}, one can obtain the free parameter \(\hat{B}\) that ensures the projected Riccati equation \eqref{proj_ricc_lqg} admits a positive-definite stabilizing solution. Dually, the free parameter \(\hat{C}\) can be obtained to ensure that the projected Riccati equation \eqref{proj_ricc_lqg_Q} admits a positive-definite stabilizing solution. Consequently, \(P_{\mathrm{LQG}}\) and \(Q_{\mathrm{LQG}}\) admit projection-based approximations
\[
P_{\mathrm{LQG}}\approx \hat{V}\hat{P}_{\mathrm{LQG}}\hat{V}^*=\hat{V}\hat{L}_p\hat{L}_p^*\hat{V}^*, \quad
Q_{\mathrm{LQG}}\approx \hat{W}\hat{Q}_{\mathrm{LQG}}\hat{W}^*=\hat{W}\hat{L}_q\hat{L}_q^*\hat{W}^*.
\]
Replacing the exact Gramian factors \(L_p\) and \(L_q\) in BSA with the approximations \(\tilde{L}_p=\hat{V}\hat{L}_p\) and \(\tilde{L}_q=\hat{W}\hat{L}_q\) yields a data-driven projection-based approximation of LQGBT (DD-P-LQGBT). This approach enables the implementation of DD-P-LQGBT using only measurable data \(G(j\omega_i)\), \(G(j\nu_i)\), and \(G(\infty)\).

Furthermore, if the amount of data makes it computationally infeasible to solve the projected Riccati equations \eqref{proj_ricc_lqg} and \eqref{proj_ricc_lqg_Q}, one can use diagonally dominant approximations of \(\hat{P}_{\mathrm{LQG}}\) and \(\hat{Q}_{\mathrm{LQG}}\) from Theorem \ref{Theorem10} as follows:
\begin{align}
\hat{P}_{\mathrm{LQG}} &\approx\tilde{P}_{\mathrm{LQG}}= \mathrm{blkdiag}(p_1, \dots, p_v),\label{p_lqg_analy}\\
\hat{Q}_{\mathrm{LQG}}&\approx\tilde{Q}_{\mathrm{LQG}}=\mathrm{blkdiag}\big(q_1,\cdots,q_w\big),\label{q_lqg_analy}
\end{align}
where
\begin{align}
p_i&= \epsilon\big(H^*(j\omega_i)H(j\omega_i)\big)^{-1}\Big(\big(I_m+H^*(j\omega_i)H(j\omega_i)\big)^{\frac{1}{2}}-I_m\Big),\nonumber\\
q_i&=\epsilon\big(H(j\nu_i)H^*(j\nu_i)\big)^{-1}\Big(\big(I_p+H(j\nu_i)H^*(j\nu_i)\big)^{\frac{1}{2}}-I_p\Big).\nonumber
\end{align}
\begin{remark}
A data-driven projection-based approximation of \(\mathcal{H}_\infty\)BT (DD-P-\(\mathcal{H}_\infty\)BT) can be obtained similarly using Theorem \ref{Theorem10} and Proposition \ref{Prop12}. Furthermore, if the amount of data makes it computationally infeasible to solve the projected Riccati equations in DD-P-\(\mathcal{H}_\infty\)BT, one can use diagonally dominant approximations from Theorem \ref{Theorem10} as follows:
\begin{align}
\hat{P}_{\mathcal{H}_\infty}\approx\mathrm{blkdiag}(p_1,\cdots,p_v)\quad \text{and}\quad\hat{Q}_{\mathcal{H}_\infty}\approx\mathrm{blkdiag}(q_1,\cdots,q_w),\nonumber
\end{align}
where
\begin{align}
p_i&=\epsilon\big((1-\gamma^2)H^*(j\omega_i)H(j\omega_i)\big)^{-1}\Big(\big(I_m+(1-\gamma^2)H^*(j\omega_i)H(j\omega_i)\big)^{\frac{1}{2}}-I_m\Big),\nonumber\\
q_i&=\epsilon\big((1-\gamma^2)H(j\nu_i)H^*(j\nu_i)\big)^{-1}\Big(\big(I_p+(1-\gamma^2)H(j\nu_i)H^*(j\nu_i)\big)^{\frac{1}{2}}-I_p\Big).\nonumber
\end{align}
\end{remark}
\subsection{Data-driven Projection-based Approximation of PRBT}\label{sec4_2}
In this subsection, similar to RADI-based approximations, Petrov-Galerkin projection-based approximations of $P_{\mathrm{PR}}$ and $Q_{\mathrm{PR}}$ are proposed, i.e., $P_{\mathrm{PR}}\approx V_{\mathrm{PR}}\hat{P}_{\mathrm{PR}}V_{\mathrm{PR}}^*=\hat{V}T_v\hat{P}_{\mathrm{PR}}T_v^*\hat{V}^*$ and $Q_{\mathrm{PR}}\approx W_{\mathrm{PR}}\hat{Q}_{\mathrm{PR}}W_{\mathrm{PR}}^*=\hat{W}T_w\hat{Q}_{\mathrm{PR}}T_w^*W^*$. Unlike RADI-based approximations, the projection matrices $V_{\mathrm{PR}}$ and $W_{\mathrm{PR}}$ are computed with interpolation points on the imaginary axis, i.e., $\sigma_i=j\omega_i$ and $\mu_i=j\nu_i$. We use Theorem \ref{Theorem10} to guarantee that the projected Riccati equation \eqref{proj_ricc_Ppr} admits a stabilizing solution. Setting
\[
\hat{B}_{\mathrm{PR}}=\begin{bmatrix}\epsilon I_m\\\vdots\\\epsilon I_m\end{bmatrix},
\]
with \(\epsilon\) chosen according to Theorem \ref{Theorem10}, yields
\[
\hat{P}_{\mathrm{PR}}(\epsilon)=\mathrm{blkdiag}(p_1,\cdots,p_v)+E(\epsilon)
\]
where
\[
p_i=\epsilon\big[G_{\mathrm{PR}}^*(j\omega_i)G_{\mathrm{PR}}(j\omega_i)\big]^{-1}\Big(I_m-\big(I_m-G_{\mathrm{PR}}^*(j\omega_i)G_{\mathrm{PR}}(j\omega_i)\big)^{\frac{1}{2}}\Big).\nonumber
\]
Recall that
\begin{align}
\begin{bmatrix}G_{\mathrm{PR}}(j\omega_1)&\cdots&G_{\mathrm{PR}}(j\omega_v)\end{bmatrix}=R_{\text{PR}}^{-\frac{1}{2}}\begin{bmatrix}H(j\omega_1)&\cdots&H(j\omega_v)\end{bmatrix}T_v,\nonumber
\end{align}
where $T_v$ is the solution to the Sylvester equation \eqref{Tv_pr}.

Resultantly, the projected Riccati equation \eqref{proj_ricc_Ppr} can be computed non-intrusively from the samples $G(\infty)$ and $G(j\omega_i)$. Proposition \ref{prop13} gives the appropriate selection of the free parameter $\zeta$ in the Sylvester equation \eqref{Tv_pr}, which ensures the invertibility and diagonal dominance of $T_v$.
\begin{proposition}\label{prop13}
Let Proposition \ref{Prop12} hold. Let $\epsilon > 0$ and let $\omega_1, \dots, \omega_v \in \mathbb{R}$ be distinct non-zero frequencies with minimum separation $\Delta_{\mathrm{min}} := \min_{i \neq j} |\omega_i - \omega_j|$. Let $\mathcal{P} \in \mathbb{C}^{m \times vm}$ and $\mathcal{Q} \in \mathbb{C}^{m \times m}$ be matrices independent of $\epsilon$. Partition $\mathcal{P} = [\mathcal{P}_1, \dots, \mathcal{P}_v]$ where $\mathcal{P}_i \in \mathbb{C}^{m \times m}$.

Assume the following conditions hold:
\begin{enumerate}
    \item[(A1)] $(I_m + \mathcal{P}_i)$ is invertible for all $i=1,\dots,v$.
    \item[(A2)] $\mathcal{Q}$ is invertible.
\end{enumerate}

Define the constants:
\[
\gamma := \max_{i} \| (I_m + \mathcal{P}_i)^{-1} \|, \quad K_{\mathcal{P}} := \max_{i} \| I_m + \mathcal{P}_i \|, \quad \sigma_{\mathcal{Q}} := \sigma_{\min}(\mathcal{Q}),
\]
where $\|\cdot\|$ denotes the spectral norm. Let $T_v \in \mathbb{C}^{vm \times vm}$ be the solution to the Sylvester equation
\begin{equation}
\label{eq:sylvester_Tv}
    (S_v - \zeta L_v - \zeta \mathcal{P}) T_v - T_v S_v + \zeta \mathcal{Q} L_v = 0.
\end{equation}
Assume $v \ge 2$. For any tolerance $\delta \in (0, 1)$, if $\epsilon$ satisfies
\begin{equation}
\label{eq:epsilon_bound_Tv_final}
0 < \epsilon \le \min \left( \frac{\Delta_{\mathrm{min}}}{2 v K_{\mathcal{P}}}, \quad \frac{\delta \cdot \sigma_{\mathcal{Q}} \cdot \Delta_{\mathrm{min}}}{4 v K_{\mathcal{P}}^2 (v-1) \gamma (1 + \gamma K_{\mathcal{P}}) } \right),
\end{equation}
then the following properties hold:
\begin{enumerate}
    \item The spectra of $(S_v - \zeta L_v - \zeta \mathcal{P})$ and $S_v$ are disjoint, ensuring $T_v$ exists and is unique.
    \item As $\epsilon \to 0$, the blocks $T_{ij}$ of $T_v$ satisfy:
    \[
    \lim_{\epsilon \to 0} T_{ii} = (I_m + \mathcal{P}_i)^{-1} \mathcal{Q}, \quad \lim_{\epsilon \to 0} T_{ij} = 0 \quad (i \neq j).
    \]
    \item Each diagonal block $T_{ii}$ is invertible, with $\| T_{ii}^{-1} \| \le 2 K_{\mathcal{P}} \sigma_{\mathcal{Q}}^{-1}$.
    \item $T_v$ is block diagonally dominant in the sense that for every block row $i$:
    \[
    \sum_{j \neq i} \| T_{ij} \| \le \delta \cdot \sigma_{\min}(T_{ii}).
    \]
    Consequently, $T_v$ is invertible.
    \item The block diagonal approximation $\tilde{T}_v = \mathrm{blkdiag}(T_{11}, \dots, T_{vv})$ satisfies the relative error bound:
    \[
    \| T_v - \tilde{T}_v \|_\infty \le \delta \cdot \min_i \sigma_{\min}(T_{ii}).
    \]
\end{enumerate}
\end{proposition}
\begin{proof}
The proof is given in Appendix H.
\end{proof}
Petrov-Galerkin projection-based approximation of $Q_{\mathrm{PR}}$ can be obtained dually by selecting the free parameter $\hat{C}_{\mathrm{PR}}$ in the projected Riccati equation \eqref{proj_ricc_Qpr} according to the dual of Theorem \ref{Theorem10}. Furthermore, the free parameter $\zeta$ in the Sylvester equation \eqref{Tw_pr} can be set according to the dual of Proposition \ref{prop13} to ensure the invertibility of $T_w$. Consequently, $P_{\mathrm{PR}}$ and $Q_{\mathrm{PR}}$ admit projection-based approximations $P_{\mathrm{PR}}\approx \hat{V}T_v\hat{P}_{\mathrm{PR}}T_v^*\hat{V}^*=\hat{V}T_v\hat{L}_p\hat{L}_p^*T_v^*\hat{V}^*$, and $Q_{\mathrm{PR}}\approx \hat{W}T_w\hat{Q}_{\mathrm{PR}}T_w^*\hat{W}^*=\hat{W}T_w\hat{L}_q\hat{L}_q^*T_w^*\hat{W}^*$. Replacing the exact Gramian factors $L_p$ and $L_q$ in BSA with the approximations $\tilde{L}_p=\hat{V}T_v\hat{L}_p$ and $\tilde{L}_q=\hat{W}T_w\hat{L}_q$ yields a data-driven projection-based approximation of PRBT (DD-P-PRBT). This approach enables the implementation of DD-P-PRBT using only measurable data $G(j\omega_i)$, $G(j\nu_i)$, and $G(\infty)$.

If the amount of data to be processed makes it computationally infeasible to solve the projected Riccati equations \eqref{proj_ricc_Ppr} and \eqref{proj_ricc_Qpr} and Sylvester equations \eqref{Tv_pr} and \eqref{Tw_pr}, these can be replaced with their diagonally dominant approximations as follows:
\begin{align}
\hat{P}_{\text{PR}}&\approx\mathrm{blkdiag}(p_1,\cdots,p_v),\nonumber\\
T_v&\approx\mathrm{blkdiag}(t_{v,1},\cdots,t_{v,v}),\nonumber\\
\hat{Q}_{\text{PR}}&\approx\mathrm{blkdiag}(q_1,\cdots,q_w),\nonumber\\
T_w&\approx= \mathrm{blkdiag}(t_{w,1},\cdots,t_{w,w}),\nonumber
\end{align}
where
\begin{align}
p_i&=\epsilon\big[t_{v,i}^*H^*(j\omega_i)R_{\text{PR}}^{-1}H(j\omega_i)t_{v,i}\big]^{-1}\Big(I_m-\big(I_m-t_{v,i}^*H^*(j\omega_i)R_{\text{PR}}^{-1}H(j\omega_i)t_{v,i}\big)^{\frac{1}{2}}\Big),\nonumber\\
t_{v,i}&=\big(I_m+R_{\text{PR}}^{-1}H(j\omega_i)\big)^{-1}R_{\text{PR}}^{-\frac{1}{2}},\nonumber\\
q_i&=\epsilon\big[t_{w,i}^*H(j\nu_i)R_{\text{PR}}^{-1}H^*(j\nu_i)t_{w,i}\big]^{-1}\Big(I_p-\big(I_p-t_{w,i}^*H(j\nu_i)R_{\text{PR}}^{-1}H^*(j\nu_i)t_{w,i}\big)^{\frac{1}{2}}\Big),\nonumber\\
t_{w,i}&=\big(I_p+R_{\text{PR}}^{-1}H^*(j\nu_i)\big)^{-1}R_{\text{PR}}^{-\frac{1}{2}}.\nonumber
\end{align}
\subsection{Data-driven Projection-based Approximation of BRBT}\label{sec4_3}
The data-driven projection-based approximation of BRBT (DD-P-BRBT) can be obtained similarly using Theorem \ref{Theorem10} and Proposition \ref{prop13} to obtain Petrov-Galerkin projection-based approximations of $P_{\mathrm{BR}}$ and $Q_{\mathrm{BR}}$, i.e., $P_{\mathrm{BR}}\approx V_{\mathrm{BR}}\hat{P}_{\mathrm{BR}}V_{\mathrm{BR}}^*=\hat{V}T_v\hat{P}_{\mathrm{BR}}T_v^*\hat{V}^*$ and $Q_{\mathrm{BR}}\approx W_{\mathrm{BR}}\hat{Q}_{\mathrm{BR}}W_{\mathrm{BR}}^*=\hat{W}T_w\hat{Q}_{\mathrm{BR}}T_w^*\hat{W}^*$. Unlike RADI-based approximations, the projection matrices $V_{\mathrm{BR}}$ and $W_{\mathrm{BR}}$ are computed with interpolation points on the imaginary axis, i.e., $\sigma_i=j\omega_i$ and $\mu_i=j\nu_i$. Furthermore, $\hat{P}_{\mathrm{BR}}(\epsilon)$ and $\hat{Q}_{\mathrm{BR}}(\epsilon)$ are given as
\begin{align}
\hat{P}_{\mathrm{BR}}(\epsilon)&=\mathrm{blkdiag}(p_1,\cdots,p_v)+E_p(\epsilon),\nonumber\\
\hat{Q}_{\mathrm{BR}}(\epsilon)&=\mathrm{blkdiag}(q_1,\cdots,q_w)+E_q(\epsilon),\nonumber
\end{align}
where
\begin{align}
p_i&=\epsilon\big[\big(G_{\mathrm{BR}}^{\mathrm{P}}(j\omega_i)\big)^*G_{\mathrm{PR}}^{\mathrm{P}}(j\omega_i)\big]^{-1}\Big(I_m-\Big(I_m-\big(G_{\mathrm{PR}}^{\mathrm{P}}(j\omega_i)\big)^*G_{\mathrm{PR}}^{\mathrm{P}}(j\omega_i)\Big)^{\frac{1}{2}}\Big),\nonumber\\
q_i&=\epsilon\big[G_{\mathrm{BR}}^{\mathrm{Q}}(j\nu_i)\big(G_{\mathrm{BR}}^{\mathrm{Q}}(j\nu_i)\big)^*\big]^{-1}\Big(I_p-\Big(I_p-G_{\mathrm{BR}}^{\mathrm{Q}}(j\nu_i)\big(G_{\mathrm{BR}}^{\mathrm{Q}}(j\nu_i)\big)^*\Big)^{\frac{1}{2}}\Big).\nonumber
\end{align}
Recall that
\begin{align}
\begin{bmatrix}G_{\mathrm{BR}}^{\mathrm{P}}(j\omega_1)&\cdots&G_{\mathrm{BR}}^{\mathrm{P}}(j\omega_v)\end{bmatrix}&=R_{p}^{-\frac{1}{2}}\begin{bmatrix}H(j\omega_1)&\cdots&H(j\omega_v)\end{bmatrix}T_v,\nonumber\\
\begin{bmatrix}G_{\mathrm{BR}}^{\mathrm{Q}}(j\nu_1)\\\vdots\\G_{\mathrm{BR}}^{\mathrm{Q}}(j\nu_w)\end{bmatrix}&=T_w^*\begin{bmatrix}H(j\nu_1)\\\vdots\\H(j\nu_w)\end{bmatrix}R_{q}^{-\frac{1}{2}},\nonumber
\end{align}
where $T_v$ and $T_w$ are the solutions to the Sylvester equations \eqref{Tv_br} and \eqref{Tw_br}, respectively.

Consequently, the projected Riccati equations \eqref{proj_ricc_Pbr} and \eqref{proj_ricc_Qbr} can be computed non-intrusively from the samples $G(\infty)$, $G(j\omega_i)$, and $G(j\nu_i)$. If the amount of data to be processed makes it computationally infeasible to solve the projected Riccati equations \eqref{proj_ricc_Pbr} and \eqref{proj_ricc_Qbr} and Sylvester equations \eqref{Tv_br} and \eqref{Tw_br}, these can be replaced with their diagonally-dominant approximations as follows:
\begin{align}
\hat{P}_{\text{BR}}&\approx\mathrm{blkdiag}(p_1,\cdots,p_v),\nonumber\\
T_v&\approx\mathrm{blkdiag}(t_{v,1},\cdots,t_{v,v}),\nonumber\\
\hat{Q}_{\text{BR}}&\approx\mathrm{blkdiag}(q_1,\cdots,q_w),\nonumber\\
T_w&\approx \mathrm{blkdiag}(t_{w,1},\cdots,t_{w,w}),\nonumber
\end{align}
where
\begin{align}
p_i&=\epsilon\big[t_{v,i}^*H^*(j\omega_i)R_{p}^{-1}H(j\omega_i)t_{v,i}\big]^{-1}\Big(I_m-\big(I_m-t_{v,i}^*H^*(j\omega_i)R_{p}^{-1}H(j\omega_i)t_{v,i}\big)^{\frac{1}{2}}\Big),\nonumber\\
t_{v,i}&=\big(I_m+D^TR_{p}^{-1}H(j\omega_i)\big)^{-1}R_{b}^{\frac{1}{2}},\nonumber\\
q_i&=\epsilon\big[t_{w,i}^*H(j\nu_i)R_{q}^{-1}H^*(j\nu_i)t_{w,i}\big]^{-1}\Big(I_p-\big(I_p-t_{w,i}^*H(j\nu_i)R_{q}^{-1}H^*(j\nu_i)t_{w,i}\big)^{\frac{1}{2}}\Big),\nonumber\\
t_{w,i}&=\big(I_p+DR_{q}^{-1}H^*(j\nu_i)\big)^{-1}R_{c}^{\frac{1}{2}}.\nonumber
\end{align}
\subsection{Data-driven Projection-based Approximation of SWBT}\label{sec4_4}
The data-driven projection-based approximation of SWBT (DD-P-SWBT) can be obtained using Propositions \ref{prop_gramian} and \ref{prop13}. Similar to ADI-based approximations, Petrov-Galerkin projection-based approximations of $P$ and $Q_{\mathrm{SW}}$, i.e., $P\approx \hat{V}\hat{P}\hat{V}^*$ and $Q_{\mathrm{SW}}\approx W_{\mathrm{SW}}\hat{Q}_{\mathrm{SW}}W_{\mathrm{SW}}^*=\hat{W}T_w\hat{Q}_{\mathrm{SW}}T_w^*\hat{W}^*$, can be obtained. Unlike ADI-based approximations, the projection matrices $\hat{V}$ and $W_{\mathrm{SW}}$ are computed with interpolation points on the imaginary axis, i.e., $\sigma_i=j\omega_i$ and $\mu_i=j\nu_i$.

$\hat{P}$ and $\hat{Q}_{\mathrm{SW}}$ solve the following Lyapunov equations:
\begin{align}
(S_v-\hat{B}L_v)\hat{P}+\hat{P}(S_v-\hat{B}L_v)^*+\hat{B}\hat{B}^*&=0,\label{proj_p}\\
(S_w-L_w\hat{C}_{\text{SW}})^*\hat{Q}_{\text{SW}}+\hat{Q}_{\text{SW}}(S_w-L_w\hat{C}_{\text{SW}})+\hat{C}_{\text{SW}}^*\hat{C}_{\text{SW}}&=0.\label{proj_qsw}
\end{align} The free parameters $\hat{B}$ and $\hat{C}_{\text{SW}}$ in the projected Lyapunov equations \eqref{proj_p} and \eqref{proj_qsw}, respectively, can be chosen according to Proposition \ref{prop_gramian} to ensure that $\hat{P}$ and $\hat{Q}_{\text{SW}}$ have full rank. Furthermore, Proposition \ref{prop13} can be used to set $\zeta$ in the Sylvester equation \eqref{Tw_sw} to ensure that $T_w$ is invertible.

If the amount of data to be processed makes it computationally infeasible to solve the projected Lyapunov equations \eqref{proj_p} and \eqref{proj_qsw} and Sylvester equation \eqref{Tw_sw}, these can be replaced with their diagonally-dominant approximations as follows:
\begin{align}
\hat{P}&\approx\text{blkdiag}\big(\tfrac{\epsilon}{2} I_m,\cdots,\tfrac{\epsilon}{2} I_m\big),\\
\hat{Q}_{\text{SW}}&\approx\text{blkdiag}\big(\tfrac{\epsilon}{2} I_p,\cdots,\tfrac{\epsilon}{2} I_p\big),\\
T_w&\approx\text{blkdiag}\big((G(j\nu_1))^{-*},\cdots,(G(j\nu_w))^{-*}\big).
\end{align}
\subsection{Data-driven Projection-based Approximation of BST}\label{sec4_5}
The data-driven projection-based approximation of BST (DD-P-BST) can be obtained similarly using Proposition \ref{prop_gramian}, Theorem \ref{Theorem10}, and Proposition \ref{prop13} to obtain Petrov-Galerkin projection-based approximations of $P$ and $Q_{\mathrm{S}}$, i.e., $P\approx \hat{V}\hat{P}\hat{V}^*$ and $Q_{\mathrm{S}}\approx W_{\mathrm{S}}\hat{Q}_{\mathrm{S}}W_{\mathrm{S}}^*=\hat{W}T_w\hat{Q}_{\mathrm{S}}T_w^*\hat{W}^*$. Unlike ADI-based approximations, the projection matrices $\hat{V}$ and $W_{\mathrm{S}}$ are computed with interpolation points on the imaginary axis, i.e., $\mu_i=j\nu_i$.

$\hat{Q}_{\mathrm{S}}$ solves the following Riccati equation:
\begin{align}
(S_w-L_w\hat{C}_{\mathrm{S}})^*\hat{Q}_{\mathrm{S}}+\hat{Q}_{\mathrm{S}}(S_w-L_w\hat{C}_{\mathrm{S}})+\hat{C}_{\mathrm{S}}^*\hat{C}_{\mathrm{S}}+\hat{Q}_{\mathrm{S}}T_w^*\hat{W}^*\tilde{B}_{\mathrm{S}}\tilde{B}_{\mathrm{S}}^*\hat{W}T_w\hat{Q}_{\mathrm{S}}=0,\label{proj_qst}
\end{align} where $T_w$ solves the following Sylvester equation:
\begin{align}
\big(S_w-L_w\zeta-[(\hat{W}^*\hat{V})\hat{P}(C\hat{V})^*+(\hat{W}^*B)D^T]R_{\text{S}}^{-1}\zeta\big)^*T_{w} - T_{w}S_w^*+\zeta^*R_{\text{S}}^{-\frac{1}{2}}L_w^T=0.\label{Tw_bst}
\end{align}
Recall that
\[
\hat{W}^*\tilde{B}_{\mathrm{S}}=\big((\hat{W}^*\hat{V})\hat{P}(C\hat{V})^*+(\hat{W}^*B)D^T\big)R_{\mathrm{S}}^{-\frac{1}{2}}.
\]
By selecting the free parameter $\hat{C}_{\mathrm{S}}$ according to the dual of Theorem \ref{Theorem10}, a stabilizing solution to the projected Riccati equation \eqref{proj_qst} can be obtained non-intrusively from the samples $G(\infty)$, $G(j\omega_i)$, and $G(j\nu_i)$. Furthermore, by selecting the free parameter $\zeta$ in Sylvester equation \eqref{Tw_bst} according to the dual of Proposition \ref{prop13}, the invertibility of $T_w$ can be ensured.

If the amount of data to be processed makes it computationally infeasible to solve the projected Riccati equation \eqref{proj_qst} and Sylvester equation \eqref{Tw_bst}, these can be replaced with their diagonally-dominant approximations as follows:
\begin{align}
\hat{Q}_{\text{S}}&\approx\mathrm{blkdiag}(q_1,\cdots,q_w),\nonumber\\
T_w&\approx\mathrm{blkdiag}\big(t_{w,1},\cdots,t_{w,w}\big),\nonumber
\end{align}
where
\begin{align}
q_i&=\epsilon[t_{w,i}^*\mathscr{H}(j\nu_i)R_{\text{S}}^{-1}\mathscr{H}^*(j\nu_i)t_{w,i}]^{-1}\big(I_p-(I_p-t_{w,i}^*\mathscr{H}(j\nu_i)R_{\text{S}}^{-1}\mathscr{H}^*(j\nu_i)t_{w,i})^{\frac{1}{2}}\big),\nonumber\\
t_{w,i}&=\Big[R_{\text{S}}+\mathscr{H}^*(j\nu_i)\Big]^{-1}R_{\text{S}}^{\frac{1}{2}},\nonumber\\
\mathscr{H}(j\nu_i) &= \begin{cases}
\dfrac{\epsilon}{2(j\nu_i-j\omega_i)}\big[H(j\omega_i)-H(j\nu_i)\big]H^*(j\omega_i)+H(j\nu_i)D^T, & \text{if } \omega_i\neq\nu_i, \\[2ex]
\dfrac{\epsilon}{2} H^{\prime}(j\omega_i)H^*(j\omega_i)+H(j\nu_i)D^T, & \text{if } \omega_i=\nu_i.
\end{cases}\nonumber
\end{align}
\section{Numerical Examples}\label{sec5}
This section evaluates the numerical performance of the proposed algorithms using benchmark dynamical system models for testing MOR algorithms. The MATLAB codes to reproduce the results are provided in \citep{mycode}. All simulations were performed in MATLAB R2021b on a laptop with a 2 GHz Intel i7 processor and 16 GB of RAM.\\

\noindent\textbf{Example 1: CD Player}\\
Consider the $120$-th order multiple-input multiple-output (MIMO) model from the benchmark collection of dynamic systems for testing MOR algorithms \citep{chahlaoui2005benchmark}. For demonstration purposes, the interpolation points consist of two sets of $300$ logarithmically spaced frequencies in $[10^{-3},10^{3}] \cup [-10^{3},-10^{-3}]$ rad/sec, forming $150$ conjugate pairs. The right and left interpolation points do not share any common elements. The free parameter $\epsilon$ is set to $10^{-5}$. For QuadBT, quadrature weights are computed for these nodes using the exponential trapezoidal rule \citep{goseaQuad}. For NI-ADI-BT, the ADI shifts $-\epsilon + j\omega_i$ are used to approximate $P$, and the ADI shifts $-\epsilon + j\nu_i$ are used to approximate $Q$. The samples $G(\epsilon + j\omega_i)$, $G(j\omega_i)$, $G(\epsilon + j\nu_i)$, and $G(j\nu_i)$ are computed numerically using the state-space realization of the CD player. Since the number of samples is moderate, the Sylvester equations and projected Lyapunov equations can be solved directly. Nevertheless, both the exact solutions and their block diagonally dominant approximations are computed and compared. The non-intrusive implementations, which use block diagonally dominant approximations of the projected Lyapunov equations, are marked with an asterisk in the figures. The Hankel singular values of the ROMs of order $25$ produced by BT, QuadBT, NI-ADI-BT, NI-ADI-BT$^*$, DD-P-BT, and DD-P-BT$^*$ are compared in Fig. \ref{fig1a}.
All non-intrusive BT algorithms capture the $20$ most significant Hankel singular values of the original system. Fig. \ref{fig1b} compares the relative error $\frac{\|G(s)-\hat{G}(s)\|_{\mathcal{H}_\infty}}{\|G(s)\|_{\mathcal{H}_\infty}}$ for ROMs of orders $1–25$. Both intrusive and non-intrusive BT algorithms perform comparably.
\begin{figure}[!h]
    \centering
    \begin{subfigure}{0.48\textwidth}
        \centering
        \includegraphics[width=\linewidth]{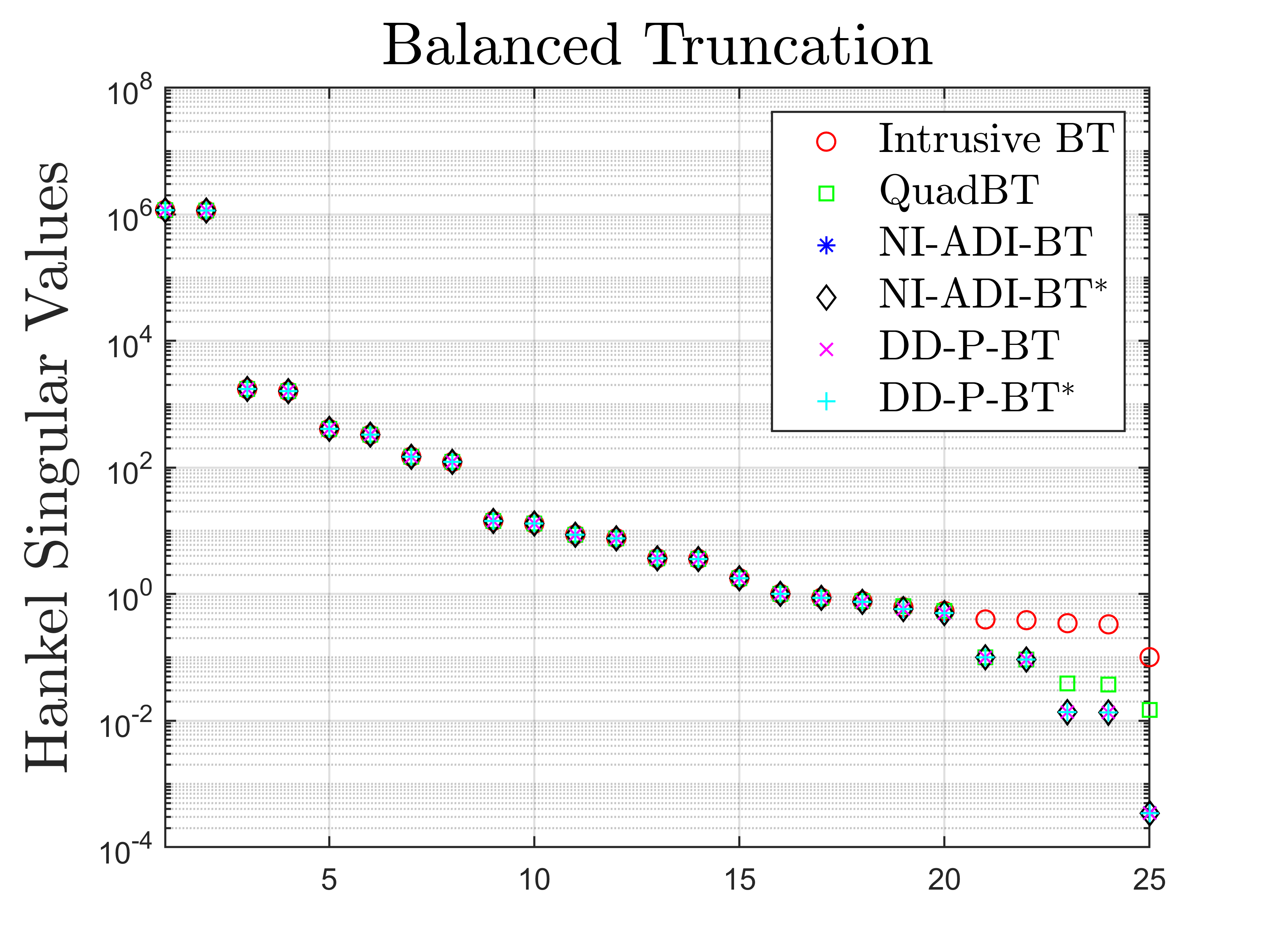}
        \caption{Hankel Singular Values Comparison}\label{fig1a}
    \end{subfigure}
    \hfill
    \begin{subfigure}{0.48\textwidth}
        \centering
        \includegraphics[width=\linewidth]{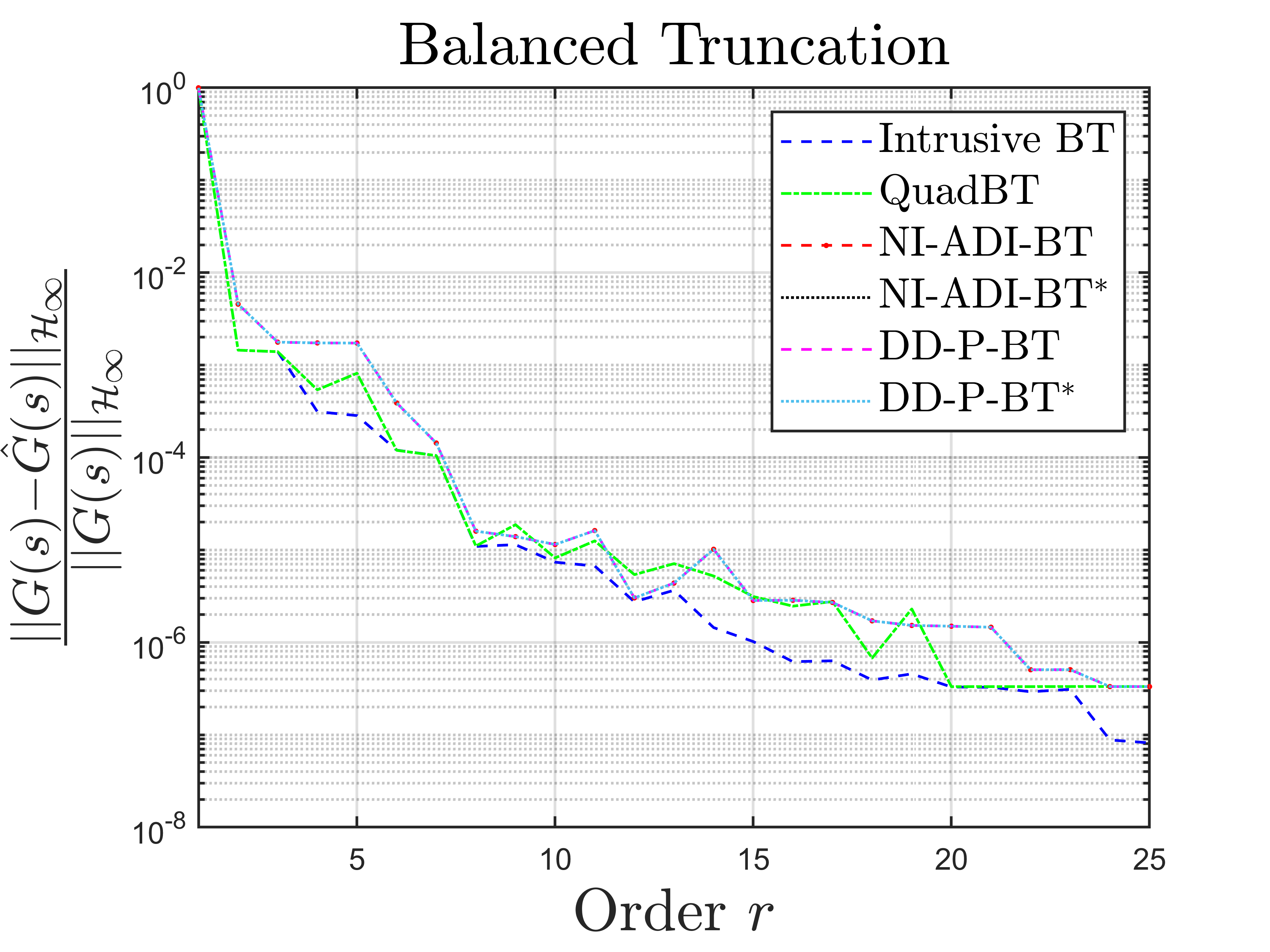}
        \caption{Relative Error Comparison}\label{fig1b}
    \end{subfigure}
    \caption{Performance Comparison between Intrusive and Non-intrusive BT}
\end{figure}

Next, the impact of different values of $\epsilon$ on capturing the dominant $25$ Hankel singular values is investigated for $\epsilon = 10^{-2}$, $10^{-3}$, $10^{-4}$, and $10^{-5}$. The relative difference between the exact Hankel singular values and the captured Hankel singular values in the $L_2$ norm is shown in Table \ref{tab2}. For $\epsilon=10^{-2}$, $Q_v^{-1}$ and $P_{w}^{-1}$ in NI-ADI-BT are ill-conditioned, so the ROM did not capture most Hankel singular values accurately, as indicated by the significant difference in Table \ref{tab2}. Interestingly, when these inverses are replaced with their diagonally dominant approximation, the resulting ROM captured the most significant Hankel singular values accurately, as shown by the small relative difference. For the remaining values of $\epsilon$, no numerical issues arise, leading to accurate ROMs that captured the most significant Hankel singular values well. It can be noted that block-diagonally dominant approximations are promising for avoiding numerical issues and ensuring ROM accuracy. Moreover, these approximations eliminate the need to compute the projected Lyapunov equations. Therefore, we advocate using block-diagonally dominant approximations instead of solving the projected Lyapunov equations.
\begin{table}[!h]
\centering
\caption{Relative Difference in Hankel Singular values}\label{tab2}
\begin{tabular}{|c|c|c|c|c|}
\hline
 &  &   &   &\\ 
Method     & $\epsilon=10^{-2}$ & $\epsilon=10^{-3}$ & $\epsilon=10^{-4}$ & $\epsilon=10^{-5}$\\
 &  &   &   &\\ \hline
  &  &   &   &\\
NI-ADI-BT  & $0.0019$  & $8.7849\times10^{-7}$  & $3.8577\times10^{-7}$  &$3.8576\times10^{-7}$ \\
 &  &   &   &\\ 
NI-ADI-BT$^*$ &  $3.8613\times10^{-7}$  & $3.8579\times10^{-7}$    & $3.8576\times10^{-7}$&  $3.8576\times10^{-7}$  \\
 &  &   &   &\\ 
DD-P-BT    &  $3.8576\times10^{-7}$  & $3.8576\times10^{-7}$   & $3.8576\times10^{-7}$   & $3.8576\times10^{-7}$ \\
 &  &   &   &\\ 
DD-P-BT$^*$   & $3.8576\times10^{-7}$   & $3.8576\times10^{-7}$   &$3.8576\times10^{-7}$   & $3.8576\times10^{-7}$  \\
 &  &   &   &\\ \hline
\end{tabular}
\end{table}\\

\noindent\textbf{Example 2: International Space Station}\\
Consider the $270$-th order MIMO model from the benchmark collection of dynamic systems for testing MOR algorithms \citep{chahlaoui2005benchmark}. For demonstration purposes, the interpolation points consist of two sets of $500$ logarithmically spaced frequencies in $[10^{-3},10^{2}] \cup [-10^{2},-10^{-3}]$ rad/sec, forming $250$ conjugate pairs. The right and left interpolation points do not share any common elements. The free parameter $\epsilon$ is set to $10^{-5}$. For QuadBT, quadrature weights are computed for these nodes using the exponential trapezoidal rule \citep{goseaQuad}. For NI-ADI-BT, the ADI shifts $-\epsilon + j\omega_i$ are used to approximate $P$, and the ADI shifts $-\epsilon + j\nu_i$ are used to approximate $Q$. The samples $G(\epsilon + j\omega_i)$, $G(j\omega_i)$, $G(\epsilon + j\nu_i)$, and $G(j\nu_i)$ are computed numerically using the state-space realization of the international space station model. Since the number of samples is moderate, the Sylvester equations and projected Lyapunov equations can be solved directly. Nevertheless, both the exact solutions and their block diagonally dominant approximations are computed and compared. The non-intrusive implementations, which use block diagonally dominant approximations of the projected Lyapunov equations, are marked with an asterisk in the figures. The Hankel singular values of the ROMs of order $30$ produced by BT, QuadBT, NI-ADI-BT, NI-ADI-BT$^*$, DD-P-BT, and DD-P-BT$^*$ are compared in Fig. \ref{fig2a}. All non-intrusive BT algorithms capture the $20$ most significant Hankel singular values of the original system. Fig. \ref{fig2b} compares the relative error $\frac{\|G(s)-\hat{G}(s)\|_{\mathcal{H}_\infty}}{\|G(s)\|_{\mathcal{H}_\infty}}$ for ROMs of orders $1–30$. As the order of the ROM is increased, the relative error for both intrusive and non-intrusive algorithms declines. Among the non-intrusive methods, QuadBT provides the best performance in this example in terms of capturing the significant Hankel singular values and overall accuracy.
\begin{figure}[!h]
    \centering
    \begin{subfigure}{0.48\textwidth}
        \centering
        \includegraphics[width=\linewidth]{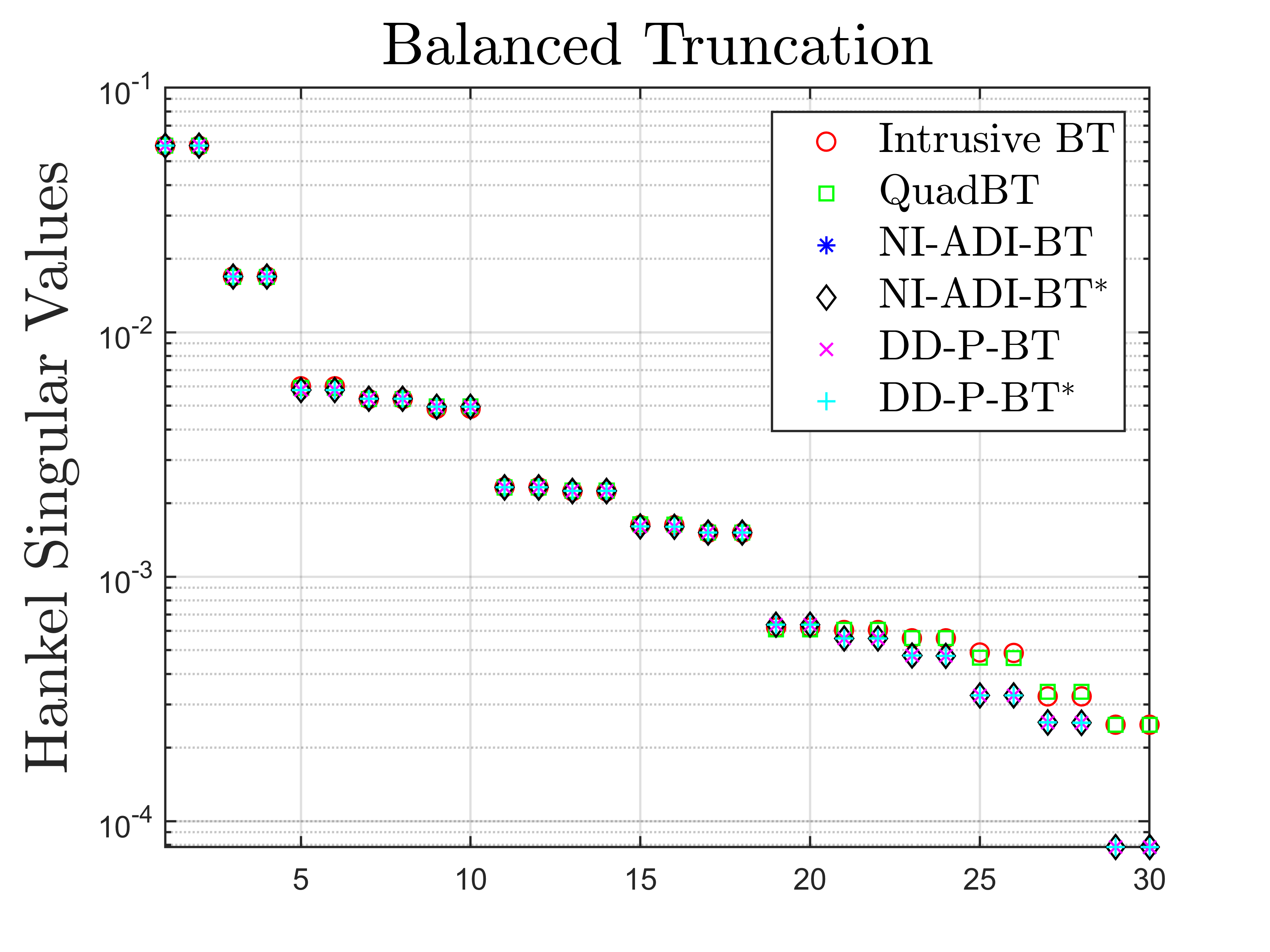}
        \caption{Hankel Singular Values Comparison}\label{fig2a}
    \end{subfigure}
    \hfill
    \begin{subfigure}{0.48\textwidth}
        \centering
        \includegraphics[width=\linewidth]{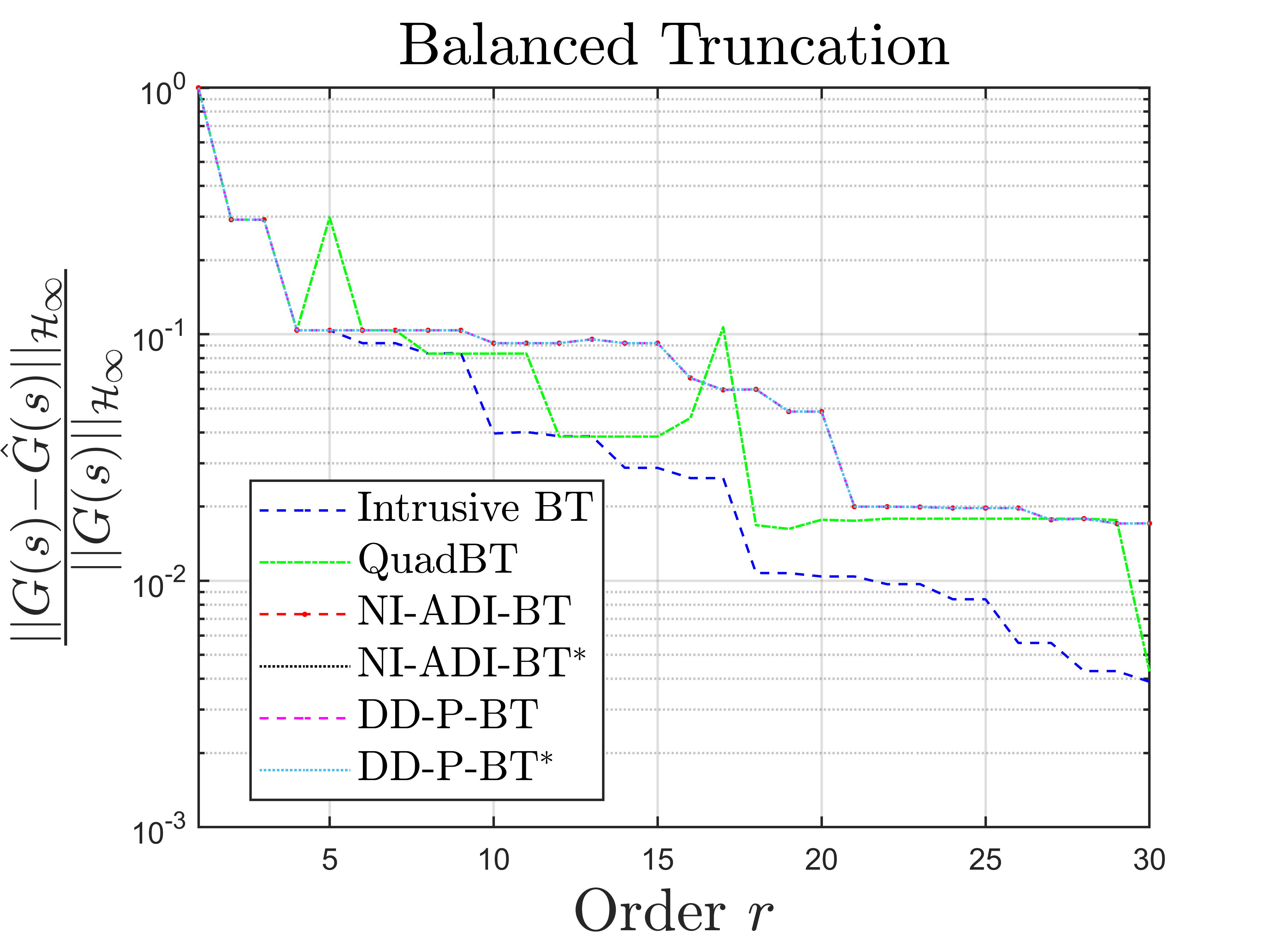}
        \caption{Relative Error Comparison}\label{fig2b}
    \end{subfigure}
    \caption{Performance Comparison between Intrusive and Non-intrusive BT}
\end{figure}\\

\noindent\textbf{Example 3: Los Angeles Building}\\
Consider the \(48\)-th order single-input single-output (SISO) model from the benchmark collection of dynamic systems for testing MOR algorithms \citep{chahlaoui2005benchmark}. For demonstration purposes, the interpolation points consist of two sets of \(500\) logarithmically spaced frequencies in \([10^{-1},10^{3}] \cup [-10^{3},-10^{-1}]\) rad/sec, forming \(250\) conjugate pairs. The right and left interpolation points do not share any common elements. The free parameter \(\epsilon\) is set to \(10^{-5}\), and the parameter \(\gamma\) in \(\mathcal{H}_\infty\)BT is set to \(2.5\). For NI-ADI-LQGBT and NI-ADI-\(\mathcal{H}_\infty\)BT, the RADI shifts \(-\epsilon + j\omega_i\) are used to approximate \(P_{\mathrm{LQG}}\) and \(P_{\mathcal{H}_\infty}\), and the shifts \(-\epsilon + j\nu_i\) are used to approximate \(Q_{\mathrm{LQG}}\) and \(Q_{\mathcal{H}_\infty}\). The samples \(G(\epsilon + j\omega_i)\), \(G(j\omega_i)\), \(G(\epsilon + j\nu_i)\), and \(G(j\nu_i)\) are computed numerically using the state-space realization of the building model. Since the number of samples is moderate, the projected Riccati equations can be solved directly. Nevertheless, both the exact solutions and their block diagonally dominant approximations are computed and compared. The non-intrusive implementations that use block diagonally dominant approximations are marked with an asterisk in the figures. The LQG and \(\mathcal{H}_\infty\) characteristics, defined by \(\sqrt{\lambda_i(P_{\mathrm{LQG}}Q_{\mathrm{LQG}})}\) and \(\sqrt{\lambda_i(P_{\mathcal{H}_\infty}Q_{\mathcal{H}_\infty})}\), respectively, for ROMs of order \(25\) produced by LQGBT, NI-ADI-LQGBT, NI-ADI-LQGBT\(^*\), DD-P-LQGBT, DD-P-LQGBT\(^*\), $\mathcal{H}_\infty$BT, NI-ADI-\(\mathcal{H}_\infty\)BT, NI-ADI-\(\mathcal{H}_\infty\)BT\(^*\), DD-P-\(\mathcal{H}_\infty\)BT, and DD-P-\(\mathcal{H}_\infty\)BT\(^*\) are compared in Fig. \ref{fig3}.
\begin{figure}[!h]
    \centering
    \begin{subfigure}{0.48\textwidth}
        \centering
        \includegraphics[width=\linewidth]{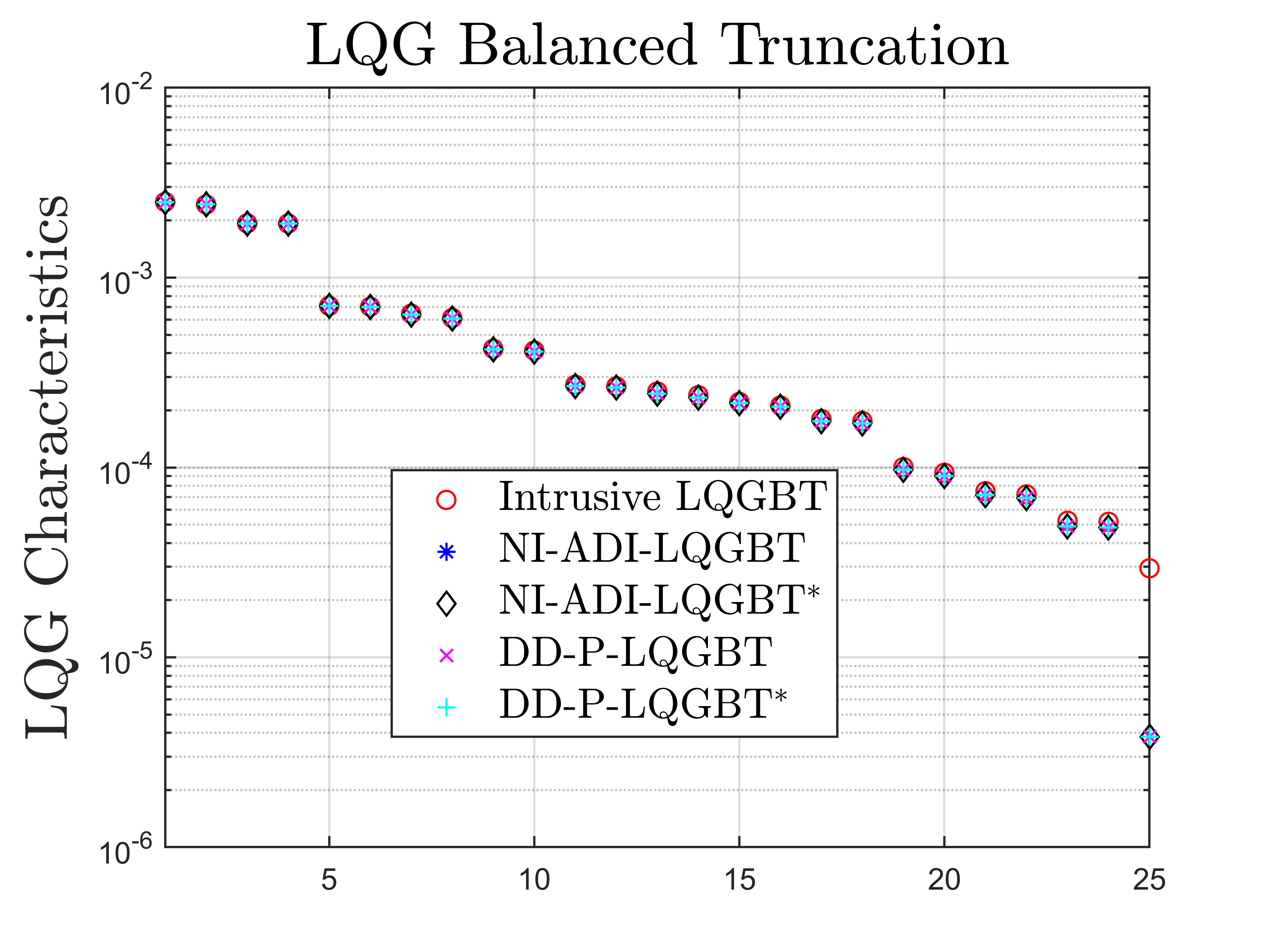}
        \caption{LQG Characteristics Comparison}\label{fig3a}
    \end{subfigure}
    \hfill
    \begin{subfigure}{0.48\textwidth}
        \centering
        \includegraphics[width=\linewidth]{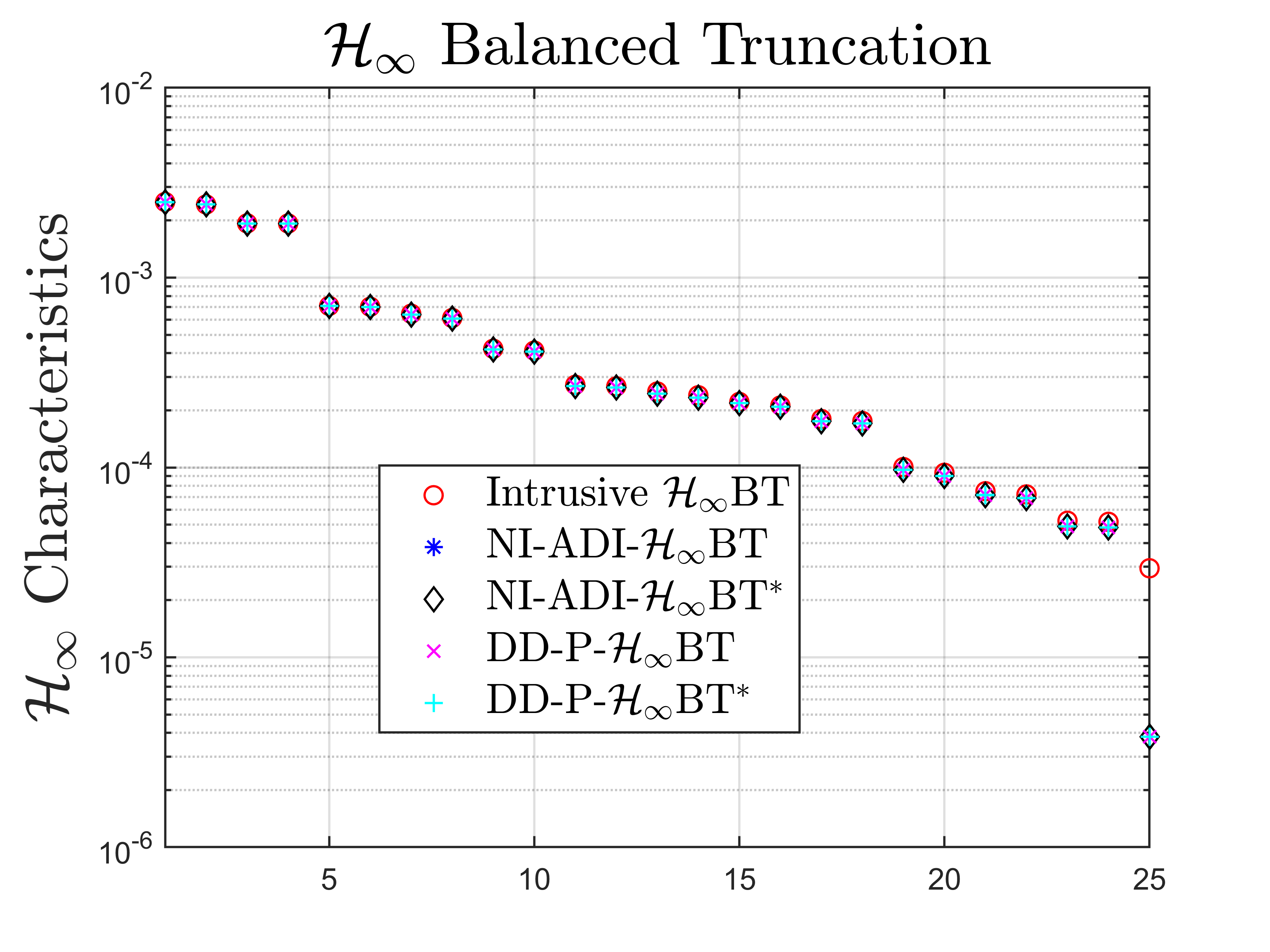}
        \caption{$\mathcal{H}_\infty$ Characteristics Comparison}\label{fig3b}
    \end{subfigure}
    \caption{Comparison between Intrusive and Non-intrusive LQGBT and $\mathcal{H}_\infty$BT}\label{fig3}
\end{figure}
All non-intrusive LQGBT and $\mathcal{H}_\infty$BT algorithms capture the $24$ most significant LQG and $\mathcal{H}_\infty$ characteristics of the original system.  Fig. \ref{fig4} compares the relative error $\frac{|G(s)-\hat{G}(s)|-{\mathcal{H}_\infty}}{|G(s)|_{\mathcal{H}_\infty}}$ for ROMs of orders $1–25$. As the order of the ROM increases, the performance of both intrusive and non-intrusive algorithms becomes comparable.
\begin{figure}[!h]
    \centering
    \begin{subfigure}{0.48\textwidth}
        \centering
        \includegraphics[width=\linewidth]{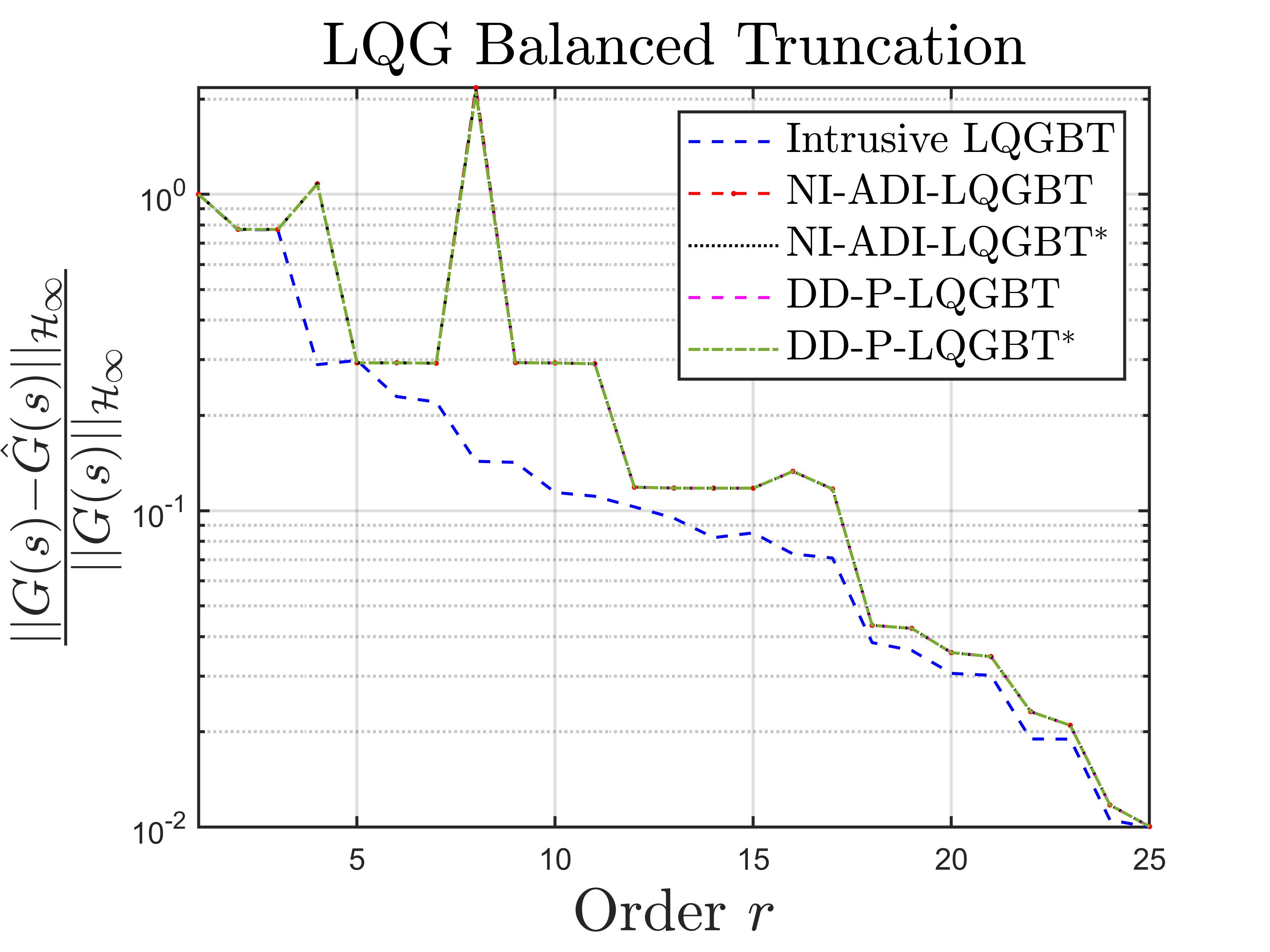}
        \caption{Relative Error Comparison}\label{fig4a}
    \end{subfigure}
    \hfill
    \begin{subfigure}{0.48\textwidth}
        \centering
        \includegraphics[width=\linewidth]{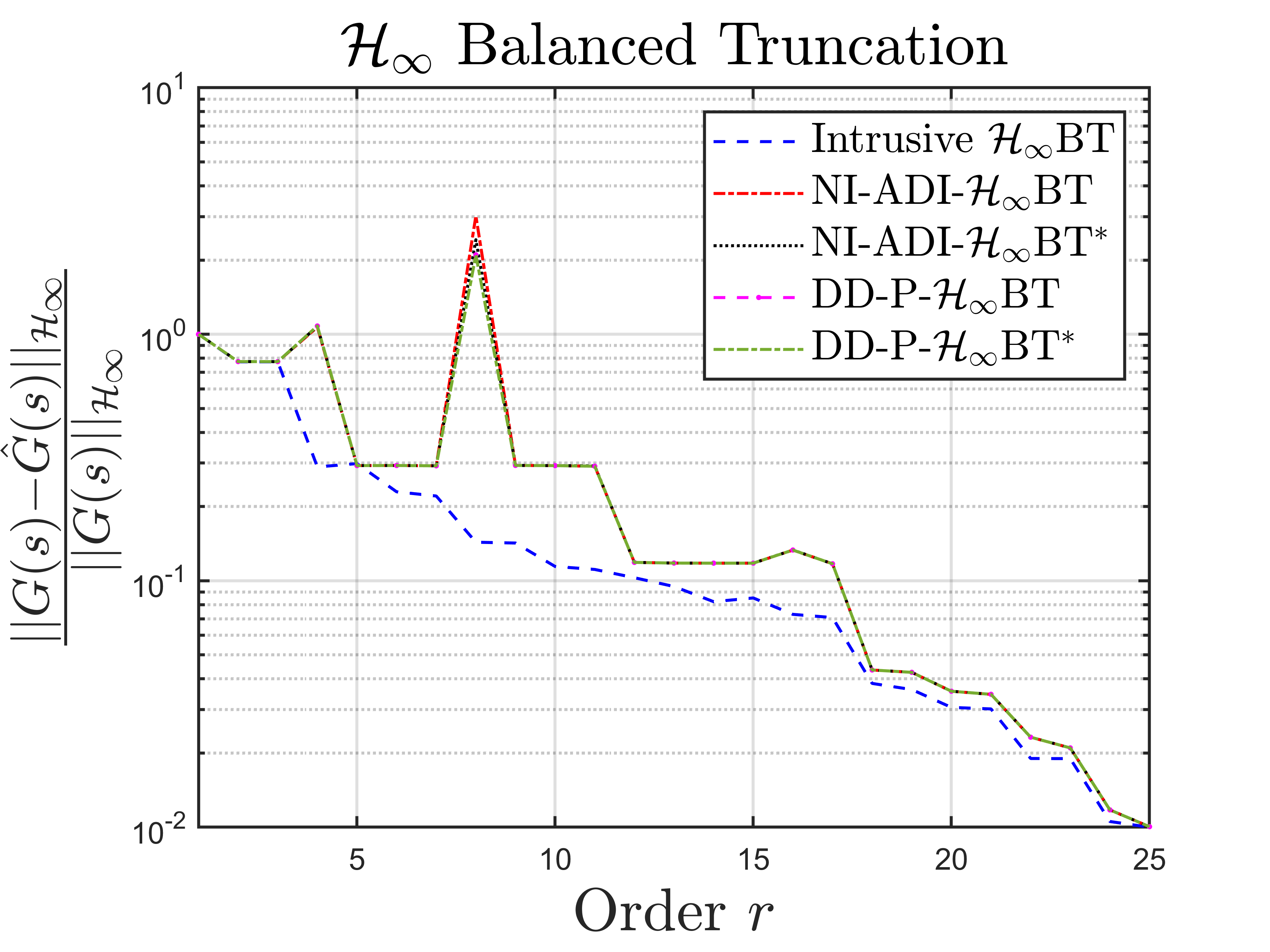}
        \caption{Relative Comparison}\label{fig4b}
    \end{subfigure}
    \caption{Accuracy Comparison between Intrusive and Non-intrusive LQGBT and $\mathcal{H}_\infty$BT}\label{fig4}
\end{figure}\\

\noindent\textbf{Example 4: RLC Ladder}\\
Consider the $400$th-order SISO RLC circuit model from \citep{reiter2023generalizations}. The state-space realization is normalized to ensure $\|G(s)\|_{\mathcal{H}_\infty}<1$, making BRBT applicable. The resulting model is stable, positive-real, bounded-real, and minimum-phase, which guarantees the invertibility of $D$, $DD^T$, $D+D^T$, $I_m-D^TD$, and $I_p-DD^T$. Thus, PRBT, BRBT, SWBT, and BST are applicable to this model and are tested in this example. For demonstration purposes, the interpolation points consist of two sets of \(200\) logarithmically spaced frequencies in \([10^{-1},10^{4}] \cup [-10^{4},-10^{-1}]\) rad/sec, forming \(100\) conjugate pairs. The right and left interpolation points do not share any common elements. The free parameter \(\epsilon\) is set to \(10^{-5}\). For quadrature-based PRBT (QuadPRBT), quadrature-based BRBT (QuadBRBT), and quadrature-based BST (QuadBST), quadrature weights are computed for these nodes using the trapezoidal rule \citep{reiter2023generalizations}. For NI-ADI-PRBT, NI-ADI-BRBT, NI-ADI-SWBT, and NI-ADI-BST, the ADI shifts \(-\epsilon + j\omega_i\) are used to approximate the generalized controllability Gramians, and the shifts \(-\epsilon + j\nu_i\) are used to approximate the generalized observability Gramians. The samples \(G(\epsilon + j\omega_i)\), \(G(j\omega_i)\), \(G(\epsilon + j\nu_i)\), and \(G(j\nu_i)\) are computed numerically using the state-space realization of the RLC model. Since the number of samples is moderate, the projected Lyapunov and Riccati equations can be solved directly. Nevertheless, both the exact solutions and their block diagonally dominant approximations are computed and compared. The non-intrusive implementations that use block diagonally dominant approximations are marked with an asterisk in the figures. The quantities $\sqrt{\lambda_i(P_{\text{PR}}Q_{\text{PR}})}$, $\sqrt{\lambda_i(P_{\text{BR}}Q_{\text{BR}})}$, $\sqrt{\lambda_i(PQ_{\text{SW}})}$, and $\sqrt{\lambda_i(PQ_{\text{S}})}$ are referred to as Hankel-like singular values. Figures \ref{fig5}-\ref{fig8} compare the Hankel-like singular values and the relative error $\frac{\|G(s)-\hat{G}(s)\|_{\mathcal{H}_\infty}}{\|G(s)\|_{\mathcal{H}_\infty}}$. It can be seen that the $25^{th}$-order ROMs generated by intrusive methods and their non-intrusive counterparts accurately capture the $20$ most dominant Hankel-like singular values. Moreover, the non-intrusive algorithms achieve accuracy comparable to that of the intrusive methods for ROMs of orders ranging from $1$ to $25$.
\begin{figure}[!h]
    \centering
    \begin{subfigure}{0.48\textwidth}
        \centering
        \includegraphics[width=\linewidth]{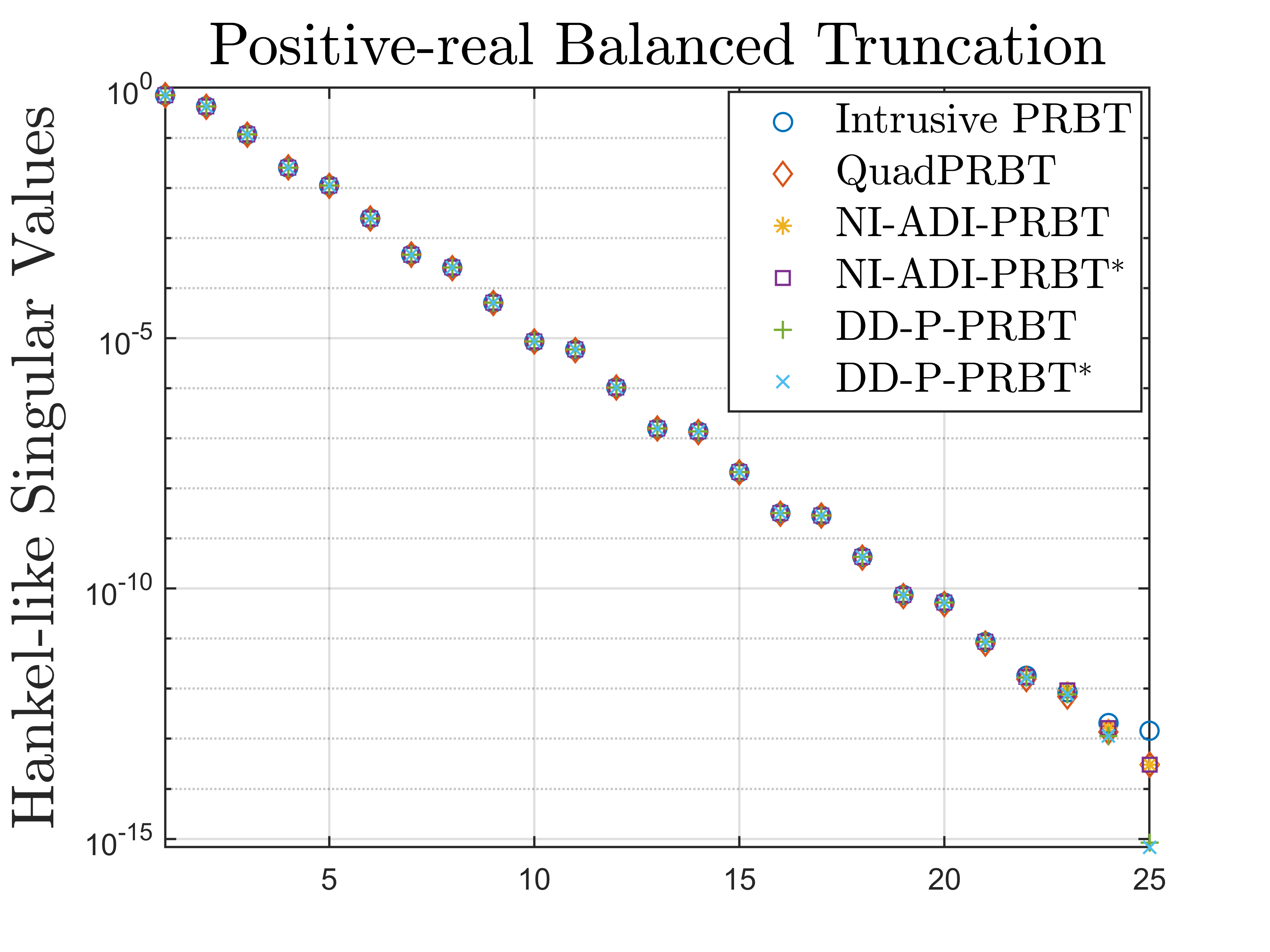}
        \caption{Hankel-like Singular Values Comparison}\label{fig5a}
    \end{subfigure}
    \hfill
    \begin{subfigure}{0.48\textwidth}
        \centering
        \includegraphics[width=\linewidth]{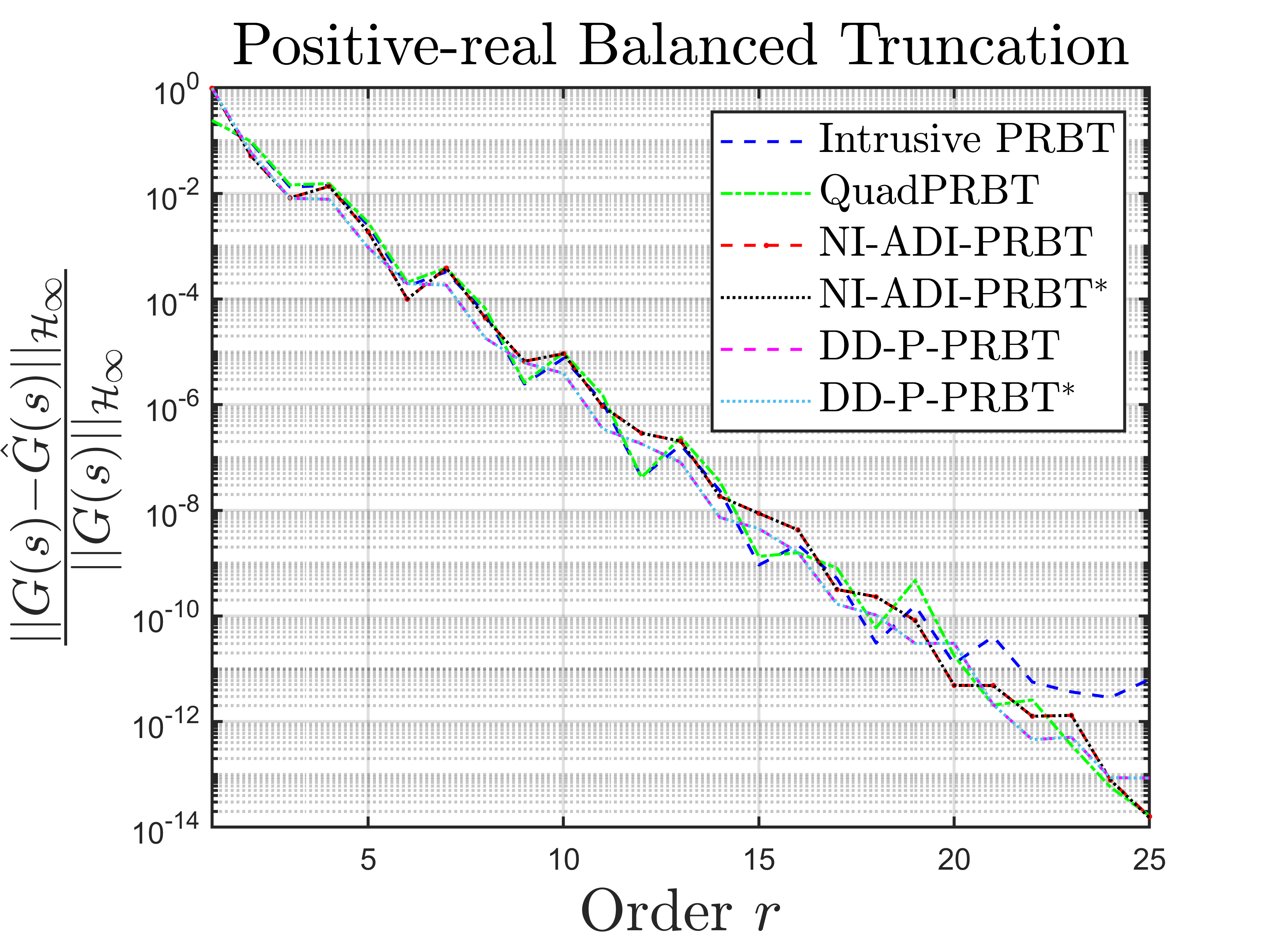}
        \caption{Relative Error Comparison}\label{fig5b}
    \end{subfigure}
    \caption{Performance Comparison between Intrusive and Non-intrusive PRBT}\label{fig5}
\end{figure}
\begin{figure}[!h]
    \centering
    \begin{subfigure}{0.48\textwidth}
        \centering
        \includegraphics[width=\linewidth]{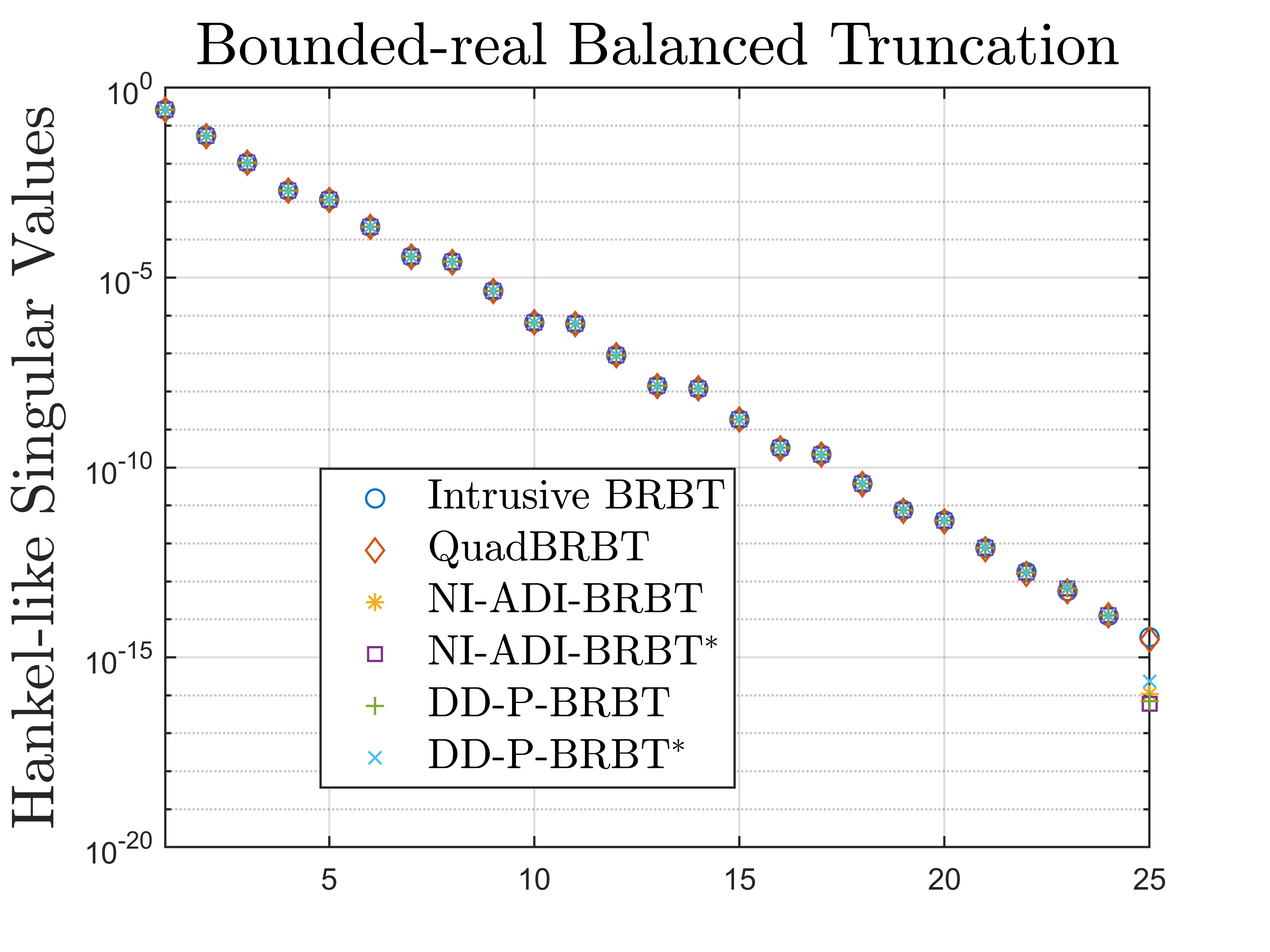}
        \caption{Hankel-like Singular Values Comparison}\label{fig6a}
    \end{subfigure}
    \hfill
    \begin{subfigure}{0.48\textwidth}
        \centering
        \includegraphics[width=\linewidth]{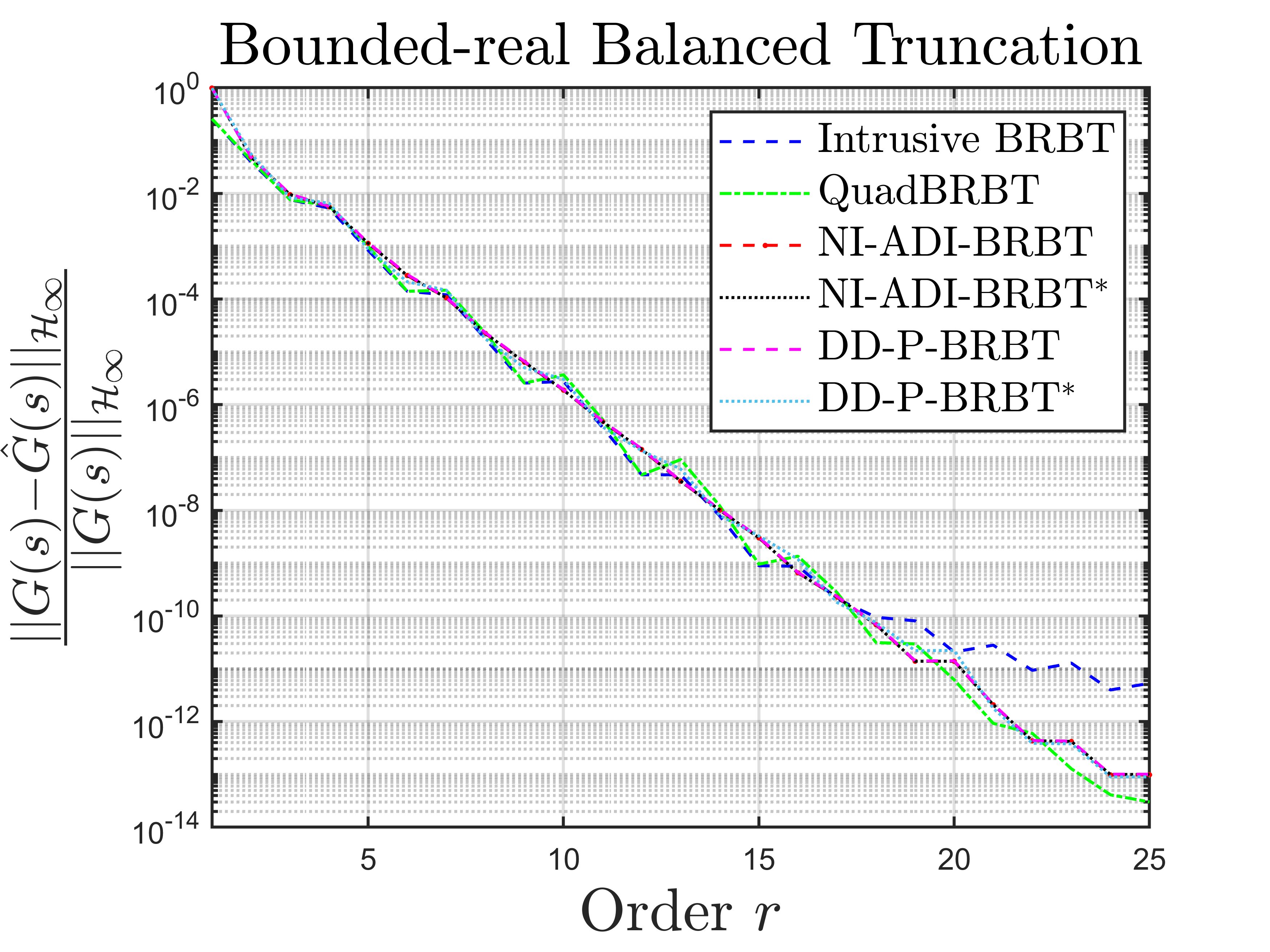}
        \caption{Relative Error Comparison}\label{fig6b}
    \end{subfigure}
    \caption{Performance Comparison between Intrusive and Non-intrusive BRBT}\label{fig6}
\end{figure}
\begin{figure}[!h]
    \centering
    \begin{subfigure}{0.48\textwidth}
        \centering
        \includegraphics[width=\linewidth]{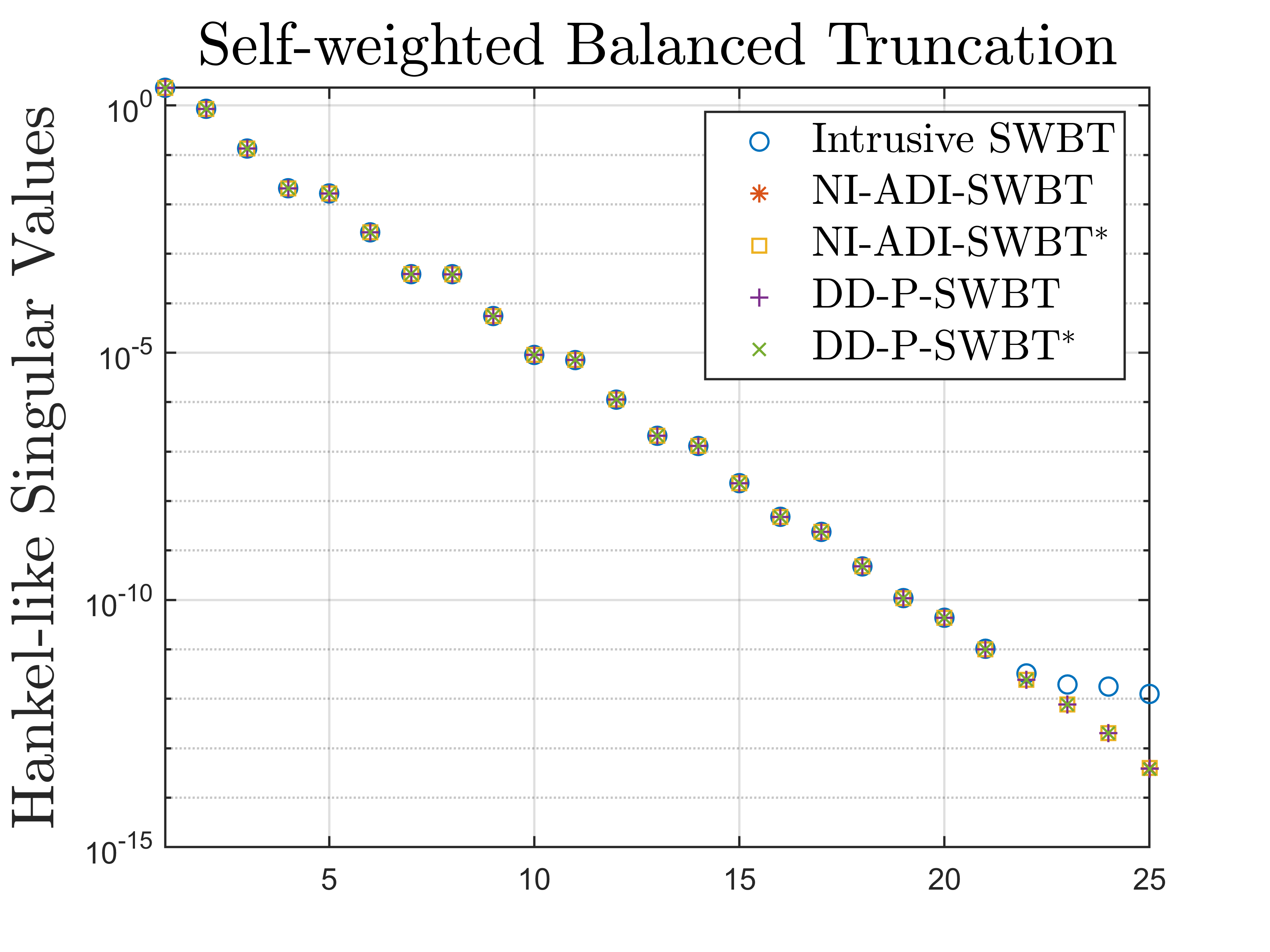}
        \caption{Hankel-like Singular Values Comparison}\label{fig7a}
    \end{subfigure}
    \hfill
    \begin{subfigure}{0.48\textwidth}
        \centering
        \includegraphics[width=\linewidth]{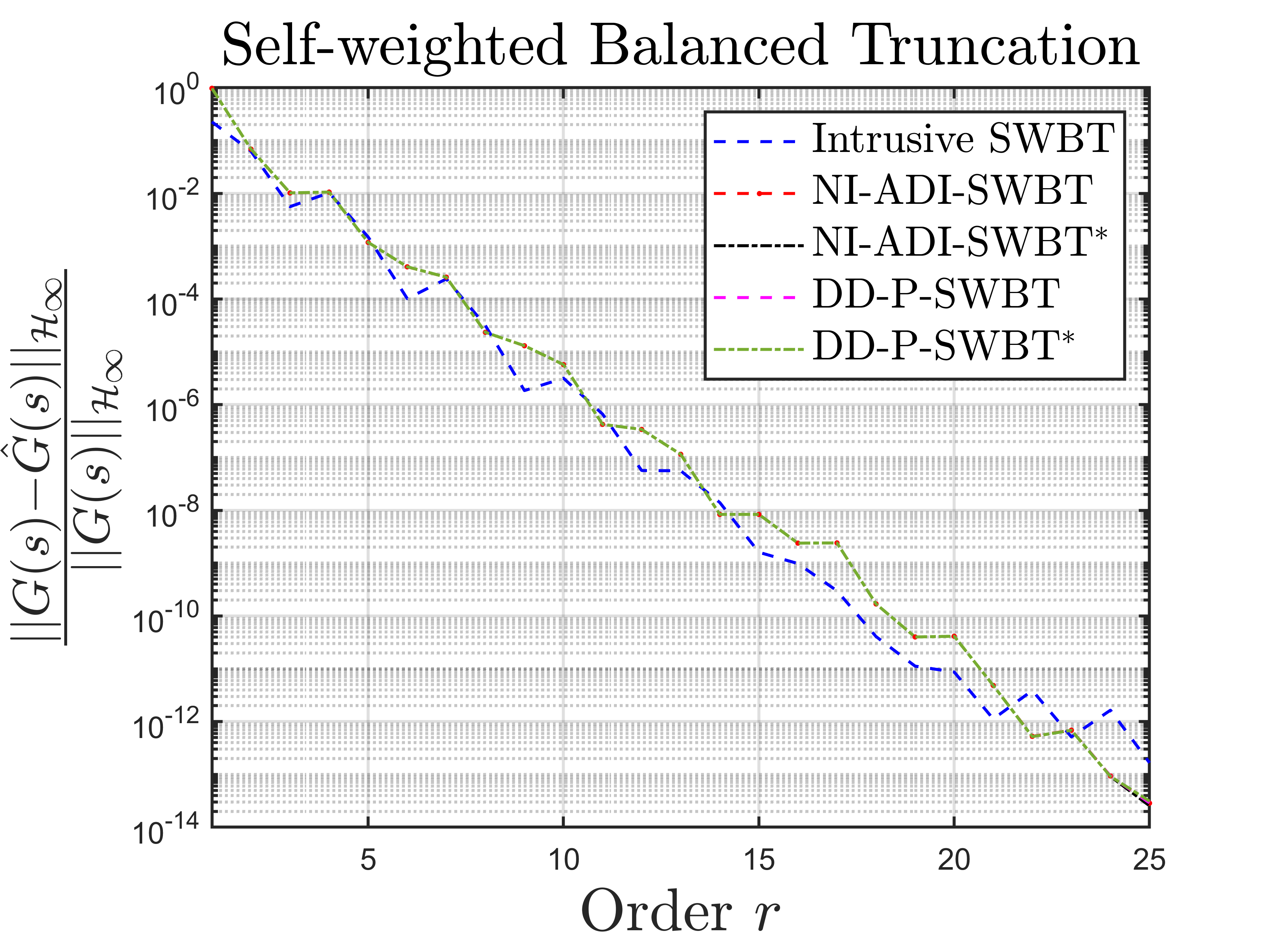}
        \caption{Relative Error Comparison}\label{fig7b}
    \end{subfigure}
    \caption{Performance Comparison between Intrusive and Non-intrusive SWBT}\label{fig7}
\end{figure}
\begin{figure}[!h]
    \centering
    \begin{subfigure}{0.48\textwidth}
        \centering
        \includegraphics[width=\linewidth]{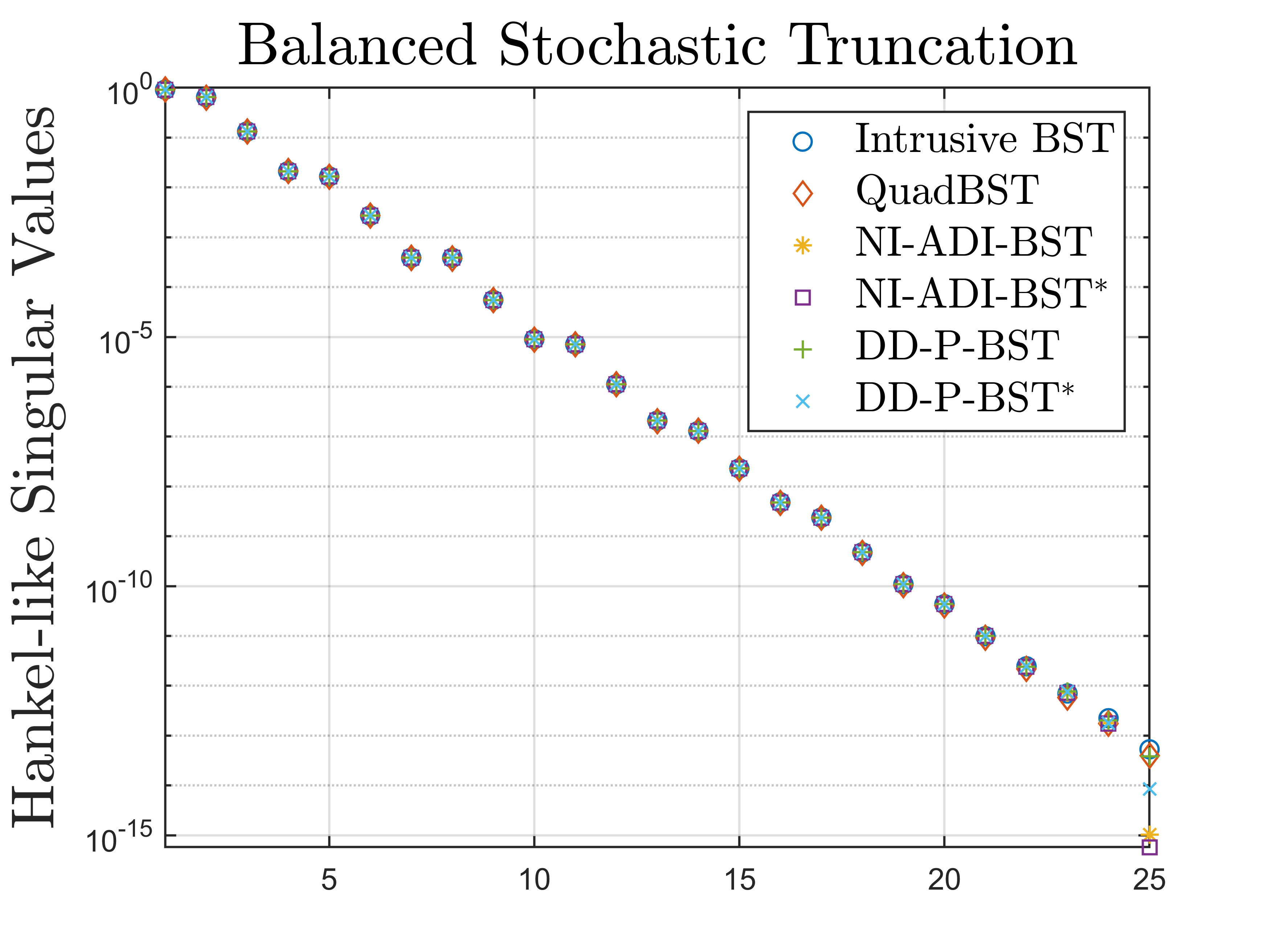}
        \caption{Hankel-like Singular Values Comparison}\label{fig8a}
    \end{subfigure}
    \hfill
    \begin{subfigure}{0.48\textwidth}
        \centering
        \includegraphics[width=\linewidth]{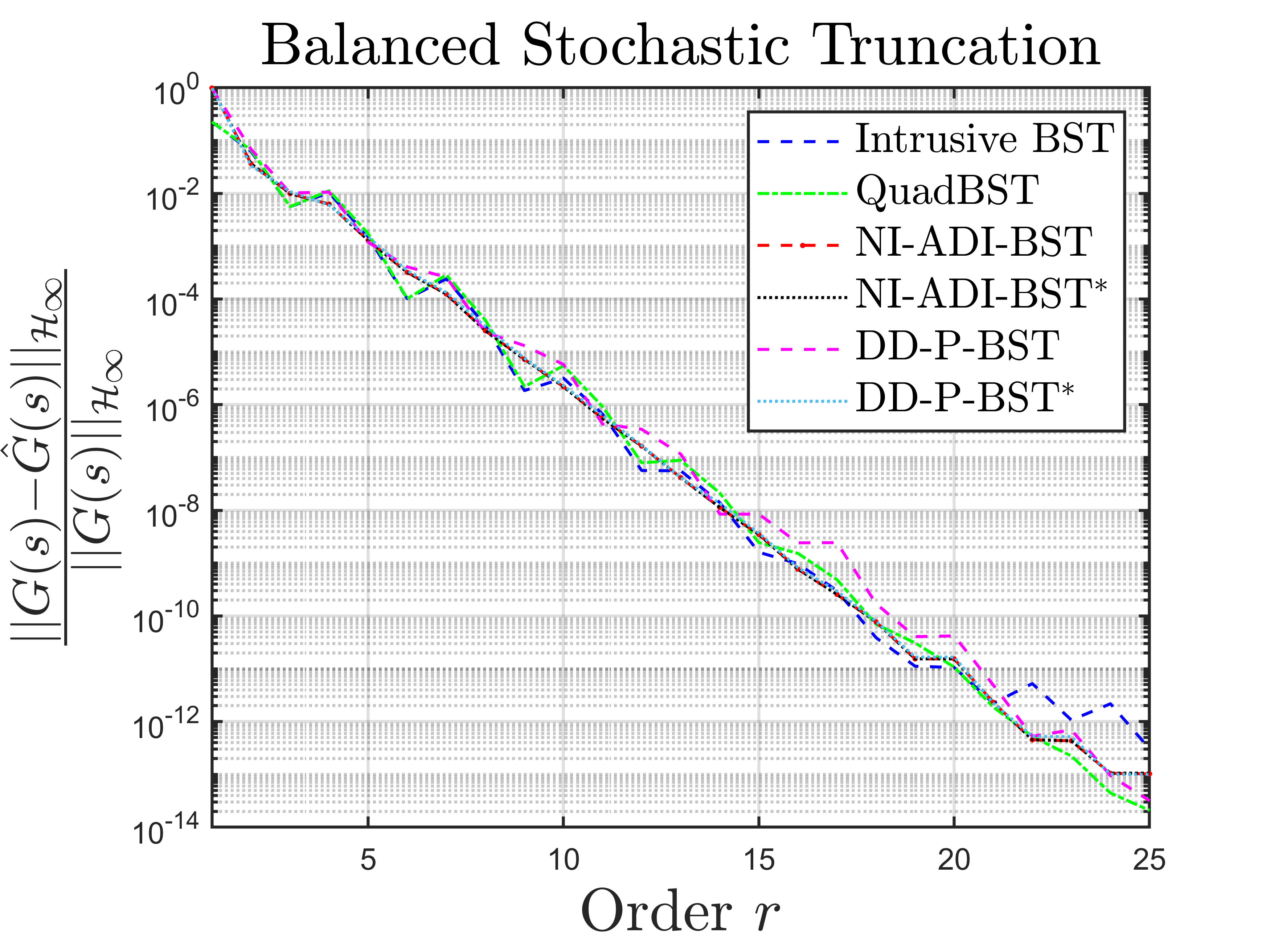}
        \caption{Relative Error Comparison}\label{fig8b}
    \end{subfigure}
    \caption{Performance Comparison between Intrusive and Non-intrusive BST}\label{fig8}
\end{figure}
\section{Conclusion}\label{sec6}
This paper shows that replacing the Gramians in various generalizations of BT—such as LQGBT, \(\mathcal{H}_\infty\)BT, PRBT, BRBT, SWBT, and BST—with their Krylov subspace-based approximations in BSA leads to non-intrusive algorithms. In these methods, ROMs can be constructed from transfer function samples without access to the state-space realization of the original system. Moreover, when the shifts (interpolation points) in the Krylov subspace framework are chosen on the imaginary axis of the \(s\)-plane, the resulting non-intrusive BT approximations become data-driven, since transfer function samples on the imaginary axis can be measured experimentally. This resolves the issue in quadrature-based approximations of these BT generalizations, which require samples of spectral factorizations on the imaginary axis that cannot be measured in practice. Potential numerical issues in the implementation are also discussed, along with remedies to address them. The performance of the proposed non-intrusive BT algorithms is evaluated against their intrusive counterparts using benchmark dynamical system models for MOR. Numerical results show that the proposed methods perform comparably to the intrusive ones.
\section*{Appendix A}
\begin{proof}
By the mixed-product property of Kronecker products \citep{hardy2019matrix}, $(A \otimes B)(C \otimes D) = (AC) \otimes (BD)$, the product $\zeta L_v$ is given by
\begin{equation}
    \zeta L_v = \left( (2\epsilon \mathbf{1}_v) \otimes I_m \right) \left( \mathbf{1}_v^T \otimes I_m \right) = 2\epsilon (\mathbf{1}_v \mathbf{1}_v^T) \otimes I_m.
\end{equation}
The matrix $\mathbf{1}_v \mathbf{1}_v^T$ is a $v \times v$ matrix of ones. Thus, $\zeta L_v$ is a block matrix where every $(i,k)$-th block is $2\epsilon I_m$. The matrix $S_v$ is block diagonal with $(i,i)$-th block $(\epsilon + j\omega_i)I_m$. Consequently, the blocks of $\hat{A} = S_v - \zeta L_v$ are
\begin{equation}
    [\hat{A}]_{ik} = \begin{cases} 
    (j\omega_i - \epsilon)I_m & \text{if } i = k, \\
    -2\epsilon I_m & \text{if } i \neq k.
    \end{cases}
\end{equation}
The diagonal part is $\hat{A}_d = \operatorname{blkdiag}([\hat{A}]_{11}, \dots, [\hat{A}]_{vv})$, and the error matrix is $\hat{E} = \hat{A} - \hat{A}_d$. Explicitly,
\begin{equation}
    \hat{E} = -2\epsilon \left( (\mathbf{1}_v \mathbf{1}_v^T - I_v) \otimes I_m \right).
\end{equation}
Using the property $\|X \otimes Y\|_2 = \|X\|_2 \|Y\|_2$ and noting $\|I_m\|_2 = 1$, we have
\begin{equation}
    \|\hat{E}\|_2 = 2\epsilon \left\| \mathbf{1}_v \mathbf{1}_v^T - I_v \right\|_2.
\end{equation}
The matrix $M = \mathbf{1}_v \mathbf{1}_v^T - I_v$ is real and symmetric. Therefore, its spectral norm equals its spectral radius, $\|M\|_2 = \rho(M)$. The eigenvalues of $\mathbf{1}_v \mathbf{1}_v^T$ are $v$ (multiplicity 1) and $0$ (multiplicity $v-1$). Thus, the eigenvalues of $M$ are $v-1$ (multiplicity 1) and $-1$ (multiplicity $v-1$). For $v=1$, the only eigenvalue is $0$, so $\rho(M)=0$. For $v \ge 2$, the spectral radius is $\max(|v-1|, |-1|) = v-1$. In both cases, $\|M\|_2 = v-1$. Therefore,
\begin{equation}
    \|\hat{A} - \hat{A}_d\|_2 = 2\epsilon(v-1).
\end{equation}
This proves the first claim. The bound $\|\hat{E}\|_2 \le \delta$ is satisfied if $\epsilon \le \frac{\delta}{2(v-1)}$ for $v \ge 2$.

A block matrix is strictly block diagonally dominant with respect to the spectral norm if for every block row $i$:
\begin{equation}
    \|[\hat{A}]_{ii}\|_2 > \sum_{k \neq i} \|[\hat{A}]_{ik}\|_2.
\end{equation}
Substituting the block norms:
\begin{align}
    \|[\hat{A}]_{ii}\|_2 &= \|(j\omega_i - \epsilon)I_m\|_2 = \sqrt{\omega_i^2 + \epsilon^2}, \\
    \sum_{k \neq i} \|[\hat{A}]_{ik}\|_2 &= \sum_{k \neq i} \|-2\epsilon I_m\|_2 = 2\epsilon(v-1).
\end{align}
The condition becomes $\sqrt{\omega_i^2 + \epsilon^2} > 2\epsilon(v-1)$. Since $\omega_i \neq 0$ and $\epsilon > 0$, both sides are strictly positive, allowing us to square the inequality:
\begin{equation}
    \omega_i^2 + \epsilon^2 > 4\epsilon^2(v-1)^2 \implies \omega_i^2 > \epsilon^2 \left( 4(v-1)^2 - 1 \right).
\end{equation}
For $v \ge 2$, the term $4(v-1)^2 - 1 > 0$. Solving for $\epsilon$ yields $\epsilon < \frac{|\omega_i|}{\sqrt{4(v-1)^2 - 1}}$. To satisfy this for all $i$, we require $\epsilon < \frac{\omega_{\min}}{\sqrt{4(v-1)^2 - 1}}$. This proves the second claim.

Combining the conditions from Part 1 and Part 2, selecting $\epsilon$ strictly less than the dominance bound and satisfying the error tolerance ensures both properties. Specifically, choosing $\epsilon$ according to \eqref{eq:epsilon_bound} (with strict inequality for the dominance term) guarantees $\|\hat{A} - \hat{A}_d\|_2 \le \delta$ and that $\hat{A}$ is strictly block diagonally dominant. As $\epsilon \to 0$, $\|\hat{E}\|_2 \to 0$ linearly, rendering $\hat{A}$ essentially diagonal.
\end{proof}
\section*{Appendix B}
\begin{proof}
The eigenvalues of $S_v$ are $\epsilon + j\omega_i$. Since $\epsilon > 0$, the spectra of $-S_v^*$ and $S_v$ are disjoint, guaranteeing a unique solution to \eqref{Qv_LQG}. Let $Q_v$ be partitioned into $m \times m$ blocks $[Q_v]_{ij}$. Since $S_v$ is block diagonal with blocks $S_{ii} = (\epsilon + j\omega_i)I_m$, the $(i,j)$-th block of the Lyapunov equation decouples as:
\begin{equation}
    S_{ii}^* [Q_v]_{ij} + [Q_v]_{ij} S_{jj} = M_{ij}.
\end{equation}
Substituting the scalar forms yields:
\begin{equation}
    (\epsilon - j\omega_i) [Q_v]_{ij} + [Q_v]_{ij} (\epsilon + j\omega_j) = M_{ij}.
\end{equation}
Since scalar multiples of the identity commute with any matrix, factoring out $[Q_v]_{ij}$ gives:
\begin{equation}
    [Q_v]_{ij} \left( (\epsilon - j\omega_i) + (\epsilon + j\omega_j) \right) = M_{ij},
\end{equation}
which simplifies to \eqref{eq:Q_block_sol}. Note that $L_v^* L_v$ has identity blocks on the diagonal, implying $M_{ii} \succeq I_m$. Given $\operatorname{Re}(\lambda(S_v)) > 0$ and the observability assumption in Theorem \ref{Theorem1}, $Q_v$ is Hermitian positive definite.

Strict block diagonal dominance requires $\| [Q_v]_{ii} \|_2 > \sum_{j \neq i} \| [Q_v]_{ij} \|_2$.
From Part 1, the norms are:
\begin{align}
    \| [Q_v]_{ii} \|_2 &= \frac{\| M_{ii} \|_2}{2\epsilon}, \\
    \| [Q_v]_{ij} \|_2 &= \frac{\| M_{ij} \|_2}{\sqrt{4\epsilon^2 + (\omega_j - \omega_i)^2}} \le \frac{\| M_{ij} \|_2}{|\omega_j - \omega_i|} \quad (i \neq j).
\end{align}
Note that $M = L_v^* L_v + \hat{C}^* \hat{C}$. Since $L_v^* L_v = \mathbf{1}_v \mathbf{1}_v^T \otimes I_m$, its diagonal blocks are $I_m$. Since $\hat{C}^* \hat{C} \succeq 0$, the diagonal blocks of $M$ satisfy $M_{ii} \succeq I_m$, implying $\| M_{ii} \|_2 \ge 1$.
The dominance condition is satisfied if:
\begin{equation}
    \frac{\| M_{ii} \|_2}{2\epsilon} > \sum_{j \neq i} \frac{\| M_{ij} \|_2}{|\omega_j - \omega_i|}.
\end{equation}
Rearranging for $\epsilon$ yields \eqref{eq:epsilon_dominance_Q}. Since $\|M_{ii}\|_2 \ge 1$ and the frequencies are distinct, the right-hand side is strictly positive (or infinite), ensuring such an $\epsilon^*$ exists.

Taking the limit $\epsilon \to 0$ in \eqref{eq:Q_block_sol} for $i \neq j$:
\begin{equation}
    \lim_{\epsilon \to 0} \frac{M_{ij}}{2\epsilon + j(\omega_j - \omega_i)} = \frac{M_{ij}}{j(\omega_j - \omega_i)},
\end{equation}
where we assume $M_{ij}$ remains bounded and continuous as $\epsilon \to 0$. For the diagonal terms ($i=j$), the denominator is $2\epsilon$, so $\| [Q_v]_{ii} \|_2 = \frac{\|M_{ii}\|_2}{2\epsilon} \to \infty$ as $\epsilon \to 0$ since $\|M_{ii}\|_2 \ge 1$.
\end{proof}
\section*{Appendix C}
\begin{proof}
Equation \eqref{eq:Tv_sylvester} is rewritten as $\mathcal{A}_\epsilon T_v - T_v S_v = -\zeta \mathcal{Q} L_v$, where $\mathcal{A}_\epsilon = S_v - \zeta (L_v + \mathcal{P})$. The $(i,j)$-th block equation is given by:
\begin{equation}
    \sum_{k=1}^v [\mathcal{A}_\epsilon]_{ik} [T_v]_{kj} - [T_v]_{ij} [S_v]_{jj} = -[\zeta \mathcal{Q} L_v]_{ij}.
\end{equation}
Substituting the explicit block structures derived from the definitions yields: $[S_v]_{jj} = (\epsilon + j\omega_j)I_m$, $[\mathcal{A}_\epsilon]_{ii} = (j\omega_i - \epsilon)I_m - 2\epsilon \mathcal{P}_i$, $[\mathcal{A}_\epsilon]_{ik} = -2\epsilon (I_m + \mathcal{P}_k)$ for $k \neq i$, and $[\zeta \mathcal{Q} L_v]_{ij} = 2\epsilon \mathcal{Q}$ for all $i,j$.
The block equation becomes:
\begin{equation}
    \left( (j\omega_i - \epsilon)I_m - 2\epsilon \mathcal{P}_i \right) [T_v]_{ij} - [T_v]_{ij} (\epsilon + j\omega_j)I_m - 2\epsilon \sum_{k \neq i} (I_m + \mathcal{P}_k) [T_v]_{kj} = -2\epsilon \mathcal{Q}.
\end{equation}
Grouping terms involving $[T_v]_{ij}$:
\begin{equation}
    [T_v]_{ij} \left( j(\omega_i - \omega_j) - 2\epsilon \right) - 2\epsilon \mathcal{P}_i [T_v]_{ij} - 2\epsilon \sum_{k \neq i} (I_m + \mathcal{P}_k) [T_v]_{kj} = -2\epsilon \mathcal{Q}.
\end{equation}
For $i \neq j$, the term $j(\omega_i - \omega_j)I_m$ is invertible since $\omega_i \neq \omega_j$. As $\epsilon \to 0$, the equation is dominated by $j(\omega_i - \omega_j) [T_v]_{ij} = O(\epsilon)$. Thus, $\| [T_v]_{ij} \|_2 \le \frac{2\epsilon \|\mathcal{Q}\|_2}{|\omega_i - \omega_j|} + O(\epsilon^2)$, proving \eqref{eq:Tv_offdiag_limit}.

For $i = j$, the term $j(\omega_i - \omega_j)$ vanishes. The equation becomes:
\begin{equation}
    -2\epsilon (I_m + \mathcal{P}_i) [T_v]_{ii} - 2\epsilon \sum_{k \neq i} (I_m + \mathcal{P}_k) [T_v]_{ki} = -2\epsilon \mathcal{Q}.
\end{equation}
Dividing by $-2\epsilon$ (valid since $\epsilon > 0$):
\begin{equation}
    (I_m + \mathcal{P}_i) [T_v]_{ii} + \sum_{k \neq i} (I_m + \mathcal{P}_k) [T_v]_{ki} = \mathcal{Q}.
\end{equation}
From the off-diagonal case, $[T_v]_{ki} = O(\epsilon)$ for $k \neq i$. Thus, the summation term is $O(\epsilon)$. Taking the limit $\epsilon \to 0$:
\begin{equation}
    (I_m + \mathcal{P}_i) \lim_{\epsilon \to 0} [T_v]_{ii} = \mathcal{Q}.
\end{equation}
Given $(I_m + \mathcal{P}_i)$ is invertible, \eqref{eq:Tv_diag_limit} follows.

Strict block diagonal dominance requires $\| [T_v]_{ii} \|_2 > \sum_{j \neq i} \| [T_v]_{ij} \|_2$.
From Part 1, $\| [T_v]_{ii} \|_2 \to \| (I_m + \mathcal{P}_i)^{-1} \mathcal{Q} \|_2 > 0$.
Meanwhile, $\| [T_v]_{ij} \|_2 \le \frac{2\epsilon \|\mathcal{Q}\|_2}{|\omega_i - \omega_j|} + O(\epsilon^2)$.
Thus, for sufficiently small $\epsilon$, the diagonal norm dominates the sum of off-diagonal norms. The bound \eqref{eq:epsilon_dominance_T} is derived by retaining the leading order terms and ensuring the inequality holds strictly.

Let $\tilde{T}_v = \operatorname{blkdiag}([T_v]_{11}, \dots, [T_v]_{vv})$. The error matrix $E_T = T_v - \tilde{T}_v$ contains only off-diagonal blocks. Since each off-diagonal block is bounded by $C_{ij} \epsilon$, the spectral norm of the block matrix is bounded by $K \epsilon$ for some constant $K$ dependent on dimensions and frequency separations.
\end{proof}
\section*{Appendix D}
The Sylvester equation \eqref{Xp} has a unique solution because the spectra of $S_p^{*}$ and $-S_v$ are disjoint: $\operatorname{spec}(S_p^{*}) = \{\epsilon - j\omega_i\}_{i=1}^{v}$ (each eigenvalue with multiplicity $m$) and $\operatorname{spec}(-S_v) = \{-j\omega_i\}_{i=1}^{v}$ (each with multiplicity $m$) satisfy $\operatorname{Re}(\epsilon - j\omega_i) = \epsilon > 0 = \operatorname{Re}(-j\omega_i)$ for all $i$.

Since $S_v$, $S_p^{*}$, and $L_v^TL_v = \mathbf{1}_{v}\mathbf{1}_{v}^{\mathsf{T}} \otimes I_m$ share the Kronecker structure $\cdot \otimes I_m$, we seek a solution of the form $X_p = Y \otimes I_m$ with $Y \in \mathbb{C}^{v \times v}$. Substituting into \eqref{Xp} and exploiting the mixed-product property of Kronecker products \citep{hardy2019matrix} yields the reduced equation
\begin{equation}
\operatorname{diag}(\epsilon - j\omega_1, \dots, \epsilon - j\omega_{v}) \, Y + Y \, \operatorname{diag}(j\omega_1, \dots, j\omega_{v}) = \mathbf{1}_{v}\mathbf{1}_{v}^{\mathsf{T}}.
\end{equation}
Entry-wise for $i,j = 1,\dots,v$:
\begin{equation}
(\epsilon - j\omega_i) Y_{ij} + Y_{ij} (j\omega_j) = 1
\;\Longrightarrow\;
Y_{ij} = \frac{1}{\epsilon + j(\omega_j - \omega_i)},
\end{equation}
which is well-defined for all $\epsilon > 0$ since $\operatorname{Re}(\epsilon + j(\omega_j - \omega_i)) = \epsilon > 0$.

For diagonal entries ($i=j$), $Y_{ii} = 1/\epsilon$. For off-diagonal entries ($i \neq j$),
\begin{equation}
|Y_{ij}| = \frac{1}{\sqrt{\epsilon^2 + (\omega_j - \omega_i)^2}} \leq \frac{1}{|\omega_j - \omega_i|} \leq \frac{1}{\Delta_{\min}}.
\end{equation}
Define $E := Y - \frac{1}{\epsilon}I_{v}$, so $E_{ii}=0$ and $E_{ij}=Y_{ij}$ for $i \neq j$. Using the Frobenius norm identity $\|A \otimes B\|_F = \|A\|_F \|B\|_F$,
\begin{align}
\|X_p - \widetilde{X}_p\|_F &= \|E \otimes I_m\|_F = \|E\|_F \, \|I_m\|_F = \|E\|_F \sqrt{m}, \\
\|E\|_F^2 &= \sum_{i \neq j} |Y_{ij}|^2 \leq \sum_{i \neq j} \frac{1}{\Delta_{\min}^2} = \frac{v(v-1)}{\Delta_{\min}^2}
\;\Longrightarrow\;\nonumber\\
&\hspace*{2cm}\|E\|_F \leq \frac{\sqrt{v(v-1)}}{\Delta_{\min}}, \\
\|\widetilde{X}_p\|_F &= \frac{1}{\epsilon} \|I_{mv}\|_F = \frac{\sqrt{mv}}{\epsilon}.
\end{align}
Therefore,
\begin{equation}
\frac{\|X_p - \widetilde{X}_p\|_F}{\|\widetilde{X}_p\|_F}
\leq \frac{\sqrt{v(v-1)}/\Delta_{\min} \cdot \sqrt{m}}{\sqrt{mv}/\epsilon}
= \frac{\epsilon \sqrt{v-1}}{\Delta_{\min}},
\end{equation}
proving \eqref{eq:relative_error}. The tolerance condition follows immediately.

For diagonal dominance, observe that for each row $i$,
\begin{equation}
|Y_{ii}| - \sum_{j \neq i} |Y_{ij}| \geq \frac{1}{\epsilon} - \sum_{j \neq i} \frac{1}{|\omega_j - \omega_i|}
\geq \frac{1}{\epsilon} - \frac{v-1}{\Delta_{\min}}.
\end{equation}
This quantity is positive when $\epsilon < \Delta_{\min}/(v-1)$, establishing strict diagonal dominance of $Y$ and hence of $X_p = Y \otimes I_m$.
\section*{Appendix E}
\begin{proof}
From Proposition \ref{prop_BT_diag}, $X_p$ is strictly diagonally dominant for $\epsilon < \Delta_{\min}/(v-1)$. The bound \eqref{eq:epsilon_bound_gramian} satisfies this condition since $v \ge 2$. Thus, $X_p$ is invertible, and $\hat{B}$ is well-defined.
From Proposition \ref{prop_BT_diag}, $X_p = \frac{1}{\epsilon} I_{vm} + \mathcal{O}(1)$. Using the Neumann series expansion \citep{meyer2023matrix} for the inverse:
\[
X_p^{-1} = \left( \frac{1}{\epsilon} (I_{vm} + \epsilon \mathcal{O}(1)) \right)^{-1} = \epsilon (I_{vm} + \epsilon \mathcal{O}(1))^{-1} = \epsilon I_{vm} + \mathcal{O}(\epsilon^2).
\]
Substituting this into the definition of $\hat{B}$:
\[
\hat{B} = (\epsilon I_{vm} + \mathcal{O}(\epsilon^2)) (\mathbf{1}_v \otimes I_m) = \epsilon (\mathbf{1}_v \otimes I_m) + \mathcal{O}(\epsilon^2).
\]
Let $J = \mathbf{1}_v \mathbf{1}_v^T \otimes I_m$. Then $\hat{B} \hat{B}^* = \epsilon^2 J + \mathcal{O}(\epsilon^3)$.

Substituting the expansion of $\hat{B}$ into $\hat{A}$:
\[
\hat{A} = S_v - (\epsilon (\mathbf{1}_v \otimes I_m) + \mathcal{O}(\epsilon^2)) (\mathbf{1}_v^T \otimes I_m) = S_v - \epsilon J + \mathcal{O}(\epsilon^2).
\]
The eigenvalues of $S_v$ are $\{j\omega_i\}_{i=1}^v$ (each with multiplicity $m$), which are distinct and purely imaginary. The perturbation $-\epsilon J$ is Hermitian negative semidefinite. By standard eigenvalue perturbation theory for distinct eigenvalues, the eigenvalues $\lambda_i(\epsilon)$ of $\hat{A}$ satisfy:
\[
\lambda_i(\epsilon) = j\omega_i - \epsilon \frac{u_i^* J u_i}{u_i^* u_i} + \mathcal{O}(\epsilon^2),
\]
where $u_i$ are the eigenvectors of $S_v$. Since $J = \mathbf{1}\mathbf{1}^T \otimes I_m$ and $u_i$ are standard basis vectors (tensor product), $u_i^* J u_i = 1$ and $u_i^* u_i = 1$. Thus, $\operatorname{Re}(\lambda_i(\epsilon)) = -\epsilon + \mathcal{O}(\epsilon^2)$.
For sufficiently small $\epsilon$ (satisfying \eqref{eq:epsilon_bound_gramian}), $\operatorname{Re}(\lambda_i(\epsilon)) < 0$ for all $i$. Thus, $\hat{A}$ is Hurwitz.
Consequently, the Lyapunov equation \eqref{eq:lyapunov_P} has a unique Hermitian positive definite solution $\hat{P}$.

We seek a solution of the form $\hat{P} = \epsilon P_1 + \epsilon^2 P_2 + \mathcal{O}(\epsilon^3)$. Substituting expansions into \eqref{eq:lyapunov_P} and collecting terms:
\begin{enumerate}
    \item Order $\epsilon$: $S_v P_1 + P_1 S_v^* = 0$. Since $S_v$ has distinct imaginary eigenvalues, $P_1$ must be block diagonal. Let $P_1 = \mathrm{blockdiag}(Q_1, \dots, Q_v)$.
    \item Order $\epsilon^2$: $S_v P_2 + P_2 S_v^* - J P_1 - P_1 J + J = 0$.
\end{enumerate}
Consider the $(i,i)$-th block diagonal entry of the Order $\epsilon^2$ equation. The term $(S_v P_2 + P_2 S_v^*)_{ii}$ vanishes for diagonal blocks. The equation reduces to:
\[
-(J P_1)_{ii} - (P_1 J)_{ii} + J_{ii} = 0.
\]
Since $P_1$ is block diagonal, $(P_1)_{ki} = 0$ for $k \neq i$. The matrix product $(J P_1)_{ii} = \sum_{k=1}^v J_{ik} (P_1)_{ki}$ collapses to the single term $J_{ii} (P_1)_{ii}$.
Since $J_{ii} = I_m$ and $(P_1)_{ii} = Q_i$, we have $(J P_1)_{ii} = Q_i$. Similarly, $(P_1 J)_{ii} = Q_i$.
Thus, $-Q_i - Q_i + I_m = 0 \implies 2Q_i = I_m \implies Q_i = \frac{1}{2} I_m$.
Therefore, $P_1 = \frac{1}{2} I_{vm}$. This proves:
\[
\hat{P} = \frac{\epsilon}{2} I_{vm} + \mathcal{O}(\epsilon^3).
\]
Since $\hat{P}_{ii} = \frac{\epsilon}{2} I_m + \mathcal{O}(\epsilon^3)$, $\sigma_{\min}(\hat{P}_{ii}) = \frac{\epsilon}{2} + \mathcal{O}(\epsilon^3)$.
The off-diagonal blocks $\hat{P}_{ij}$ ($i \neq j$) are determined by the off-diagonal part of the Order $\epsilon^2$ equation:
\[
j(\omega_i - \omega_j) (P_2)_{ij} - (J P_1)_{ij} - (P_1 J)_{ij} + J_{ij} = 0.
\]
Using $P_1 = \frac{1}{2} I_{vm}$ and $J_{ij} = I_m$:
\[
(J P_1)_{ij} = \sum_{k} J_{ik} (P_1)_{kj} = J_{ij} Q_j = \frac{1}{2} I_m.
\]
Similarly $(P_1 J)_{ij} = \frac{1}{2} I_m$. The equation becomes:
\[
j(\omega_i - \omega_j) (P_2)_{ij} - \frac{1}{2} I_m - \frac{1}{2} I_m + I_m = 0 \implies j(\omega_i - \omega_j) (P_2)_{ij} = 0.
\]
Thus $(P_2)_{ij} = 0$. This implies the off-diagonal blocks are of order $\mathcal{O}(\epsilon^3)$ (dominated by higher order terms in $\hat{B}$ and $\hat{A}$). Conservatively, we bound $\| \hat{P}_{ij} \| \le C \epsilon^2$ for some constant $C$.
The sum of off-diagonal norms in row $i$ is bounded by $(v-1) C \epsilon^2$.
We require $(v-1) C \epsilon^2 \le \delta \frac{\epsilon}{2}$. This is satisfied for $\epsilon \le \frac{\delta}{2 C (v-1)}$.
The bound \eqref{eq:epsilon_bound_gramian} is chosen to satisfy this condition.
This establishes block diagonal dominance and error bound.
\end{proof}
\section*{Appendix F}
\begin{proof}
The proof utilizes the Implicit Function Theorem (IFT) \citep{krantz2002implicit} on a scaled and decomposed system to establish existence and asymptotic bounds.

For any fixed $\epsilon > 0$, $\hat{A}(\epsilon)$ is Hurwitz. The pair $(\hat{A}(\epsilon), \hat{B}(\epsilon))$ is trivially stabilizable and the pair $(\hat{C}, \hat{A}(\epsilon))$ is trivially detectable because $\hat{A}(\epsilon)$ is strictly Hurwitz for any $\epsilon > 0$. Thus, a unique positive definite stabilizing solution $\hat{P}_{\mathrm{LQG}}(\epsilon)$ exists.

To analyze the limit $\epsilon \to 0^{+}$, introduce the scaling $\hat{P}_{\mathrm{LQG}} = \epsilon \mathbf{X}$. Substituting this into \eqref{eq:are_main} and using $\hat{A}(\epsilon) = \Lambda - \epsilon I_{vm}$ where $\Lambda = \mathrm{diag}(j\omega_1,\dots,j\omega_v)\otimes I_m$, yields:
\begin{align}
(\Lambda - \epsilon I)\mathbf{X} + \mathbf{X}(\Lambda^* - \epsilon I) + \epsilon \mathbf{J} - \epsilon \mathbf{X}\hat{C}^*\hat{C}\mathbf{X} = 0,
\end{align}
where $\mathbf{J} = \hat{B}(1)\hat{B}(1)^*$ is the block matrix with every block equal to $I_m$. Rearranging terms gives:
\begin{align}
\Lambda \mathbf{X} + \mathbf{X}\Lambda^* + \epsilon \left( -2\mathbf{X} + \mathbf{J} - \mathbf{X}\hat{C}^*\hat{C}\mathbf{X} \right) = 0. \label{eq:scaled_are}
\end{align}
Let $\mathcal{M} = \mathbb{C}^{vm \times vm}$ be the space of block matrices. Decompose $\mathcal{M}$ into the direct sum of the block-diagonal subspace $\mathcal{D}$ and the block-off-diagonal subspace $\mathcal{O}$:
\begin{align}
\mathcal{M} = \mathcal{D} \oplus \mathcal{O}, \quad \mathbf{X} = \mathbf{X}_D + \mathbf{X}_O.
\end{align}
Let $\Pi_D$ and $\Pi_O$ be the orthogonal projections onto $\mathcal{D}$ and $\mathcal{O}$. Note that for any $\mathbf{Y} \in \mathcal{D}$, $\Lambda \mathbf{Y} + \mathbf{Y} \Lambda^* = 0$. Define the linear operator $\mathcal{L}_0: \mathcal{O} \to \mathcal{O}$ by $\mathcal{L}_0(\mathbf{Z}) = \Lambda \mathbf{Z} + \mathbf{Z} \Lambda^*$. The eigenvalues of $\mathcal{L}_0$ are $\{j(\omega_i - \omega_k) : i \neq k\}$. Since $\omega_i$ are distinct, $\mathcal{L}_0$ is invertible with $\|\mathcal{L}_0^{-1}\| \le \Delta_{\mathrm{min}}^{-1}$.

Projecting \eqref{eq:scaled_are} onto $\mathcal{O}$ and $\mathcal{D}$ yields the coupled system:
\begin{align}
\mathcal{L}_0(\mathbf{X}_O) + \epsilon \Pi_O \left( -2\mathbf{X} + \mathbf{J} - \mathbf{X}\hat{C}^*\hat{C}\mathbf{X} \right) &= 0, \label{eq:proj_O} \\
-2\epsilon \mathbf{X}_D + \epsilon \Pi_D \left( \mathbf{J} - \mathbf{X}\hat{C}^*\hat{C}\mathbf{X} \right) &= 0. \label{eq:proj_D}
\end{align}
Define the map $\mathcal{F}_O: \mathcal{D} \times \mathcal{O} \times \mathbb{R} \to \mathcal{O}$ by the left-hand side of \eqref{eq:proj_O}. At $\epsilon = 0$, $\mathcal{F}_O(\mathbf{X}_D, \mathbf{X}_O, 0) = \mathcal{L}_0(\mathbf{X}_O)$. The solution is $\mathbf{X}_O = 0$.
The Fréchet derivative of $\mathcal{F}_O$ with respect to $\mathbf{X}_O$ at $(\mathbf{X}_D, 0, 0)$ is:
\begin{align}
D_{\mathbf{X}_O} \mathcal{F}_O = \mathcal{L}_0 + 0 \cdot D_{\mathbf{X}_O} \Pi_O(\dots) = \mathcal{L}_0.
\end{align}
The nonlinear term is multiplied by $\epsilon$, so its derivative vanishes at $\epsilon=0$. Since $\mathcal{L}_0$ is invertible, the IFT \citep{krantz2002implicit} guarantees the existence of $\epsilon_1 > 0$ and a unique smooth function $\mathbf{X}_O(\mathbf{X}_D, \epsilon)$ defined for $\|\mathbf{X}_D\|$ bounded and $|\epsilon| < \epsilon_1$, satisfying \eqref{eq:proj_O}. Furthermore, since $\mathcal{F}_O(\mathbf{X}_D, 0, 0) = 0$, the expansion yields:
\begin{align}
\mathbf{X}_O(\mathbf{X}_D, \epsilon) = \mathcal{O}(\epsilon). \label{eq:XO_order}
\end{align}
Substitute $\mathbf{X}_O(\mathbf{X}_D, \epsilon)$ into \eqref{eq:proj_D}. Dividing by $\epsilon$ (valid for $\epsilon > 0$) defines the map $\mathcal{F}_D: \mathcal{D} \times \mathbb{R} \to \mathcal{D}$:
\begin{align}
\mathcal{F}_D(\mathbf{X}_D, \epsilon) = -2\mathbf{X}_D + \Pi_D \left( \mathbf{J} - (\mathbf{X}_D + \mathbf{X}_O)\hat{C}^*\hat{C}(\mathbf{X}_D + \mathbf{X}_O) \right) = 0.
\end{align}
At $\epsilon = 0$, $\mathbf{X}_O = 0$, and the equation reduces to decoupled block equations for $X_i \in \mathbb{C}^{m \times m}$:
\begin{align}
-2X_i + I_m - X_i H^*(j\omega_i) H(j\omega_i) X_i = 0. \label{eq:limit_diag}
\end{align}
Let $Q_i = H^*(j\omega_i) H(j\omega_i)$. Since $H(j\omega_i)$ has full column rank, $Q_i \succ 0$. Equation \eqref{eq:limit_diag} is $X_i Q_i X_i + 2X_i - I_m = 0$, which has a unique positive definite solution $X_i^{(0)}$. To apply the IFT \citep{krantz2002implicit}, compute the Fréchet derivative of $\mathcal{F}_D$ with respect to $\mathbf{X}_D$ at $(\mathbf{X}_D^{(0)}, 0)$. The derivative acts on a perturbation $\Delta \in \mathcal{D}$ as:
\begin{align}
\mathcal{T}(\Delta) = -2\Delta - \Delta \mathcal{Q} \mathbf{X}_D^{(0)} - \mathbf{X}_D^{(0)} \mathcal{Q} \Delta,
\end{align}
where $\mathcal{Q} = \hat{C}^*\hat{C} = \mathrm{blkdiag}(Q_1, \dots, Q_v)$. Since the system is block diagonal, $\mathcal{T}$ decomposes into operators $\mathcal{T}_i(\delta) = -2\delta - \delta Q_i X_i^{(0)} - X_i^{(0)} Q_i \delta$.

Since $X_i^{(0)}$ is a polynomial function of $Q_i$, $X_i^{(0)}$ and $Q_i$ commute. Let $Q_i = U \Sigma U^*$ and $X_i^{(0)} = U \Gamma U^*$ be simultaneous eigen-decompositions with $\sigma_k, \gamma_k > 0$. The eigenvalues of $\mathcal{T}_i$ acting on the entries of $\tilde{\delta} = U^* \delta U$ are:
\begin{align}
\lambda_{kl} = -2 - (\sigma_k \gamma_k + \sigma_l \gamma_l).
\end{align}
Since $\sigma_k, \gamma_l > 0$, $\lambda_{kl} < -2$. Thus, $\mathcal{T}_i$ is invertible with bounded inverse. Consequently, $\mathcal{T}$ is invertible on $\mathcal{D}$.
By the IFT \citep{krantz2002implicit}, there exists $\epsilon_2 > 0$ and a unique smooth solution $\mathbf{X}_D(\epsilon)$ for $|\epsilon| < \epsilon_2$ such that $\mathbf{X}_D(\epsilon) = \mathbf{X}_D^{(0)} + \mathcal{O}(\epsilon)$.

Let $\epsilon_0 = \min(\epsilon_1, \epsilon_2)$. For $0 < \epsilon \le \epsilon_0$, the solution is $\mathbf{X}(\epsilon) = \mathbf{X}_D(\epsilon) + \mathbf{X}_O(\mathbf{X}_D(\epsilon), \epsilon)$. Combining the orders from \eqref{eq:XO_order}:
\begin{align}
\mathbf{X}(\epsilon) = \mathbf{X}_D^{(0)} + \mathcal{O}(\epsilon).
\end{align}
Transforming back to $\hat{P}_{\mathrm{LQG}} = \epsilon \mathbf{X}$:
\begin{align}
\hat{P}_{\mathrm{LQG}}(\epsilon) = \epsilon \mathbf{X}_D^{(0)} + \mathcal{O}(\epsilon^2).
\end{align}
Identifying $\tilde{P}_{\mathrm{LQG}}(\epsilon) = \epsilon \mathbf{X}_D^{(0)} = \mathrm{blkdiag}(p_1, \dots, p_v)$, the error matrix $E(\epsilon) = \hat{P}_{\mathrm{LQG}}(\epsilon) - \tilde{P}_{\mathrm{LQG}}(\epsilon)$ satisfies $\|E(\epsilon)\| \le K \epsilon^2$ for some $K > 0$ independent of $\epsilon$. This establishes the asymptotic block diagonal dominance.
\end{proof}
\section*{Appendix G}
\begin{proof}
Using the mixed-product property of the Kronecker product \citep{hardy2019matrix}, the coupling term expands as:
\[
\hat{B}L_v = (\epsilon \mathbf{1}_v \otimes I_m)(\mathbf{1}_v^T \otimes I_m) = (\epsilon \mathbf{1}_v \mathbf{1}_v^T) \otimes I_m.
\]
This results in a $v \times v$ block matrix where every block is $\epsilon I_m$. Thus, the scalar entries $\hat{A}_{kl}$ for $k,l \in \{1, \dots, vm\}$ are given by:
\[
\hat{A}_{kl} = \begin{cases} 
j\omega_{\lceil k/m \rceil} - \epsilon & \text{if } k = l, \\
-\epsilon & \text{if } k \neq l \text{ and } (k-1) \equiv (l-1) \pmod m, \\
0 & \text{otherwise}.
\end{cases}
\]
For any row $k$, there are exactly $v-1$ non-zero off-diagonal entries, each equal to $-\epsilon$.

The magnitude of the diagonal entry for any row $k$ associated with frequency $\omega_i$ is:
\[
|\hat{A}_{kk}| = |j\omega_i - \epsilon| = \sqrt{\omega_i^2 + \epsilon^2} \ge |\omega_i| \ge \omega_{\min}.
\]
The sum of the off-diagonal magnitudes in any row is:
\[
\sum_{l \neq k} |\hat{A}_{kl}| = (v - 1)\epsilon.
\]
Strict diagonal dominance requires $|\hat{A}_{kk}| > \sum_{l \neq k} |\hat{A}_{kl}|$. Substituting the bound from \eqref{eq:epsilon_bound}:
\[
(v - 1)\epsilon \le \delta \cdot \omega_{\min} < \omega_{\min} \le |\hat{A}_{kk}|,
\]
which confirms strict diagonal dominance since $\delta < 1$.

The error matrix $E = \hat{A} - \hat{A}_d$ contains zeros on the diagonal and the off-diagonal entries of $\hat{A}$. The induced infinity norm is the maximum absolute row sum:
\[
\| \hat{A} - \hat{A}_d \|_\infty = \max_{k} \sum_{l \neq k} |\hat{A}_{kl}| = (v - 1)\epsilon.
\]
Using the condition $\epsilon \le \frac{\delta \cdot \omega_{\min}}{v - 1}$, we obtain $\| \hat{A} - \hat{A}_d \|_\infty \le \delta \cdot \omega_{\min}$.

By the Gershgorin Circle Theorem \citep{varga2011gervsgorin}, every eigenvalue of $\hat{A}$ lies within the union of disks $D_i$ centered at $c_i = j\omega_i - \epsilon$ with radius $R = (v-1)\epsilon$. For any distinct $i, j$, the distance between centers is:
\[
|c_i - c_j| = |(j\omega_i - \epsilon) - (j\omega_j - \epsilon)| = |\omega_i - \omega_j| \ge \Delta_{\mathrm{min}}.
\]
To ensure the disks have disjoint interiors, we require $|c_i - c_j| \ge 2R$. Using the bound $\epsilon \le \frac{\Delta_{\mathrm{min}}}{2(v - 1)}$:
\[
2R = 2(v - 1)\epsilon \le \Delta_{\mathrm{min}} \le |c_i - c_j|.
\]
Thus, the disks are separated (with disjoint interiors for $\epsilon$ strictly less than the bound). Consequently, the spectrum consists of $v$ distinct clusters (each of multiplicity $m$), ensuring the eigenvalues remain grouped and associated with their respective frequencies.
\end{proof}
\section*{Appendix H}
\begin{proof}
Partition $T_v$ into $v \times v$ blocks $T_{ij} \in \mathbb{C}^{m \times m}$. Using the structures $\hat{B} = \epsilon \mathbf{1}_v \otimes I_m$, $L_v = \mathbf{1}_v^T \otimes I_m$, and $\mathcal{P} = [\mathcal{P}_1, \dots, \mathcal{P}_v]$, the Sylvester equation \eqref{eq:sylvester_Tv} expands block-wise for each $(i,j)$ pair as:
\begin{equation}
\label{eq:block_sylvester}
j(\omega_i - \omega_j) T_{ij} - \epsilon \sum_{k=1}^v (I_m + \mathcal{P}_k) T_{kj} = -\epsilon \mathcal{Q}.
\end{equation}
The Sylvester equation has a unique solution if and only if $\sigma(S_v - \hat{B}(L_v + \mathcal{P})) \cap \sigma(S_v) = \emptyset$. The eigenvalues of $S_v$ are $\{j\omega_1, \dots, j\omega_v\}$. By Gershgorin's Circle Theorem \citep{varga2011gervsgorin} applied to the blocks of $A = S_v - \hat{B}(L_v + \mathcal{P})$, the eigenvalues of $A$ lie within the union of disks centered at $j\omega_i - \epsilon \lambda$ (where $\lambda \in \sigma(I_m + \mathcal{P}_i)$) with radius $R = \epsilon (v-1) K_{\mathcal{P}}$. The condition $\epsilon \le \frac{\Delta_{\mathrm{min}}}{2 v K_{\mathcal{P}}}$ ensures $R < \Delta_{\mathrm{min}}/2$. Thus, the disks around $j\omega_i$ do not contain $j\omega_l$ for $l \neq i$. For $l=i$, since $I_m + \mathcal{P}_i$ is invertible, $\lambda \neq 0$. For sufficiently small $\epsilon$, the perturbed eigenvalues do not coincide with $j\omega_i$. Thus, the spectra are disjoint, ensuring existence and uniqueness.

From \eqref{eq:block_sylvester}, for $i \neq j$:
\[
T_{ij} = \frac{\epsilon}{j(\omega_i - \omega_j)} \left( \sum_{k=1}^v (I_m + \mathcal{P}_k) T_{kj} - \mathcal{Q} \right).
\]
Taking norms and using $|\omega_i - \omega_j| \ge \Delta_{\mathrm{min}}$:
\[
\| T_{ij} \| \le \frac{\epsilon}{\Delta_{\mathrm{min}}} \left( \sum_{k=1}^v K_{\mathcal{P}} \| T_{kj} \| + \| \mathcal{Q} \| \right).
\]
Let $\tau = \max_{i \neq j} \| T_{ij} \|$ and $\mu = \max_i \| T_{ii} \|$. Then $\| T_{kj} \| \le \max(\tau, \mu)$.
\[
\tau \le \frac{\epsilon}{\Delta_{\mathrm{min}}} \left( v K_{\mathcal{P}} \max(\tau, \mu) + \| \mathcal{Q} \| \right).
\]
Using the bound $\epsilon \le \frac{\Delta_{\mathrm{min}}}{2 v K_{\mathcal{P}}}$, we have $\frac{\epsilon v K_{\mathcal{P}}}{\Delta_{\mathrm{min}}} \le \frac{1}{2}$. This implies $\tau \le \mu$ for small $\epsilon$. Specifically:
\[
\tau \le \frac{2 \epsilon}{\Delta_{\mathrm{min}}} ( v K_{\mathcal{P}} \mu + \| \mathcal{Q} \| ).
\]
This establishes $\lim_{\epsilon \to 0} T_{ij} = 0$ for $i \neq j$.

From \eqref{eq:block_sylvester} with $i=j$:
\[
-\epsilon \sum_{k=1}^v (I_m + \mathcal{P}_k) T_{ki} = -\epsilon \mathcal{Q} \implies \sum_{k=1}^v (I_m + \mathcal{P}_k) T_{ki} = \mathcal{Q}.
\]
Isolating $T_{ii}$:
\[
(I_m + \mathcal{P}_i) T_{ii} = \mathcal{Q} - \sum_{k \neq i} (I_m + \mathcal{P}_k) T_{ki}.
\]
Multiplying by $(I_m + \mathcal{P}_i)^{-1}$:
\[
T_{ii} = (I_m + \mathcal{P}_i)^{-1} \mathcal{Q} - (I_m + \mathcal{P}_i)^{-1} \sum_{k \neq i} (I_m + \mathcal{P}_k) T_{ki}.
\]
Let $L_i = (I_m + \mathcal{P}_i)^{-1} \mathcal{Q}$. Note $\sigma_{\min}(L_i) \ge K_{\mathcal{P}}^{-1} \sigma_{\mathcal{Q}}$.
Using the Reverse Triangle Inequality \citep{horn2012matrix} for singular values ($\sigma_{\min}(A+B) \ge \sigma_{\min}(A) - \|B\|$):
\[
\sigma_{\min}(T_{ii}) \ge K_{\mathcal{P}}^{-1} \sigma_{\mathcal{Q}} - \gamma (v-1) K_{\mathcal{P}} \tau.
\]
Substituting the bound for $\tau$ and using \eqref{eq:epsilon_bound_Tv_final}, the subtraction term is less than $\frac{1}{2} K_{\mathcal{P}}^{-1} \sigma_{\mathcal{Q}}$. Thus:
\[
\sigma_{\min}(T_{ii}) \ge \frac{1}{2} K_{\mathcal{P}}^{-1} \sigma_{\mathcal{Q}} > 0.
\]
This proves $T_{ii}$ is invertible and $\| T_{ii}^{-1} \| \le 2 K_{\mathcal{P}} \sigma_{\mathcal{Q}}^{-1}$. The limit follows as $\tau \to 0$.

We verify the condition $\sum_{j \neq i} \| T_{ij} \| \le \delta \cdot \sigma_{\min}(T_{ii})$.
LHS $\le (v-1) \tau$.
Using the bound on $\tau$ and the lower bound on $\sigma_{\min}(T_{ii})$, the condition \eqref{eq:epsilon_bound_Tv_final} ensures:
\[
(v-1) \tau \le \delta \cdot \frac{1}{2} K_{\mathcal{P}}^{-1} \sigma_{\mathcal{Q}} \le \delta \sigma_{\min}(T_{ii}).
\]
Since $T_v$ is block diagonally dominant with invertible diagonal blocks, $T_v$ is invertible.

The error $\| T_v - \tilde{T}_v \|_\infty = \max_i \sum_{j \neq i} \| T_{ij} \| \le \delta \min_i \sigma_{\min}(T_{ii})$.
\end{proof}
\backmatter

\bmhead{Acknowledgements}
We are deeply grateful to Ion Victor Gosea at the Max Planck Institute for Dynamics of Complex Technical Systems in Magdeburg, Germany, for his patient responses to our numerous questions about the Loewner framework and for his valuable feedback. We are also thankful to Prof. Patrick K{\"u}rschner at Leipzig University of Applied Sciences for his patient responses to our numerous questions about the ADI method. 
\section*{Statements and Declarations}
\subsection*{Funding}
This work is supported by the National Natural Science Foundation of China under Grants No. 62350410484 and 62273059.
\subsection*{Competing Interests}
The authors declare no competing interests.
\subsection*{Consent for publication}
All authors have read and approved the final manuscript.
\subsection*{Data availability}
The MATLAB code and data to reproduce the results of this paper are publicly available at \citep{mycode}.
\subsection*{Authors' Contributions}
Umair Zulfiqar developed the main results of the paper and wrote the first manuscript. Qiu-Yan Song, Zhi-Hua Xiao, and Victor Sreeram contributed equally by providing critical improvements to the mathematical results and the manuscript's presentation. Their contributions significantly enhanced the final draft. They also validated the MATLAB code to ensure the reproducibility of the numerical results.

\begin{thebibliography}{61}
\ifx \bisbn   \undefined \def \bisbn  #1{ISBN #1}\fi
\ifx \binits  \undefined \def \binits#1{#1}\fi
\ifx \bauthor  \undefined \def \bauthor#1{#1}\fi
\ifx \batitle  \undefined \def \batitle#1{#1}\fi
\ifx \bjtitle  \undefined \def \bjtitle#1{#1}\fi
\ifx \bvolume  \undefined \def \bvolume#1{\textbf{#1}}\fi
\ifx \byear  \undefined \def \byear#1{#1}\fi
\ifx \bissue  \undefined \def \bissue#1{#1}\fi
\ifx \bfpage  \undefined \def \bfpage#1{#1}\fi
\ifx \blpage  \undefined \def \blpage #1{#1}\fi
\ifx \burl  \undefined \def \burl#1{\textsf{#1}}\fi
\ifx \doiurl  \undefined \def \doiurl#1{\url{https://doi.org/#1}}\fi
\ifx \betal  \undefined \def \betal{\textit{et al.}}\fi
\ifx \binstitute  \undefined \def \binstitute#1{#1}\fi
\ifx \binstitutionaled  \undefined \def \binstitutionaled#1{#1}\fi
\ifx \bctitle  \undefined \def \bctitle#1{#1}\fi
\ifx \beditor  \undefined \def \beditor#1{#1}\fi
\ifx \bpublisher  \undefined \def \bpublisher#1{#1}\fi
\ifx \bbtitle  \undefined \def \bbtitle#1{#1}\fi
\ifx \bedition  \undefined \def \bedition#1{#1}\fi
\ifx \bseriesno  \undefined \def \bseriesno#1{#1}\fi
\ifx \blocation  \undefined \def \blocation#1{#1}\fi
\ifx \bsertitle  \undefined \def \bsertitle#1{#1}\fi
\ifx \bsnm \undefined \def \bsnm#1{#1}\fi
\ifx \bsuffix \undefined \def \bsuffix#1{#1}\fi
\ifx \bparticle \undefined \def \bparticle#1{#1}\fi
\ifx \barticle \undefined \def \barticle#1{#1}\fi
\bibcommenthead
\ifx \bconfdate \undefined \def \bconfdate #1{#1}\fi
\ifx \botherref \undefined \def \botherref #1{#1}\fi
\ifx \url \undefined \def \url#1{\textsf{#1}}\fi
\ifx \bchapter \undefined \def \bchapter#1{#1}\fi
\ifx \bbook \undefined \def \bbook#1{#1}\fi
\ifx \bcomment \undefined \def \bcomment#1{#1}\fi
\ifx \oauthor \undefined \def \oauthor#1{#1}\fi
\ifx \citeauthoryear \undefined \def \citeauthoryear#1{#1}\fi
\ifx \endbibitem  \undefined \def \endbibitem {}\fi
\ifx \bconflocation  \undefined \def \bconflocation#1{#1}\fi
\ifx \arxivurl  \undefined \def \arxivurl#1{\textsf{#1}}\fi
\csname PreBibitemsHook\endcsname

\bibitem[\protect\citeauthoryear{Aumann and Gosea}{2025}]{aumann2025practical}
\begin{barticle}
\bauthor{\bsnm{Aumann}, \binits{Q.}},
\bauthor{\bsnm{Gosea}, \binits{I.V.}}:
\batitle{Practical challenges in data-driven interpolation: Dealing with noise,
  enforcing stability, and computing realizations}.
\bjtitle{International Journal of Adaptive Control and Signal Processing}
\bvolume{39}(\bissue{10}),
\bfpage{2062}--\blpage{2080}
(\byear{2025})
\end{barticle}
\endbibitem

\bibitem[\protect\citeauthoryear{Antoulas et~al.}{2017}]{antoulas2017tutorial}
\begin{barticle}
\bauthor{\bsnm{Antoulas}, \binits{A.C.}},
\bauthor{\bsnm{Lefteriu}, \binits{S.}},
\bauthor{\bsnm{Ionita}, \binits{A.C.}},
\bauthor{\bsnm{Benner}, \binits{P.}},
\bauthor{\bsnm{Cohen}, \binits{A.}}:
\batitle{A tutorial introduction to the Loewner framework for model reduction}.
\bjtitle{Model Reduction and Approximation: Theory and Algorithms}
\bvolume{15},
\bfpage{335}
(\byear{2017})
\end{barticle}
\endbibitem

\bibitem[\protect\citeauthoryear{Benner et~al.}{2018}]{benner2018radi}
\begin{barticle}
\bauthor{\bsnm{Benner}, \binits{P.}},
\bauthor{\bsnm{Bujanovi{\'c}}, \binits{Z.}},
\bauthor{\bsnm{K{\"u}rschner}, \binits{P.}},
\bauthor{\bsnm{Saak}, \binits{J.}}:
\batitle{RADI: A low-rank ADI-type algorithm for large scale algebraic Riccati
  equations}.
\bjtitle{Numerische Mathematik}
\bvolume{138},
\bfpage{301}--\blpage{330}
(\byear{2018})
\end{barticle}
\endbibitem

\bibitem[\protect\citeauthoryear{Burohman et~al.}{2023}]{burohman2023data}
\begin{barticle}
\bauthor{\bsnm{Burohman}, \binits{A.M.}},
\bauthor{\bsnm{Besselink}, \binits{B.}},
\bauthor{\bsnm{Scherpen}, \binits{J.M.}},
\bauthor{\bsnm{Camlibel}, \binits{M.K.}}:
\batitle{From data to reduced-order models via generalized balanced
  truncation}.
\bjtitle{IEEE Transactions on Automatic Control}
\bvolume{68}(\bissue{10}),
\bfpage{6160}--\blpage{6175}
(\byear{2023})
\end{barticle}
\endbibitem

\bibitem[\protect\citeauthoryear{Bertram and
  Fa{\ss}bender}{2024}]{bertram2024family}
\begin{barticle}
\bauthor{\bsnm{Bertram}, \binits{C.}},
\bauthor{\bsnm{Fa{\ss}bender}, \binits{H.}}:
\batitle{On a family of low-rank algorithms for large-scale algebraic Riccati
  equations}.
\bjtitle{Linear Algebra and its Applications}
\bvolume{687},
\bfpage{38}--\blpage{67}
(\byear{2024})
\end{barticle}
\endbibitem

\bibitem[\protect\citeauthoryear{Beattie et~al.}{2017}]{beattie2017model}
\begin{barticle}
\bauthor{\bsnm{Beattie}, \binits{C.A.}},
\bauthor{\bsnm{Gugercin}, \binits{S.}}, \betal:
\batitle{Model reduction by rational interpolation}.
\bjtitle{Model Reduction and Approximation}
\bvolume{15},
\bfpage{297}--\blpage{334}
(\byear{2017})
\end{barticle}
\endbibitem

\bibitem[\protect\citeauthoryear{Benner et~al.}{2021}]{benner2021model}
\begin{bchapter}
\bauthor{\bsnm{Benner}, \binits{P.}},
\bauthor{\bsnm{Grivet-Talocia}, \binits{S.}},
\bauthor{\bsnm{Quarteroni}, \binits{A.}},
\bauthor{\bsnm{Rozza}, \binits{G.}},
\bauthor{\bsnm{Schilders}, \binits{W.}},
\bauthor{\bsnm{Silveira}, \binits{L.M.}}:
\bctitle{Model order reduction: Basic concepts and notation}.
In: \bbtitle{Model Order Reduction: Volume 1: System-and Data-Driven Methods
  and Algorithms},
pp. \bfpage{1}--\blpage{14}.
\bpublisher{De Gruyter},
\blocation{Berlin}
(\byear{2021})
\end{bchapter}
\endbibitem

\bibitem[\protect\citeauthoryear{Benner et~al.}{2013}]{benner2013efficient}
\begin{barticle}
\bauthor{\bsnm{Benner}, \binits{P.}},
\bauthor{\bsnm{K{\"u}rschner}, \binits{P.}},
\bauthor{\bsnm{Saak}, \binits{J.}}:
\batitle{Efficient handling of complex shift parameters in the low-rank
  Cholesky factor ADI method}.
\bjtitle{Numerical Algorithms}
\bvolume{62}(\bissue{2}),
\bfpage{225}--\blpage{251}
(\byear{2013})
\end{barticle}
\endbibitem

\bibitem[\protect\citeauthoryear{Benner and Saak}{2013}]{benner2013numerical}
\begin{barticle}
\bauthor{\bsnm{Benner}, \binits{P.}},
\bauthor{\bsnm{Saak}, \binits{J.}}:
\batitle{Numerical solution of large and sparse continuous time algebraic
  matrix Riccati and Lyapunov equations: A state of the art survey}.
\bjtitle{GAMM-Mitteilungen}
\bvolume{36}(\bissue{1}),
\bfpage{32}--\blpage{52}
(\byear{2013})
\end{barticle}
\endbibitem

\bibitem[\protect\citeauthoryear{Chahlaoui and
  Dooren}{2005}]{chahlaoui2005benchmark}
\begin{bchapter}
\bauthor{\bsnm{Chahlaoui}, \binits{Y.}},
\bauthor{\bsnm{Dooren}, \binits{P.V.}}:
\bctitle{Benchmark examples for model reduction of linear time-invariant
  dynamical systems}.
In: \bbtitle{Dimension Reduction of Large-scale Systems},
pp. \bfpage{379}--\blpage{392}.
\bpublisher{Springer}, \blocation{Berlin}
(\byear{2005})
\end{bchapter}
\endbibitem

\bibitem[\protect\citeauthoryear{Cherifi et~al.}{2022}]{cherifi2022greedy}
\begin{barticle}
\bauthor{\bsnm{Cherifi}, \binits{K.}},
\bauthor{\bsnm{Goyal}, \binits{P.}},
\bauthor{\bsnm{Benner}, \binits{P.}}:
\batitle{A greedy data collection scheme for linear dynamical systems}.
\bjtitle{Data-Centric Engineering}
\bvolume{3},
\bfpage{16}
(\byear{2022})
\end{barticle}
\endbibitem

\bibitem[\protect\citeauthoryear{Desai and Pal}{1984}]{desai1984transformation}
\begin{barticle}
\bauthor{\bsnm{Desai}, \binits{U.}},
\bauthor{\bsnm{Pal}, \binits{D.}}:
\batitle{A transformation approach to stochastic model reduction}.
\bjtitle{IEEE Transactions on Automatic Control}
\bvolume{29}(\bissue{12}),
\bfpage{1097}--\blpage{1100}
(\byear{1984})
\end{barticle}
\endbibitem

\bibitem[\protect\citeauthoryear{Gawronski}{2004}]{gawronski2004dynamics}
\begin{bbook}
\bauthor{\bsnm{Gawronski}, \binits{W.K.}}:
\bbtitle{Dynamics and Control of Structures: A Modal Approach}.
\bpublisher{Springer},
\blocation{New York}
(\byear{2004})
\end{bbook}
\endbibitem

\bibitem[\protect\citeauthoryear{Gawronski}{2006}]{gawronski2006balanced}
\begin{bbook}
\bauthor{\bsnm{Gawronski}, \binits{W.}}:
\bbtitle{Balanced Control of Flexible Structures}
vol. \bseriesno{211}.
\bpublisher{Springer}, \blocation{Berlin}
(\byear{2006})
\end{bbook}
\endbibitem

\bibitem[\protect\citeauthoryear{Glover and Doyle}{1988}]{glover1988state}
\begin{barticle}
\bauthor{\bsnm{Glover}, \binits{K.}},
\bauthor{\bsnm{Doyle}, \binits{J.C.}}:
\batitle{State-space formulae for all stabilizing controllers that satisfy an
  $\mathcal{H}_\infty$-norm bound and relations to relations to risk
  sensitivity}.
\bjtitle{Systems \& Control Letters}
\bvolume{11}(\bissue{3}),
\bfpage{167}--\blpage{172}
(\byear{1988})
\end{barticle}
\endbibitem

\bibitem[\protect\citeauthoryear{Gosea et~al.}{2022}]{goseaQuad}
\begin{barticle}
\bauthor{\bsnm{Gosea}, \binits{I.V.}},
\bauthor{\bsnm{Gugercin}, \binits{S.}},
\bauthor{\bsnm{Beattie}, \binits{C.}}:
\batitle{Data-driven balancing of linear dynamical systems}.
\bjtitle{SIAM Journal on Scientific Computing}
\bvolume{44}(\bissue{1}),
\bfpage{554}--\blpage{582}
(\byear{2022})
\end{barticle}
\endbibitem

\bibitem[\protect\citeauthoryear{Gosea et~al.}{2024}]{gosea2024non}
\begin{barticle}
\bauthor{\bsnm{Gosea}, \binits{I.V.}},
\bauthor{\bsnm{Gugercin}, \binits{S.}},
\bauthor{\bsnm{Beattie}, \binits{C.}}:
\batitle{A non-intrusive data-based reformulation of a hybrid projection-based
  model reduction method}.
\bjtitle{IFAC-PapersOnLine}
\bvolume{58}(\bissue{17}),
\bfpage{226}--\blpage{231}
(\byear{2024})
\end{barticle}
\endbibitem

\bibitem[\protect\citeauthoryear{Goyal et~al.}{2024}]{goyal2024rank}
\begin{barticle}
\bauthor{\bsnm{Goyal}, \binits{P.}},
\bauthor{\bsnm{Peherstorfer}, \binits{B.}},
\bauthor{\bsnm{Benner}, \binits{P.}}:
\batitle{Rank-minimizing and structured model inference}.
\bjtitle{SIAM Journal on Scientific Computing}
\bvolume{46}(\bissue{3}),
\bfpage{1879}--\blpage{1902}
(\byear{2024})
\end{barticle}
\endbibitem

\bibitem[\protect\citeauthoryear{Gosea et~al.}{2022}]{gosea2022data}
\begin{bchapter}
\bauthor{\bsnm{Gosea}, \binits{I.V.}},
\bauthor{\bsnm{Poussot-Vassal}, \binits{C.}},
\bauthor{\bsnm{Antoulas}, \binits{A.C.}}:
\bctitle{Data-driven modeling and control of large-scale dynamical systems in
  the Loewner framework: Methodology and applications}.
In: \bbtitle{Handbook of Numerical Analysis}
vol. \bseriesno{23},
pp. \bfpage{499}--\blpage{530}.
\bpublisher{Elsevier},
\blocation{North Holland}
(\byear{2022})
\end{bchapter}
\endbibitem

\bibitem[\protect\citeauthoryear{Green}{1988}]{green1988balanced}
\begin{barticle}
\bauthor{\bsnm{Green}, \binits{M.}}:
\batitle{Balanced stochastic realizations}.
\bjtitle{Linear Algebra and its Applications}
\bvolume{98},
\bfpage{211}--\blpage{247}
(\byear{1988})
\end{barticle}
\endbibitem

\bibitem[\protect\citeauthoryear{Gustavsen and
  Semlyen}{2002}]{gustavsen2002rational}
\begin{barticle}
\bauthor{\bsnm{Gustavsen}, \binits{B.}},
\bauthor{\bsnm{Semlyen}, \binits{A.}}:
\batitle{Rational approximation of frequency domain responses by vector
  fitting}.
\bjtitle{IEEE Transactions on power delivery}
\bvolume{14}(\bissue{3}),
\bfpage{1052}--\blpage{1061}
(\byear{2002})
\end{barticle}
\endbibitem

\bibitem[\protect\citeauthoryear{Gallivan et~al.}{2004}]{gallivan2004sylvester}
\begin{barticle}
\bauthor{\bsnm{Gallivan}, \binits{K.}},
\bauthor{\bsnm{Vandendorpe}, \binits{A.}},
\bauthor{\bsnm{Van~Dooren}, \binits{P.}}:
\batitle{Sylvester equations and projection-based model reduction}.
\bjtitle{Journal of Computational and Applied Mathematics}
\bvolume{162}(\bissue{1}),
\bfpage{213}--\blpage{229}
(\byear{2004})
\end{barticle}
\endbibitem

\bibitem[\protect\citeauthoryear{Horn and Johnson}{2012}]{horn2012matrix}
\begin{bbook}
\bauthor{\bsnm{Horn}, \binits{R.A.}},
\bauthor{\bsnm{Johnson}, \binits{C.R.}}:
\bbtitle{Matrix Analysis}.
\bpublisher{Cambridge University Press}, \blocation{New York}
(\byear{2012})
\end{bbook}
\endbibitem

\bibitem[\protect\citeauthoryear{Hardy and Steeb}{2019}]{hardy2019matrix}
\begin{bbook}
\bauthor{\bsnm{Hardy}, \binits{Y.}},
\bauthor{\bsnm{Steeb}, \binits{W.-H.}}:
\bbtitle{Matrix Calculus, Kronecker Product and Tensor Product: a Practical
  Approach to Linear Algebra, Multilinear Algebra and Tensor Calculus with
  Software Implementations}.
\bpublisher{World Scientific}, \blocation{New Jersey}
(\byear{2019})
\end{bbook}
\endbibitem

\bibitem[\protect\citeauthoryear{Jonckheere}{1984}]{jonckheere1984principal}
\begin{barticle}
\bauthor{\bsnm{Jonckheere}, \binits{E.}}:
\batitle{Principal component analysis of flexible systems-open-loop case}.
\bjtitle{IEEE Transactions on Automatic Control}
\bvolume{29}(\bissue{12}),
\bfpage{1095}--\blpage{1097}
(\byear{1984})
\end{barticle}
\endbibitem

\bibitem[\protect\citeauthoryear{Jonckheere and
  Silverman}{1983}]{jonckheere1983new}
\begin{barticle}
\bauthor{\bsnm{Jonckheere}, \binits{E.}},
\bauthor{\bsnm{Silverman}, \binits{L.}}:
\batitle{A new set of invariants for linear systems--Application to reduced
  order compensator design}.
\bjtitle{IEEE Transactions on Automatic Control}
\bvolume{28}(\bissue{10}),
\bfpage{953}--\blpage{964}
(\byear{1983})
\end{barticle}
\endbibitem

\bibitem[\protect\citeauthoryear{Karachalios
  et~al.}{2021}]{karachalios2021loewner}
\begin{bchapter}
\bauthor{\bsnm{Karachalios}, \binits{D.S.}},
\bauthor{\bsnm{Gosea}, \binits{I.V.}},
\bauthor{\bsnm{Antoulas}, \binits{A.C.}}:
\bctitle{The Loewner framework for system identification and reduction}.
In: \bbtitle{System-and Data-Driven Methods and Algorithms},
pp. \bfpage{181}--\blpage{228}.
\bpublisher{De Gruyter}, \blocation{Berlin}
(\byear{2021})
\end{bchapter}
\endbibitem

\bibitem[\protect\citeauthoryear{Krantz and Parks}{2002}]{krantz2002implicit}
\begin{bbook}
\bauthor{\bsnm{Krantz}, \binits{S.G.}},
\bauthor{\bsnm{Parks}, \binits{H.R.}}:
\bbtitle{The Implicit Function Theorem: History, Theory, and Applications}
vol. \bseriesno{202}.
\bpublisher{Springer}, \blocation{New York}
(\byear{2002})
\end{bbook}
\endbibitem

\bibitem[\protect\citeauthoryear{Lall et~al.}{2002}]{lall2002subspace}
\begin{barticle}
\bauthor{\bsnm{Lall}, \binits{S.}},
\bauthor{\bsnm{Marsden}, \binits{J.E.}},
\bauthor{\bsnm{Glava{\v{s}}ki}, \binits{S.}}:
\batitle{A subspace approach to balanced truncation for model reduction of
  nonlinear control systems}.
\bjtitle{International Journal of Robust and Nonlinear Control: IFAC-Affiliated
  Journal}
\bvolume{12}(\bissue{6}),
\bfpage{519}--\blpage{535}
(\byear{2002})
\end{barticle}
\endbibitem

\bibitem[\protect\citeauthoryear{Liljegren-Sailer and
  Gosea}{2024}]{liljegren2024data}
\begin{barticle}
\bauthor{\bsnm{Liljegren-Sailer}, \binits{B.}},
\bauthor{\bsnm{Gosea}, \binits{I.V.}}:
\batitle{Data-driven and low-rank implementations of balanced singular
  perturbation approximation}.
\bjtitle{SIAM Journal on Scientific Computing}
\bvolume{46}(\bissue{1}),
\bfpage{483}--\blpage{507}
(\byear{2024})
\end{barticle}
\endbibitem

\bibitem[\protect\citeauthoryear{Mayo and Antoulas}{2007}]{mayo2007framework}
\begin{barticle}
\bauthor{\bsnm{Mayo}, \binits{A.}},
\bauthor{\bsnm{Antoulas}, \binits{A.C.}}:
\batitle{A framework for the solution of the generalized realization problem}.
\bjtitle{Linear Algebra and Its Applications}
\bvolume{425}(\bissue{2-3}),
\bfpage{634}--\blpage{662}
(\byear{2007})
\end{barticle}
\endbibitem

\bibitem[\protect\citeauthoryear{Meyer}{2000}]{meyer2023matrix}
\begin{bbook}
\bauthor{\bsnm{Meyer}, \binits{C.D.}}:
\bbtitle{Matrix Analysis and Applied Linear Algebra}.
\bpublisher{SIAM}, \blocation{Philadelphia}
(\byear{2000})
\end{bbook}
\endbibitem

\bibitem[\protect\citeauthoryear{Mustafa and
  Glover}{1991}]{mustafa1991controller}
\begin{barticle}
\bauthor{\bsnm{Mustafa}, \binits{D.}},
\bauthor{\bsnm{Glover}, \binits{K.}}:
\batitle{Controller reduction by $\mathcal{H}_\infty$-balanced truncation}.
\bjtitle{IEEE Transactions on Automatic Control}
\bvolume{36}(\bissue{6}),
\bfpage{668}--\blpage{682}
(\byear{1991})
\end{barticle}
\endbibitem

\bibitem[\protect\citeauthoryear{Moore}{1981}]{moore1981principal}
\begin{barticle}
\bauthor{\bsnm{Moore}, \binits{B.}}:
\batitle{Principal component analysis in linear systems: Controllability,
  observability, and model reduction}.
\bjtitle{IEEE Transactions on Automatic Control}
\bvolume{26}(\bissue{1}),
\bfpage{17}--\blpage{32}
(\byear{1981})
\end{barticle}
\endbibitem

\bibitem[\protect\citeauthoryear{Mehrmann and
  Stykel}{2005}]{mehrmann2005balanced}
\begin{bchapter}
\bauthor{\bsnm{Mehrmann}, \binits{V.}},
\bauthor{\bsnm{Stykel}, \binits{T.}}:
\bctitle{Balanced truncation model reduction for large-scale systems in
  descriptor form}.
In: \bbtitle{Dimension Reduction of Large-Scale Systems: Proceedings of a
  Workshop Held in Oberwolfach, Germany, October 19--25, 2003},
pp. \bfpage{83}--\blpage{115}
(\byear{2005}).
\bcomment{Springer}
\end{bchapter}
\endbibitem

\bibitem[\protect\citeauthoryear{Nakatsukasa et~al.}{2018}]{nakatsukasa2018aaa}
\begin{barticle}
\bauthor{\bsnm{Nakatsukasa}, \binits{Y.}},
\bauthor{\bsnm{S{\`e}te}, \binits{O.}},
\bauthor{\bsnm{Trefethen}, \binits{L.N.}}:
\batitle{The AAA algorithm for rational approximation}.
\bjtitle{SIAM Journal on Scientific Computing}
\bvolume{40}(\bissue{3}),
\bfpage{1494}--\blpage{1522}
(\byear{2018})
\end{barticle}
\endbibitem

\bibitem[\protect\citeauthoryear{Obinata and Anderson}{2012}]{obinata2012model}
\begin{bbook}
\bauthor{\bsnm{Obinata}, \binits{G.}},
\bauthor{\bsnm{Anderson}, \binits{B.D.}}:
\bbtitle{Model Reduction for Control System Design}.
\bpublisher{Springer},
\blocation{London}
(\byear{2012})
\end{bbook}
\endbibitem

\bibitem[\protect\citeauthoryear{Opdenacker and
  Jonckheere}{2002}]{opdenacker2002contraction}
\begin{barticle}
\bauthor{\bsnm{Opdenacker}, \binits{P.C.}},
\bauthor{\bsnm{Jonckheere}, \binits{E.A.}}:
\batitle{A contraction mapping preserving balanced reduction scheme and its
  infinity norm error bounds}.
\bjtitle{IEEE Transactions on Circuits and Systems}
\bvolume{35}(\bissue{2}),
\bfpage{184}--\blpage{189}
(\byear{2002})
\end{barticle}
\endbibitem

\bibitem[\protect\citeauthoryear{Opmeer}{2011}]{opmeer2011model}
\begin{barticle}
\bauthor{\bsnm{Opmeer}, \binits{M.R.}}:
\batitle{Model order reduction by balanced proper orthogonal decomposition and
  by rational interpolation}.
\bjtitle{IEEE Transactions on Automatic Control}
\bvolume{57}(\bissue{2}),
\bfpage{472}--\blpage{477}
(\byear{2011})
\end{barticle}
\endbibitem

\bibitem[\protect\citeauthoryear{Phillips
  et~al.}{2002}]{phillips2002guaranteed}
\begin{bchapter}
\bauthor{\bsnm{Phillips}, \binits{J.}},
\bauthor{\bsnm{Daniel}, \binits{L.}},
\bauthor{\bsnm{Silveira}, \binits{L.M.}}:
\bctitle{Guaranteed passive balancing transformations for model order
  reduction}.
In: \bbtitle{Proceedings of the 39th Annual Design Automation Conference},
pp. \bfpage{52}--\blpage{57}
(\byear{2002})
\end{bchapter}
\endbibitem

\bibitem[\protect\citeauthoryear{Phillips
  et~al.}{2003}]{phillips2003guaranteed}
\begin{barticle}
\bauthor{\bsnm{Phillips}, \binits{J.R.}},
\bauthor{\bsnm{Daniel}, \binits{L.}},
\bauthor{\bsnm{Silveira}, \binits{L.M.}}:
\batitle{Guaranteed passive balancing transformations for model order
  reduction}.
\bjtitle{IEEE Transactions on Computer-Aided Design of Integrated Circuits and
  Systems}
\bvolume{22}(\bissue{8}),
\bfpage{1027}
(\byear{2003})
\end{barticle}
\endbibitem

\bibitem[\protect\citeauthoryear{Padhi et~al.}{2025}]{padhi2025data}
\begin{botherref}
\oauthor{\bsnm{Padhi}, \binits{R.}},
\oauthor{\bsnm{Gosea}, \binits{I.V.}},
\oauthor{\bsnm{Duff}, \binits{I.P.}},
\oauthor{\bsnm{Gugercin}, \binits{S.}}:
Data-driven balanced truncation for linear systems with quadratic outputs.
arXiv preprint arXiv:2509.12393
(2025)
\end{botherref}
\endbibitem

\bibitem[\protect\citeauthoryear{Pradovera}{2023}]{pradovera2023toward}
\begin{barticle}
\bauthor{\bsnm{Pradovera}, \binits{D.}}:
\batitle{Toward a certified greedy Loewner framework with minimal sampling}.
\bjtitle{Advances in Computational Mathematics}
\bvolume{49}(\bissue{6}),
\bfpage{92}
(\byear{2023})
\end{barticle}
\endbibitem

\bibitem[\protect\citeauthoryear{Pernebo and
  Silverman}{2003}]{pernebo2003model}
\begin{barticle}
\bauthor{\bsnm{Pernebo}, \binits{L.}},
\bauthor{\bsnm{Silverman}, \binits{L.}}:
\batitle{Model reduction via balanced state space representations}.
\bjtitle{IEEE Transactions on Automatic Control}
\bvolume{27}(\bissue{2}),
\bfpage{382}--\blpage{387}
(\byear{2003})
\end{barticle}
\endbibitem

\bibitem[\protect\citeauthoryear{Phillips and
  Silveira}{2004}]{phillips2004poor}
\begin{barticle}
\bauthor{\bsnm{Phillips}, \binits{J.R.}},
\bauthor{\bsnm{Silveira}, \binits{L.M.}}:
\batitle{Poor man's TBR: A simple model reduction scheme}.
\bjtitle{IEEE Transactions on Computer-aided Design of Integrated Circuits and
  Systems}
\bvolume{24}(\bissue{1}),
\bfpage{43}--\blpage{55}
(\byear{2004})
\end{barticle}
\endbibitem

\bibitem[\protect\citeauthoryear{Quarteroni
  et~al.}{2014}]{quarteroni2014reduced}
\begin{bbook}
\bauthor{\bsnm{Quarteroni}, \binits{A.}},
\bauthor{\bsnm{Rozza}, \binits{G.}}, \betal:
\bbtitle{Reduced Order Methods for Modeling and Computational Reduction}
vol. \bseriesno{9}.
\bpublisher{Springer},
\blocation{Berlin}
(\byear{2014})
\end{bbook}
\endbibitem

\bibitem[\protect\citeauthoryear{Reiter
  et~al.}{2025}]{reiter2023generalizations}
\begin{barticle}
\bauthor{\bsnm{Reiter}, \binits{S.}},
\bauthor{\bsnm{Gosea}, \binits{I.V.}},
\bauthor{\bsnm{Gugercin}, \binits{S.}}:
\batitle{Generalizations of data-driven balancing: What to sample for different
  balancing-based reduced models}.
\bjtitle{Automatica}
\bvolume{182},
\bfpage{112518}
(\byear{2025})
\end{barticle}
\endbibitem

\bibitem[\protect\citeauthoryear{Reis and Stykel}{2008}]{reis2008balanced}
\begin{barticle}
\bauthor{\bsnm{Reis}, \binits{T.}},
\bauthor{\bsnm{Stykel}, \binits{T.}}:
\batitle{Balanced truncation model reduction of second-order systems}.
\bjtitle{Mathematical and Computer Modelling of Dynamical Systems}
\bvolume{14}(\bissue{5}),
\bfpage{391}--\blpage{406}
(\byear{2008})
\end{barticle}
\endbibitem

\bibitem[\protect\citeauthoryear{Reiter and Werner}{2025}]{reiter2025data}
\begin{botherref}
\oauthor{\bsnm{Reiter}, \binits{S.}},
\oauthor{\bsnm{Werner}, \binits{S.W.}}:
Data-driven balanced truncation for second-order systems with generalized
  proportional damping.
arXiv preprint arXiv:2506.10118
(2025)
\end{botherref}
\endbibitem

\bibitem[\protect\citeauthoryear{Scarciotti and
  Astolfi}{2024}]{scarciotti2024interconnection}
\begin{barticle}
\bauthor{\bsnm{Scarciotti}, \binits{G.}},
\bauthor{\bsnm{Astolfi}, \binits{A.}}:
\batitle{Interconnection-based model order reduction-A survey}.
\bjtitle{European Journal of Control}
\bvolume{75},
\bfpage{100929}
(\byear{2024})
\end{barticle}
\endbibitem

\bibitem[\protect\citeauthoryear{Simoncini}{2016}]{simoncini2016computational}
\begin{barticle}
\bauthor{\bsnm{Simoncini}, \binits{V.}}:
\batitle{Computational methods for linear matrix equations}.
\bjtitle{SIAM Review}
\bvolume{58}(\bissue{3}),
\bfpage{377}--\blpage{441}
(\byear{2016})
\end{barticle}
\endbibitem

\bibitem[\protect\citeauthoryear{Schilders et~al.}{2008}]{schilders2008model}
\begin{bbook}
\bauthor{\bsnm{Schilders}, \binits{W.H.}},
\bauthor{\bsnm{Vorst}, \binits{H.A.}},
\bauthor{\bsnm{Rommes}, \binits{J.}}:
\bbtitle{Model Order Reduction: Theory, Research Aspects and Applications}
vol. \bseriesno{13}.
\bpublisher{Springer},
\blocation{Berlin}
(\byear{2008})
\end{bbook}
\endbibitem

\bibitem[\protect\citeauthoryear{Tombs and
  Postlethwaite}{1987}]{tombs1987truncated}
\begin{barticle}
\bauthor{\bsnm{Tombs}, \binits{M.S.}},
\bauthor{\bsnm{Postlethwaite}, \binits{I.}}:
\batitle{Truncated balanced realization of a stable non-minimal state-space
  system}.
\bjtitle{International Journal of Control}
\bvolume{46}(\bissue{4}),
\bfpage{1319}--\blpage{1330}
(\byear{1987})
\end{barticle}
\endbibitem

\bibitem[\protect\citeauthoryear{Varga}{2011}]{varga2011gervsgorin}
\begin{bbook}
\bauthor{\bsnm{Varga}, \binits{R.S.}}:
\bbtitle{Ger{\v{s}}gorin and His Circles}
vol. \bseriesno{36}.
\bpublisher{Springer}, \blocation{Berlin}
(\byear{2011})
\end{bbook}
\endbibitem

\bibitem[\protect\citeauthoryear{Wolf}{2014}]{wolf2014h}
\begin{botherref}
\oauthor{\bsnm{Wolf}, \binits{T.}}:
$\mathcal{H}_2$ pseudo-optimal model order reduction.
PhD thesis,
Technische Universit{\"a}t M{\"u}nchen
(2014)
\end{botherref}
\endbibitem

\bibitem[\protect\citeauthoryear{Willcox and
  Peraire}{2002}]{willcox2002balanced}
\begin{barticle}
\bauthor{\bsnm{Willcox}, \binits{K.}},
\bauthor{\bsnm{Peraire}, \binits{J.}}:
\batitle{Balanced model reduction via the proper orthogonal decomposition}.
\bjtitle{AIAA journal}
\bvolume{40}(\bissue{11}),
\bfpage{2323}--\blpage{2330}
(\byear{2002})
\end{barticle}
\endbibitem

\bibitem[\protect\citeauthoryear{Wolf and Panzer}{2016}]{wolf2016adi}
\begin{barticle}
\bauthor{\bsnm{Wolf}, \binits{T.}},
\bauthor{\bsnm{Panzer}, \binits{H.K.}}:
\batitle{The ADI iteration for Lyapunov equations implicitly performs $\mathcal{H}_2$
  pseudo-optimal model order reduction}.
\bjtitle{International Journal of Control}
\bvolume{89}(\bissue{3}),
\bfpage{481}--\blpage{493}
(\byear{2016})
\end{barticle}
\endbibitem

\bibitem[\protect\citeauthoryear{Wittmuess
  et~al.}{2016}]{wittmuess2016parametric}
\begin{barticle}
\bauthor{\bsnm{Wittmuess}, \binits{P.}},
\bauthor{\bsnm{Tarin}, \binits{C.}},
\bauthor{\bsnm{Keck}, \binits{A.}},
\bauthor{\bsnm{Arnold}, \binits{E.}},
\bauthor{\bsnm{Sawodny}, \binits{O.}}:
\batitle{Parametric model order reduction via balanced truncation with Taylor
  series representation}.
\bjtitle{IEEE Transactions on Automatic Control}
\bvolume{61}(\bissue{11}),
\bfpage{3438}--\blpage{3451}
(\byear{2016})
\end{barticle}
\endbibitem

\bibitem[\protect\citeauthoryear{Wang et~al.}{2025}]{wang2025data}
\begin{botherref}
\oauthor{\bsnm{Wang}, \binits{X.}},
\oauthor{\bsnm{Yang}, \binits{X.}},
\oauthor{\bsnm{Wang}, \binits{X.}},
\oauthor{\bsnm{Jiang}, \binits{Y.}}:
Data-driven balanced truncation for second-order systems via the approximate
  gramians.
Numerical Algorithms,
1--24
(2025)
\end{botherref}
\endbibitem

\bibitem[\protect\citeauthoryear{Zhou}{1995}]{zhou1995frequency}
\begin{barticle}
\bauthor{\bsnm{Zhou}, \binits{K.}}:
\batitle{Frequency-weighted $\mathcal{L}_\infty$ norm and optimal Hankel norm
  model reduction}.
\bjtitle{IEEE Transactions on Automatic Control}
\bvolume{40}(\bissue{10}),
\bfpage{1687}--\blpage{1699}
(\byear{1995})
\end{barticle}
\endbibitem

\bibitem[\protect\citeauthoryear{Zulfiqar}{2025a}]{zulfiqar2025non}
\begin{botherref}
\oauthor{\bsnm{Zulfiqar}, \binits{U.}}:
Compression and distillation of data quadruplets in non-intrusive reduced-order
  modeling.
arXiv preprint arXiv:2501.16683
(2025)
\end{botherref}
\endbibitem

\bibitem[\protect\citeauthoryear{Zulfiqar}{2025b}]{mycode}
\begin{botherref}
\oauthor{\bsnm{Zulfiqar}, \binits{U.}}:
{MATLAB} codes for Data-driven Implementations of Various Generalizations of
  Balanced Truncation.
https://doi.org/10.5281/zenodo.19155053
(2025).
\url{https://doi.org/10.5281/zenodo.19155053}
\end{botherref}
\endbibitem

\end{thebibliography}

\end{document}